\documentclass[11pt]{article}
\usepackage{hep-title} 
\usepackage{caption}
\usepackage[caption=false]{subfig}

\usepackage{amsmath,amsfonts,amsthm,amssymb,graphics,graphicx,mathtools,epstopdf,etoolbox,hyperref,indentfirst,doi,enumitem,mleftright,xcolor,booktabs,footnote,relsize,tabularx,units}

\hypersetup{
colorlinks = true,
citecolor= blue,
urlcolor= blue,
linkcolor = blue
}

\usepackage[
backend=bibtex,
style=alphabetic,
giveninits=true, 
maxbibnames=99,
minalphanames=3,
date=year
]{biblatex}
\renewbibmacro{in:}{}

\AtEveryBibitem{\clearfield{issue}} 

\newbibmacro{string+doiurl}[1]{%
	\iffieldundef{doi}{%
		\iffieldundef{url}{
			#1%
		}{%
			\href{\thefield{url}}{#1}%
		}%
	}{%
		\href{http://doi.org/\thefield{doi}}{#1}%
	}%
}
\DeclareFieldFormat{title}{\usebibmacro{string+doiurl}{\mkbibemph{#1}}}
\DeclareFieldFormat[article,thesis,incollection,inproceedings,manual]{title}%
{\usebibmacro{string+doiurl}
	{#1} 
}
\usepackage{algorithm}
\usepackage[noend]{algpseudocode}
\floatname{algorithm}{Protocol} 
\algrenewcommand\alglinenumber[1]{\normalsize #1.} 

\newcounter{algsubstate}




\newcommand{\ceil}[1]{\left\lceil #1 \right\rceil}
\newcommand{\floor}[1]{\left\lfloor #1 \right\rfloor}
\newcommand{\ket}[1]{\left| #1 \right>}
\newcommand{\bra}[1]{\left< #1 \right|}
\newcommand{\ketbra}[2]{\ket{#1} \! \bra{#2}}
\newcommand{\pure}[1]{\ketbra{#1}{#1}}
\newcommand{\tr}[2][]{\operatorname{Tr}_{#1}\!\left[#2\right]} 
\newcommand{\lambdaEC}{\lambda_{\mathrm{EC}}} 

\newcommand{\asin}{\sin^{-1}}

\newcommand{\CvsQ}{\kappa}
\newcommand{\defvar}{\coloneqq} 
\newcommand{\dop}[1]{\operatorname{S}_{#1}} 
\newcommand{\eps}{\varepsilon}

\newcommand{\fmin}{f}
\newcommand{\freq}{\operatorname{freq}}

\newcommand{\g}{g} 

\newcommand{\id}{\mathbb{I}} 
\newcommand{\idmap}{\mathcal{I}} 

\newcommand{\Max}{\operatorname{Max}}

\newcommand{\norm}[1]{\left\lVert#1\right\rVert} 

\newcommand{\pr}[2][]{\Pr_{#1}\!\left[#2\right]}
\newcommand{\Pos}{\operatorname{Pos}} 
\newcommand{\rate}{\operatorname{rate}}

\newcommand{\suchthat}{\text{ s.t.}} 
\newcommand{\supp}{\operatorname{supp}} 
\newcommand{\term}[1]{\emph{#1}}
\newcommand{\Var}{\operatorname{Var}}

\newcommand{\phon}{\mbf{p}^\mathrm{hon}}
\newcommand{\phonc}{p^\mathrm{hon}_c}
\newcommand{\qhon}{\mbf{q}^\mathrm{hon}}
\newcommand{\qhonc}{q^\mathrm{hon}_c}
\newcommand{\svar}{\tau} 

\newcommand{\tupp}{\mbf{t}^\mathrm{upp}}
\newcommand{\tlow}{\mbf{t}^\mathrm{low}}
\newcommand{\tuppc}{t_c^\mathrm{upp}}
\newcommand{\tlowc}{t_c^\mathrm{low}}

\newcommand{\rhotmu}{\rho^{t,\mu}_{J}}
\newcommand{\rhotmun}[1][n]{\rho^{t,\mu}_{J_{#1}}}
\newcommand{\Nph}{N_{\mathrm{ph}}}
\newcommand{\musig}{\mu_{\mathrm{sig}}}
\newcommand{\ptot}{p_{\mathrm{tot}}}
\newcommand{\deltavec}{\boldsymbol{\delta}}

\newcommand{\mbf}[1]{\mathbf{#1}} 
\newcommand{\Renyi}{R\'{e}nyi}

\newcommand{\esecure}{\eps^\mathrm{secure}}
\newcommand{\ecorr}{\eps^\mathrm{correct}}
\newcommand{\esecret}{\eps^\mathrm{secret}}
\newcommand{\eEV}{\eps_\mathrm{EV}}
\newcommand{\ePA}{\eps_\mathrm{PA}}
\newcommand{\ecom}{\eps^\mathrm{com}}
\newcommand{\ecomEV}{\eps^\mathrm{com}_\mathrm{EV}}
\newcommand{\ecomAT}{\eps^\mathrm{com}_\mathrm{AT}}

\newtheorem{theorem}{Theorem}
\newtheorem{lemma}{Lemma}

\newtheorem{remark}{Remark}
\theoremstyle{definition} 
\newtheorem{definition}{Definition}
\newtheorem{prot}{Protocol}

\begin{document}

\title{\textbf{\textcolor{black}{Finite-size analysis of prepare-and-measure and decoy-state QKD via entropy accumulation}}}

\renewcommand\Affilfont{\itshape\small} 
\renewcommand{\theaffiliation}{%
\textsuperscript{\normalfont\arabic{affiliation}}%
}

\newcommand\CoAuthorMark{\footnotemark[\arabic{footnote}]} 

\author[1,3]{Lars Kamin\footnote{Both authors contributed equally to this work.}}
\email[3]{lars.kamin@outlook.com}

\author[1,4]{Amir Arqand \protect\CoAuthorMark}
\email[4]{aarqand@uwaterloo.ca}

\author[2,5]{Ian George}
\email[5]{igeorge3@illinois.edu}

\author[1,6]{Norbert L\"{u}tkenhaus}
\email[6]{lutkenhaus.office@uwaterloo.ca}

\author[1,7]{Ernest Y.-Z.\ Tan}
\email[7]{yzetan@uwaterloo.ca}

\affiliation[1]{Institute for Quantum Computing and Department of Physics and Astronomy, University of Waterloo, Waterloo, Ontario N2L 3G1, Canada}
\affiliation[2]{Department of Electrical and Computer Engineering,  Coordinated Science Laboratory, University of Illinois at Urbana-Champaign, Urbana, Illinois 61801, USA}

\date{}

\maketitle

\begin{abstract}
An important goal in quantum key distribution (QKD) is the task of providing a finite-size security proof without the assumption of collective attacks. For prepare-and-measure QKD, one approach for obtaining such proofs is the generalized entropy accumulation theorem (GEAT), but thus far it has only been applied to study a small selection of protocols. In this work, we present techniques for applying the GEAT in finite-size analysis of generic prepare-and-measure protocols, with a focus on decoy-state protocols. In particular, we present an improved approach for computing entropy bounds for decoy-state protocols, which has the dual benefits of providing tighter bounds than previous approaches (even asymptotically) and being compatible with methods for computing min-tradeoff functions in the GEAT. Furthermore, we develop methods to incorporate some improvements to the finite-size terms in the GEAT, and implement techniques to automatically optimize the min-tradeoff function. Our approach also addresses some numerical stability challenges specific to prepare-and-measure protocols, which were not addressed in previous works.
\end{abstract}

\section{Introduction}\label{sec:Introduction}
Quantum key distribution (QKD) is the task of establishing a secret shared key between two parties (Alice and Bob) in the presence of an adversary (Eve)~\cite{BB84,Eke91}. In a QKD task, Alice and Bob are connected by an insecure quantum channel and an authenticated classical channel. 
Eve is allowed to intercept any states sent over the quantum channel and perform any valid attack within the realm of quantum mechanics, though she cannot modify or impersonate messages sent over the authenticated classical channel. The goal of a security proof for a QKD protocol is to show that Alice and Bob can produce a secure key under these conditions, regardless of the attack Eve employs.

One of the major difficulties in constructing such security proofs is in handling the full scope of attacks available to Eve. In particular, Eve could vary her attack across different rounds of the protocol: for instance, by using classical information or quantum states she gathered from previous rounds. A security proof that takes into account such an attack is called a security proof against coherent attacks, as opposed to a proof that assumes Eve attacks in an independent and identically distributed (IID) manner across the rounds, which would be called a security proof against IID collective attacks. 
Some techniques for proving security against coherent attacks include 
complementarity-based techniques~\cite{Koa09,TR11,LCW+14},
de Finetti theorems~\cite{rennerthesis}, the postselection technique~\cite{CKR09,Nahar2024PRXQuantum}, and entropy accumulation theorems (EAT)~\cite{DFR20,DF19,Metger2022Mar,MR23}. The original versions of the EAT~\cite{DFR20,DF19} were developed for entanglement-based (EB) protocols, and could not be straightforwardly applied to prepare-and-measure (PM) protocols. This issue was addressed with the subsequent development of a \emph{generalized} entropy accumulation theorem (GEAT) in~\cite{Metger2022Mar,MR23} that yields security proofs for PM protocols.

Typically, these techniques show that in the asymptotic limit of infinitely many protocol rounds, the key rates (i.e.~the length of the final key divided by the number of rounds) for most protocols are the same against both coherent attacks and IID collective attacks. However, in realistic scenarios where the number of rounds is finite, these techniques may yield significantly different bounds on the finite-size key rate; furthermore, the conditions required to apply each technique are somewhat different.
In particular, complementarity-based proofs such as those in~\cite{LCW+14} require a detailed analysis of how to reduce the security proof to a model where complementary measurements are being performed on qubit systems. Such reductions can be rather protocol-specific, and involve some conditions on the detectors, such as requiring active basis choices and assuming that the detector efficiency is independent of basis.\footnote{After initial preparation of this manuscript, progress was made in weakening these assumptions~\cite{arx_TNS+24}; however, we leave a detailed comparison for future work.} 
In contrast, entropy accumulation is not subject to these restrictions, and is hence more generally applicable to a broad range of protocols.
As for de Finetti theorems and the postselection technique, they require a permutation-symmetry condition that usually involves explicitly implementing a random permutation at the start of the protocol, which introduces an additional complication in practical implementations. Furthermore, they generally give worse bounds on the finite-size key rate as compared to entropy accumulation~\cite{arx_GLvH+22}. 

Hence in principle, the GEAT provides a flexible security proof framework that applies to generic PM-QKD protocols, while also providing better key rates than de Finetti theorems or the postselection technique, following the observations in~\cite{arx_GLvH+22}.
However, previous works based on the GEAT mainly provided proof-of-principle demonstrations of such security proofs, without much focus on optimizing the parameter choices to maximize the finite-size key rates in practical applications --- in particular, the GEAT relies on a construction known as a \emph{min-tradeoff function}, which is challenging to construct optimally for protocols with large numbers of signal states or measurement outcomes.
This has been a critical limitation in using it to obtain security proofs for QKD protocols that are closer to practical realization. Our goal in this work is to overcome this obstacle, while also implementing several techniques to optimize the finite-size key rate bounds (based on ideas introduced in~\cite{arx_GLvH+22} for EB protocols).

Furthermore, when considering PM protocols of practical interest, an important family of such protocols would be decoy-state protocols. More specifically, when dealing with practical QKD implementations, the source usually employs a highly attenuated laser that does not necessarily output a single-photon pulse. There is a non-trivial probability of these pulses containing more than one photon. This would potentially introduce an insecurity into the protocol since with these pulses, Eve can now perform a PNS (photon-number splitting) attack \cite{BBB+92,BLM+00}. To limit the effects of this attack, decoy-state protocols were introduced \cite{H02,MQZ+05,W05}. In this method, Alice chooses randomly from a set of possible intensities for the pulses she sends to Bob. Since Eve cannot directly distinguish the signals by their intensities, Alice and Bob can design an acceptance test that aims to detect attempted PNS attacks. Thus, the decoy-state method makes the protocol more resistant to such attacks.

In this work, we provide a flexible framework for the finite-size analysis of PM protocols using the GEAT, including an analysis of decoy-state protocols. Our main contributions at the theoretical level can be summarized as follows:
\begin{itemize}
\item We introduce a new method for computing reliable bounds on the entropy produced in a single round of a generic decoy-state protocol. For the purposes of the GEAT, our approach gives a flexible method to construct a min-tradeoff function for such protocols, unlike previous constructions that relied on closed-form entropy bounds. It also offers some benefits in terms of improved asymptotic key rates, which we describe below.
\item For the GEAT analysis specifically, we apply techniques similar to those in \cite{arx_GLvH+22}, but with improvements and modifications to address challenges specific to PM protocols. In particular, we study a different version of the GEAT~\cite{LLR+21} that has improved finite-size terms defined by a separate convex optimization, and show that this optimization can be unified with the computation of the single-round entropies, hence simplifying the application of the form of the GEAT from~\cite{LLR+21}. We also address a technical challenge in applying the~\cite{arx_GLvH+22} approach to analyze test rounds for PM protocols by introducing a parametrization based on Choi states, which may also be useful for security proofs outside the entropy accumulation framework. 
\end{itemize}

In terms of the final implications on the key rates, the overall benefits of our approach can be summarized as follows. We view these as basically separate effects; however, we do not attempt to isolate the extent of their effects on the key rates.
\begin{itemize}
\item Regarding the key rates in the asymptotic limit, or against IID collective attacks, our approach for analyzing single rounds of a decoy-state protocol would yield better bounds than the techniques  in~\cite{WL22,NUL23,KL24,arx_KTL25}, because it bypasses some suboptimal intermediate steps in those works --- we formalize this claim rigorously in Remark~\ref{remark:twostep}, after introducing the necessary concepts.
\item Regarding the final finite-size key rates against coherent attacks, our approach would yield better results than the entropy-accumulation-based approach in~\cite{arx_GLvH+22}, as we used the strictly better finite-size bounds developed in~\cite{LLR+21}. Furthermore, it was already demonstrated in~\cite{arx_GLvH+22} that entropy-accumulation-based approaches provide better finite-size key rates than the postselection technique for many protocols, hence we do not repeat such a comparison here. (For a full analysis of applying the latter to decoy-state protocols, refer to~\cite{arx_KTL25}.)
\end{itemize}
We emphasize that the scope of our above claims is restricted to the specifically mentioned works. In particular, these claims do not encompass complementarity-based proof techniques, which use a rather different proof structure, making a fully rigorous mathematical comparison difficult. While those techniques are not the focus of this work, in Fig.~\ref{fig:Decoy BB84 unique vs realistic acceptance} we show some comparisons to the finite-size key rates we empirically obtained from such methods, which heuristically suggest that each of the techniques performs better in different regimes. Still, we note that as mentioned above, the GEAT framework covers a wider range of protocols (for instance, those based on passive basis choices) as compared to the existing literature on complementarity-based proofs. 

To briefly elaborate on the above points, previous approaches in~\cite{WL22,NUL23,KL24,arx_KTL25} for key rate evaluations in decoy-state protocols were based on a two-step procedure, which first computes bounds on various photon yields and then bounds the entropy based on the single-photon yields using the method in~\cite{WLC18}. However, since this takes place in two separate steps, it is not straightforwardly compatible with existing methods for numerically constructing min-tradeoff functions for the GEAT, which relied on the existence of a single convex optimization problem for bounding the entropies. We introduce a new approach which unifies the two steps into a single convex optimization. This lets us construct min-tradeoff functions for generic decoy-state protocols. It has the further advantage that it would yield better bounds than the methods in~\cite{WL22,NUL23,KL24,arx_KTL25} used to analyze the asymptotic or finite-size IID scenarios (see Remark~\ref{remark:twostep} for the details).

As for our improvements specific to the GEAT-based security proof framework (in particular, over the entropy accumulation analysis in \cite{arx_GLvH+22}), our first improvement is to incorporate a different version of the GEAT with improved finite-size terms~\cite{LLR+21}.
The drawback of this version is that its finite-size terms involve yet another optimization problem, making key rate computations more time-consuming. We address this difficulty by showing that this optimization can be merged with the numerical methods for bounding single-round entropies, hence significantly streamlining the calculations. Similar to~\cite{arx_GLvH+22}, we also derive a method to automatically optimize the choice of min-tradeoff function for this version of the GEAT, by showing that the optimal min-tradeoff function is the Lagrange dual solution to a modified version of an entropy minimization problem derived in \cite{WLC18}. 
However, an additional difficulty here for PM protocols is that in this method, it is more numerically stable to analyze the distribution only on the test rounds, but the formulation developed for this purpose in~\cite{arx_GLvH+22} is not immediately compatible with the standard source-replacement technique~\cite{BBM92,FL12}. To overcome this difficulty, we introduce a parametrization via Choi states --- we highlight that this approach may also improve stability of previous numerical techniques for analyzing prepare-and-measure protocols. 

We highlight that in the~\cite{MR23} analysis of PM protocols using the GEAT, there was a technical condition that Eve only interacts with a single signal at a time (although she can do so in an arbitrary non-IID manner), which limits the repetition rate of the protocol. However, it has been shown in separate works~\cite{arx_FKR+25,arx_AT25} that the key rates computed in this manner are in fact also valid even without that condition (this general proof approach was proposed in an earlier work~\cite{inprep_HB25} presented at various conferences, but the proofs in~\cite{arx_FKR+25,arx_AT25} are independent of that work). Therefore, that condition does not in fact need to be imposed when applying the results we have obtained here --- see Sec.~\ref{sec:conclusion} for further details. (An alternative approach would be to use a ``virtual tomography'' analysis developed in~\cite{Bauml2024Quantum}, though this slightly reduces the key rates due to finite-size limitations of the virtual tomography step.)

The rest of the paper is organized as follows. In Sec.~\ref{sec:Notation}, we lay out our overall notation and definitions. In Sec.~\ref{sec:prot_description} we describe the general structure of PM protocols that can be accommodated in our framework. In Sec.~\ref{sec:EATbounds}, we derive the finite key length formula for such protocols, based on the GEAT. Then in Sec.~\ref{sec:finitesize}, we present the methods we derived for efficiently computing the terms in that formula, including our methods to optimize the min-tradeoff function as well as various other parameter choices. In Sec.~\ref{sec:Qubit BB84 with loss} we present the resulting key rates we obtained for a qubit BB84 protocol with loss. In Sec.~\ref{sec:Decoy-state with improved analysis} we present our improved method for analyzing decoy-state protocols, as well as the key rates we obtained for a weak coherent pulse decoy-state BB84 protocol. Finally, in Sec.~\ref{sec:conclusion} we present some concluding discussions and remarks.

\section{Notation and definitions}\label{sec:Notation}

\begin{table}[h!]
\caption{List of notation}\label{tab:notation}
\def\arraystretch{1.5} 
\setlength\tabcolsep{.28cm}
\begin{tabular}{c l}
\toprule
\textit{Symbol} & \textit{Definition} \\
\toprule
$\log$ & Base-$2$ logarithm \\
\hline
$H$ & Base-$2$ von Neumann entropy \\
\hline
$A\perp B$ & $A$ and $B$ are orthogonal; $AB=BA=0$ \\
\hline
$\floor{\cdot}$ (resp.~$\ceil{\cdot}$) & Floor (resp.~ceiling) function \\
\hline
$\norm{\cdot}_p$ & Schatten $p$-norm \\
\hline
$\Pos(A)$ & Set of positive semi-definite operators on register $A$\\
\hline
$\dop{=}(A)$ (resp.~$\dop{\leq}(A)$) & Set of normalized (resp.~subnormalized) states on register $A$ \\
\hline
$A^n$ & Abbreviated notation for registers $A_1 \dots A_n$ \\
\toprule
\end{tabular}
\def\arraystretch{1}
\end{table}

In this section, we set our overall notation and define some concepts required in our work.

\begin{definition}\label{def:freq}
(Frequency distributions) For a string $z_1^n\in\mathcal{Z}^n$ on some alphabet $\mathcal{Z}$, $\freq_{z_1^n}$ denotes the following probability distribution on $\mathcal{Z}$:
\begin{align}
\freq_{z_1^n}(z) \defvar \frac{\text{number of occurrences of $z$ in $z_1^n$}}{n} 
.
\end{align}
\end{definition}

\begin{definition}
A state $\rho \in \dop{\leq}(CQ)$ is said to be \term{classical on $C$} (with respect to a specified basis on $C$) if it is in the form 
\begin{align}
\rho_{CQ} = \sum_c \lambda_c \pure{c} \otimes \sigma_c,
\label{eq:cq}
\end{align}
for some normalized states $\sigma_c \in \dop{=}(Q) $ and weights $\lambda_c \geq 0$, with $\ket{c}$ being the specified basis states on $C$. In most circumstances, we will not explicitly specify this ``classical basis'' of $C$, leaving it to be implicitly defined by context.
It may be convenient to absorb the weights $\lambda_c$ into the states $\sigma_c$, writing them as subnormalized states $\omega_c = \lambda_c\sigma_c \in \dop{\leq}(Q)$ instead. 
\end{definition}

\begin{definition}\label{def:cond}
(Conditioning on classical events) For a state $\rho \in \dop{=}(CQ)$ classical on $C$, written in the form
$\rho_{CQ} = \sum_c \pure{c} \otimes \omega_c$ 
for some $\omega_c \in \dop{\leq}(Q)$,
and an event $\Omega$ defined on the register $C$, we will define the corresponding \term{conditional state} as
\begin{align}
\rho_{|\Omega} \defvar \frac{
\sum_{c\in\Omega} \pure{c} \otimes \omega_c
}{\sum_{c\in\Omega} \tr{\omega_c}} .
\end{align}
In light of the above definition, we may sometimes write $\rho_{CQ}$ with the slightly abbreviated notation
\begin{align}
\rho_{CQ} = \sum_c \pr{c} \pure{c} \otimes \rho_{Q|c},
\label{eq:cqnorm}
\end{align}
where $\pr{c}$ describes the probability distribution on $C$, and $\rho_{Q|c}$ denotes the conditional state on $Q$ corresponding to the register $C$ taking value $c$.
\end{definition}

Our security proofs will involve sandwiched {\Renyi} entropies~\cite{Muller_Lennert_2013,Wilde_2014,Tom16}. Here we briefly summarize their definitions, based on sandwiched {\Renyi} divergences.
\begin{definition}[Sandwiched {\Renyi} divergence]
\label{def:sand_Renyi_div}
For any $\rho\in\dop{=}(A)$, $\sigma\in\Pos(A)$, with $\mathrm{supp}(\rho)\subseteq \mathrm{supp}(\sigma)$, and $\alpha\in(0,1)\cap (1,\infty)$, the $\alpha$-sandwiched {\Renyi} divergence between $\rho$, $\sigma$ is defined as:
\begin{align}
    \label{eq:sand_Renyi_div}
    D_\alpha(\rho||\sigma)=\begin{cases}
    \frac{1}{\alpha-1}\log\tr{ \left(\sigma^{\frac{1-\alpha}{2\alpha}}\rho\sigma^{\frac{1-\alpha}{2\alpha}}\right)^\alpha} &\left(\alpha < 1\ \land\ \rho\not\perp\sigma\right)\vee \left(\supp(\rho)\subseteq\supp(\sigma)\right) \\ 
    +\infty & \text{otherwise},
    \end{cases}  
\end{align}
where $\rho\not\perp\sigma$ stands for the case where the two operators are not orthogonal. For $\alpha>1$ the $\sigma^{\frac{1-\alpha}{2\alpha}}$ terms are defined via the Moore-Penrose pseudoinverse if $\sigma$ is not full-support.
\end{definition}

\begin{definition}
\label{def:cond_Renyi}
For any bipartite state $\rho_{AB}\in\dop{=}(AB)$, and $\alpha\in(0,1)\cup (1,\infty)$, we define the following sandwiched {\Renyi} entropies of $A$ conditioned on $B$:
\begin{align}
    \label{eq:cond_Renyi}
    &H_\alpha(A|B)_\rho=-D_\alpha(\rho_{AB}||\id_A\otimes\rho_B)\notag\\
    &H_\alpha^\uparrow(A|B)_\rho=\sup_{\sigma_B\in\dop{=}(B)}-D_\alpha(\rho_{AB}||\id_A\otimes\sigma_B).
\end{align}
\end{definition}

\section{Protocol description}
\label{sec:prot_description}

For this work, we focus on protocols that make a binary accept/abort decision and output keys of fixed length when they accept. A recent work~\cite{TTL24} has developed security proof techniques (in the entropic framework rather than phase error analysis) for protocols that instead produce keys of variable length depending on the statistics observed in the protocol, though restricted to the case of collective attacks. An interesting question is whether the proof techniques used here can be combined with that work, though we leave this for future research. 

\begin{prot}{Prepare-and-Measure Protocol.}\label{Prot:PM Protocol} \\
		\textbf{Parameters:} \\
	\begin{tabularx}{0.9\linewidth}{r c X}
			\(n \in \mathbb{N}_0\) 			    &:& 	Total number of rounds \\
			\(l \in \mathbb{N}_0\)						&:&		Length of final key\\
			\(\{\sigma_i\}_{i=1 \dots d_X}\) 	&:& 	States sent by Alice \\
            \(\mathcal{S}\)				&:&		alphabet of raw-key registers, usually \(\{0,1\}\) \\
			\(M^{B}_{k}\) 						&:& 	POVM elements acting on Bob's system describing his measurements outcomes \(k\)\\
           \(\gamma\)						&:& 	Probability that Alice chooses a round to be a test round \\
			\(S_{\mathrm{acc}}\)						&:& 	Acceptance set of accepted frequencies \\
           \(\lambda_{\mathrm{EC}}\)			&:& 	Number of bits exchanged during error correction step \\
			\(\ePA\)			&:&		Security parameter contribution from privacy amplification \\
			\(\eEV\)			        &:&		Security parameter contribution from error verification \\
			\(\Omega_{\text{AT}}\)				&:&		Event of passing the acceptance test \\
			\(\Omega_{\text{EV}}\)				&:&		Event of passing the error verification \\
			\(\Omega_{\text{acc}} = \Omega_{\text{AT}} \land \Omega_{\text{EV}}\)				&:&		Event of the protocol not aborting \\
	\end{tabularx}
	\vspace{20pt}
	
	\textbf{Protocol steps:}
    \begin{enumerate}
    \item For each round $i \in \{1,2,\dots,n\}$, Alice and Bob perform the following steps:
	\begin{enumerate}
		\item \textbf{State preparation and transmission:} Independently for each round, with probability $\gamma$ Alice chooses it to be a \term{test round}, and otherwise she chooses it to be a \term{generation round}. Then, Alice prepares one of \(d_X\) states \(\{\sigma_k\}_{k=1\dots d_X}\), according to some distribution which could depend on the choice of a test or generation round. She stores the label for her choice of signal state in a classical register \(X_i\), and computes a classical register \(C^A_i\) for public announcement (including for instance the test/generation decision). Finally, Alice sends the signal state to Bob via a quantum channel.
  
		\item \textbf{Measurements:} Bob measures his received states using a POVM with POVM elements \(\{M_k^B\}_{k=1 \dots d_B}\), and stores his results in a classical register \(Y_i\) with alphabet \(\mathcal{Y}\). Furthermore, he computes a classical register \(C^B_i\) for public announcement.
  
        \item \label{step:announce} \textbf{Public announcement:} Alice and Bob announce their values $C^A_i C^B_i$, and jointly compute a value $C_i$ via public classical discussion (which can be two-way). $C_i$ is a value that will be used later in the acceptance test to decide whether to abort, and we use the convention that it is set to a fixed symbol $\perp$ whenever Alice chose the round to be a generation round (informally, this corresponds to the fact that the acceptance test later will only depend on data from the test rounds).
        Let $I_i$ denote all the registers publicly communicated in this step, including a copy of $C_i$.

		\item \label{step:sift} \textbf{Sifting and key map:} 
        Alice applies a sifting\footnote{For the purposes of this work, when a round is sifted out, we take this to mean Alice sets $S_i$ to a fixed value in the alphabet of $S_i$ (say, $0$). However, it should be possible to modify this to a version where such rounds are actually discarded (i.e.~not included in the privacy amplification step) by using the techniques in~\cite{TTL24}.} and/or key map procedure based on her raw data $X_i$ and the public announcements $I_i$, to produce a classical register $S_i$.
	\end{enumerate}

		\item \label{step:AT} \textbf{Acceptance test (parameter estimation):} Alice and Bob compute the frequency distribution $\mbf{F}^{\text{obs}}$ of the string $C_1^n$. They accept if \(\mbf{F}^{\text{obs}} \in S_{\mathrm{acc}} \) where \(S_{\mathrm{acc}}\) is the predetermined acceptance set, and abort the protocol otherwise.
		\item \label{step:ECandEV} \textbf{Error correction and verification:} Alice and Bob publicly communicate  $\lambda_{\mathrm{EC}}$ bits for error correction, such that Bob can use those bits together with his data $Y_1^n I_1^n$ to produce a guess $\widehat{S}_1^n$ for Alice's string $S_1^n$. This is followed by an error-verification step, where Alice sends a 2-universal hash of $S_1^n$ with length $\ceil{\log(1/\eEV)}$ to Bob, who compares it to the hash of his guess and accepts if the hashes match (and otherwise aborts).
		\item \textbf{Privacy amplification:} Alice and Bob randomly choose a 2-universal hash function from some family (with fixed output length $l$) and apply it to their strings $S_1^n$ and $\widehat{S}_1^n$, producing final keys $K_A$ and $K_B$.
	\end{enumerate}
\end{prot}	

While we have described steps~\ref{step:announce}--\ref{step:sift} as taking place in individual rounds, in practical implementations it is fine for them to take place after all states have been sent and measured, as long as $C_i I_i S_i$ are all indeed only computed from the data in their respective rounds. This is because even if such operations physically take place after all measurements are completed, we can commute them with the measurements in preceding rounds to describe the protocol using a model of the form presented above. 

We also highlight that what we have chosen to label as a test or generation round is determined \emph{entirely} by Alice, and we will use this convention throughout our subsequent proofs. It is true that in each round Bob makes a separate decision of whether to measure in a ``test basis'' or ``generation basis'', and to prepare the registers $S_i$ for final key generation, we would have Bob announce his basis choice, so the parties can sift out rounds where they did not both measure in the generation basis. However, for this work we use the convention that this is described in the sifting and/or key map, rather than the label of whether we call a round a test or generation round. We give details of how to formalize this when considering specific protocols in Sec.~\ref{sec:Qubit BB84 with loss}--\ref{sec:Decoy-state with improved analysis}.

\section{Improved entropy accumulation analysis for PM protocols}
\label{sec:EATbounds}
In this section, we employ a modified version of the generalized entropy accumulation theorem in~\cite{Metger2022Mar}, along with the improvement in~\cite{LLR+21}, to bound the secure key length of a PM QKD protocol. We start by defining the notion of a GEAT channel.
\begin{definition}
    \label{def:eat_channel}
    A sequence of channels $\{\mathcal{M}_i:R_{i-1}E_{i-1}\rightarrow C_iA_iR_iE_i\}_{i\in \{1,2,\dots,n\}}$ is a sequence of GEAT channels if for all $i\in\{1,2,\dots,n\}$, they satisfy the following properties:
    \begin{itemize}
        \item $C_i$ are classical registers with common alphabet $\mathcal{C}$.
        \item $\exists\mathcal{P}_{A^nE_n\rightarrow C^nA^nE_n}$ such that $\mathcal{M}_n\circ\cdots\circ\mathcal{M}_1=\mathcal{P}\circ\mathcal{N}_n\circ\dots\circ\mathcal{N}_1$ with $\mathcal{N}_i=\operatorname{Tr}_{C_i}\circ\mathcal{M}_i$, and
        \begin{equation}
            \label{eq:Reconstruction map}
            \mathcal{P}(\rho_{A^nE_n})=\sum_{s\in\mathcal{S},t\in\mathcal{T}}\left(\Pi_{A^n}^{(s)}\otimes\Pi_{E_n}^{(t)}\right)\rho_{A^nE_n}\left(\Pi_{A^n}^{(s)}\otimes\Pi_{E_n}^{(t)}\right)\otimes\ket{r(s,t)}\bra{r(s,t)}_{C^n},
        \end{equation}
        where $\{\Pi_{A^n}^{(s)}\}_{s\in\mathcal{S}}$ and$\{\Pi_{E_n}^{(t)}\}_{t\in\mathcal{T}}$ are families of orthogonal projectors on $A^n$ and $E_n$ respectively, while $r:\mathcal{S}\times\mathcal{T}\rightarrow\mathcal{C}$ is a deterministic function. In other words, the statistics can be generated solely from the $A^n$ and $E_n$ systems.
        \item $\exists\mathcal{R}_i:E_{i-1}\rightarrow E_i$ such that $\operatorname{Tr}_{A_iR_iC_i}\circ\mathcal{M}_i=\mathcal{R}_i\circ\operatorname{Tr}_{R_{i-1}}$. Intuitively, this condition states that the channels are ``non-signaling'' from the $R_{i-1}$ registers to the $E_i$ registers.
        
    \end{itemize}
\end{definition}
We now define the notion of rate function and min-tradeoff function:
\begin{definition}
    A real function $\rate(\mbf{p})$ on $\mathcal{C}$ is called a \term{rate function} for a set of GEAT channels $\{\mathcal{M}_i\}_{i\in\{1,2,\dots,n\}}$, if it satisfies:
    \begin{equation}\label{eq:ratefuncdef}
    \rate(\mbf{p})\leq\inf_{\nu\in\Sigma_{i}(\mbf{p})}H(A_i|E_iR)_\nu\qquad\qquad\forall i\in\{1,2,\dots,n\}
    \end{equation}
where $\Sigma_i(\mbf{p})$ is the set of output states of the extended GEAT channel $\mathcal{M}_i\otimes\id_R$, such that the reduced state on the classical register $C_i$ has the same distribution as $\mbf{p}$. Note that if the set $\Sigma_i(\mbf{p})$ is empty we define the infimum to be $+\infty$ (following standard conventions in optimization theory).
\end{definition}
\begin{definition}
A \term{min-tradeoff function} is an affine rate function. 
Being an affine function, it can always be expressed in the form
\begin{align}\label{eq:affinef}
f(\mbf{p}) = \mbf{f}\cdot\mbf{p} + k_f,
\end{align}
for some vector $\mbf{f}$ (which is just the gradient of $f$) and scalar
$k_f$.\footnote{In~\cite{arx_GLvH+22}, the scalar $k_f$ was implicitly absorbed into the gradient $\mbf{f}$ by exploiting the fact that the min-tradeoff function is always only evaluated on normalized distributions; however, in this work, it will be more convenient to keep the terms separate.} 
\end{definition}
We now state the GEAT, with an improved second-order term from~\cite{LLR+21}:
\begin{theorem}[\cite{Metger2024} Theorem~4.3
and~\cite{LLR+21} Supplement, Theorem~2]
\label{thrm:EAT}
For a sequence of GEAT channels $\{\mathcal{M}_i\}_{i\in\{1,2,\dots,n\}}$ defined in Eq.~(\ref{def:eat_channel}), let $\eps\in (0,1),\ \alpha\in(1,3/2)$, let 
$\Omega$ be an event on registers $C_1^n$,
let $\rho_{R_0E_0}\in\mathcal{D}(R_0E_0)$, let $f$ be a min-tradeoff function, and let $h=\min_{c^n\in\Omega} f(\mathrm{freq}(c^n))$. Then: 
\begin{align}\label{eq:EATbound}
    H^{\uparrow}_\alpha(A^n|E_n)&_{\mathcal{M}_n\circ\cdots\circ\mathcal{M}_1(\rho_{R_0E_0})_{|_{\Omega}}}
    \ge nh+nT_\alpha(f)-\frac{\alpha}{\alpha-1}\log \frac{1}{\pr
    {\Omega}}-n\left(\frac{\alpha-1}{2-\alpha}\right)^2 K(\alpha),
\end{align}
where (letting $\CvsQ=1$ if the $A_i$ systems are classical and $\CvsQ=2$ if they are quantum):
\begin{align}\label{eq:2ndorderfuncs}
    &T_\alpha(f)=\inf_{\mbf{p}\in\mathcal{Q}}\left(\rate(\mbf{p})-f(\mbf{p})-\frac{\alpha-1}{2-\alpha}\frac{\operatorname{ln}(2)}{2}V(\mbf{p},f)\right)\notag\\
    &\Var(\mbf{p},f)=\sum_xp(x)f(\delta_x)^2-\left(\sum_xp(x)f(\delta_x)\right)^2\notag\\
    &V(\mbf{p},f)=\left(\log(1+2d_A^{\CvsQ})+\sqrt{2+\Var(\mbf{p},f)}\right)^2\notag\\
    &K(\alpha)=\frac{(2-\alpha)^3}{6(3-2\alpha)^3\operatorname{ln}2}2^{\frac{\alpha-1}{2-\alpha}(\CvsQ\log d_A+\max(f)-\min_\mathcal{Q}(f))}\ln^3\left(2^{\CvsQ\log d_A+\max(f)-\min_\mathcal{Q}(f)}+e^2\right),
\end{align}
with $d_A=\max_i{\dim{(A_i)}}$, $\delta_x$ being the distribution with all weight on element x, and $\mathcal{Q}$ being the set of all distributions on $\mathcal{C}$ that could be produced by a GEAT channel $\mathcal{M}_i$. 
\end{theorem}

Before formally defining the notion of security in QKD protocols, we state a result regarding privacy amplification with {\Renyi} entropies, which was proven in \cite{Dupuis2023IEEETrans.Inf.Theory}.
\begin{theorem}[\cite{Dupuis2023IEEETrans.Inf.Theory} Theorem~9]
\label{thrm:privacy_amp}
Let $\rho_{AE}$ be a classical-quantum state, and $(\mathcal{F}_{\mathcal{A}\rightarrow\mathcal{Z}},p_f)$ be a family of two-universal hash functions with $\mathcal{Z}=\{0,1\}^l$. Then if $f$ is a function drawn uniformly from that family, for $\alpha\in(1,2)$ we have
\begin{align}
    \label{eq:privacy_amp}
    \mathbb{E}_f\left\|\rho^f_{ZE}-\frac{1}{|\mathcal{Z}|}\id_z\otimes\rho_{E} \right\|_1=\left\|\rho_{ZEF}-\frac{1}{|\mathcal{Z}|}\id_z\otimes\rho_{EF} \right\|_1 \leq 2^{\frac{2}{\alpha}-1}2^{\frac{1-\alpha}{\alpha}\big(H_\alpha^\uparrow(A|E)_\rho-l \big)},
\end{align}
where $F$ is the register that stores the choice of hash function.
\end{theorem}
We now proceed to present the composable definition of security for a QKD protocol.
\begin{definition}
    \label{def:security}
    Consider a generic QKD protocol, and let $\rho_{K_AK_BE}$ be the final state at the end of the protocol between Alice and Bob's key register, and Eve's side information register (including all public announcements made during the protocol). We define the following terms:
\begin{itemize}
    \item The protocol is $\esecret$-secret if any possible output state of the protocol satisfies 
    \begin{align}\label{eq:secrecy}
        \frac{1}{2} \pr{\Omega_{\mathrm{acc}}} \left\| \rho_{K_AE|\Omega_{\mathrm{acc}}}-\frac{\id_{K_A}}{|\mathcal{K}_A|}\otimes\rho_{E|\Omega_{\mathrm{acc}}}\right\|_1\leq\esecret.
    \end{align}
    where $\Omega_{\mathrm{acc}}$ is the event that the protocol accepts.
    \item The protocol is $\ecorr$-correct if any possible output state of the protocol satisfies 
    \begin{align}
        \label{eq:correctness}
        \pr{K_A\neq K_B\land \text{accept}}\leq\ecorr.
    \end{align}
        \item The protocol is $\ecom$-complete if, 
        for the state produced by the honest protocol implementation (described by some given noise model), we have 
    \begin{align}
        \label{eq:completeness}
        \pr{\text{abort}}\leq\ecom.
    \end{align}
\end{itemize}
If the protocol is $\esecret$-secret and $\ecorr$-correct, we refer to it as $\esecure$-secure, where $\esecure=\esecret+\ecorr$.
\end{definition}

\subsection{Security}\label{subsec:security}

We first prove the security statement for Protocol~\ref{Prot:PM Protocol}. We highlight that by using the GEAT instead of the EAT (which allows us to directly include the test-round announcements in the conditioning registers) and the {\Renyi} privacy amplification theorem of~\cite{Dupuis2023IEEETrans.Inf.Theory}, we obtain a remarkably simple formula for the secrecy parameter as compared to previous works --- it consists only of a single parameter $\ePA$ that appears in a simple fashion in the key length formula, informally representing the ``cost'' of the privacy amplification task.
\begin{theorem}
    For any $\alpha\in(1,3/2)$, $\ePA,\eEV\in(0,1]$, Protocol~\ref{Prot:PM Protocol} is $\ePA$-secret, $\eEV$-correct, and hence $(\ePA+\eEV)$-secure, when the length $l$ of the final key satisfies
    \begin{align}
    \label{eq:keylength}
    l\leq nh+nT_\alpha(f)-n\left(\frac{\alpha-1}{2-\alpha}\right)^2K(\alpha)-\lambdaEC-\ceil{\log\frac{1}{\eEV}}-\frac{\alpha}{\alpha-1}\log\frac{1}{\ePA}+2,
    \end{align}
    where $h,T_\alpha(f)$, and $K(\alpha)$ are as defined in Theorem~\ref{thrm:EAT}, and $\lambdaEC$ is the length of the error-correction string.
\end{theorem}
As discussed in~\cite[Theorem~V.2]{DF19}, the above formula converges asymptotically to the Devetak-Winter key rate by taking $\alpha = 1+O(1/\sqrt{n})$. However, to obtain the best finite-size key rates, one should optimize $\alpha$ numerically for each $n$, which we do in our subsequent computations.

We now present the proof of the above formula.
\begin{proof}
Proving correctness is straightforward: let us denote the event of Bob correctly guessing Alice's string by $\Omega_g$, then the expression in Eq.~(\ref{eq:correctness}) can be written as:
\begin{align}
    \label{eq:correction_apply}
    \pr{K_A\neq K_B\land \Omega_{\mathrm{acc}}}\leq\pr{K_A\neq K_B\land \Omega_{\mathrm{EV}}}\leq\pr{\Omega^c_g\land \Omega_{\mathrm{EV}}}\leq\pr{\Omega_{\mathrm{EV}}|\Omega^c_g}\leq 2^{-\ceil{\log \frac{1}{\eEV}}}\leq \eEV,
\end{align}
where $\Omega^c_g$ is the complementary event to $\Omega^c_g$. Thus, the protocol has correctness parameter $\ecorr=\eEV$. 

To prove secrecy, let us denote the final state at the end of the QKD protocol by $\rho_{K_AK_BI^nE'_nL}$, where $K_A,K_B$ are Alice and Bob's key registers respectively, $I^n$ denotes the public communication registers over all rounds, $E'_n$ is Eve's quantum side information register, and $L$ is the register containing the error correction and error verification information. Note that the event in which the protocol does not abort is $\Omega_{\mathrm{acc}}=\Omega_{\mathrm{EV}}\land\Omega_{\mathrm{AT}}$, where $\Omega_{\mathrm{EV}}$ is the event where the error verification succeeds, and $\Omega_{\mathrm{AT}}$ is the event of having the acceptance test step passes.
Then, the secrecy condition in Eq.~(\ref{eq:secrecy}) can be written as
\begin{align}
    \label{eq:secrecy2}
   \frac{1}{2}\pr{\Omega_{\mathrm{EV}} \land \Omega_{\mathrm{AT}}}\left\|\rho_{K_AI^nE'_nL_{|_{\Omega_{\mathrm{acc}}}}}-\frac{\id_{K_A}}{|\mathcal{K}_A|}\otimes\rho_{I^nE'_nL_{|_{\Omega_{\mathrm{acc}}}}}\right\|_1\leq\esecret.
\end{align}

To show the above bound holds, we first apply Theorem~\ref{thrm:privacy_amp} 
to find:
\begin{align}
    \label{eq:PA_S}
    \frac{1}{2}\left\|\rho_{K_AI^nE'_nL_{|_{\Omega_{\mathrm{acc}}}}}-\frac{\id_{K_A}}{|\mathcal{K}_A|}\otimes\rho_{I^nE'_nL_{|_{\Omega_{\mathrm{acc}}}}}\right\|_1\leq 2^{\frac{2(1-\alpha)}{\alpha}}2^{\frac{1-\alpha}{\alpha}\left(H_\alpha^\uparrow(S^n|I^nE'_nL)_{\rho_{|_{\Omega_{\mathrm{acc}}}}}-l\right)},
\end{align}
where $S^n$ is the key register before applying the hash function.
On the other hand, the entropic quantity in the exponent of the RHS of the above inequality can be bounded as follows:
\begin{align}
    \label{eq:alpha_bound}
    H_\alpha^\uparrow(S^n|I^nE'_nL)_{\rho_{|_{\Omega_{\mathrm{acc}}}}}&\ge H_\alpha^\uparrow(S^n|I^nE'_nL)_{\rho_{|_{\Omega_{\mathrm{AT}}}}}+\frac{\alpha}{1-\alpha}\log\frac{1}{\pr{\Omega_{\mathrm{EV}}|\Omega_{\mathrm{AT}}}}\cr
    &\ge H_\alpha^\uparrow(S^n|I^nE'_n)_{\rho_{|_{\Omega_{\mathrm{AT}}}}}+\frac{\alpha}{1-\alpha}\log\frac{1}{\pr{\Omega_{\mathrm{EV}}|\Omega_{\mathrm{AT}}}}-\lambdaEC-\ceil{\log(1/\eEV)},\cr
\end{align}
where in the first line we used Lemma B.5 in \cite{DFR20}, and in the second line we used Lemma 6.8 in \cite{Tom16}. Note that the $\lambdaEC$ term arises from the length of the error-correction string, while the $\ceil{\log(1/\eEV)}$ term arises from the length of the error-verification hash. 
It remains to bound the first term in Eq.~(\ref{eq:alpha_bound}). To do so, we first note that by setting $E_i\equiv E'_iI_1\cdots I_i
$ in Theorem~\ref{thrm:EAT}, 
we can apply that theorem to lower bound $H_\alpha^\uparrow(S^n|E_n)_{\rho_{|_{\Omega_{\mathrm{AT}}}}}$ via Eq.~(\ref{eq:EATbound}), modelling the protocol with a sequence of GEAT channels as shown in~\cite[Claim~10]{MR23} (see also~\cite[Sec.~5]{arx_AT25} for an improved model in which we do not need to restrict Eve to interact with a single signal at a time, while still obtaining the same entropy bounds). Combining this with Eq.~(\ref{eq:PA_S})--(\ref{eq:alpha_bound}), we can write\footnote{Notice that in this calculation, the log-probability terms in the preceding bound end up corresponding exactly to the accept-probability prefactors in the secrecy definition, as noted in~\cite{Dupuis2023IEEETrans.Inf.Theory}. This is basically the reason we have such a simple formula for the final secrecy parameter. Previous calculations based on the original EAT~\cite{arx_GLvH+22} did not have this property, because handling the test-round announcements in that framework involved the use of an additional chain rule that changed the {\Renyi} parameter.}
\begin{align}
    \label{eq:PA_2}
    &\frac{1}{2}\pr{\Omega_{\mathrm{EV}} \land \Omega_{\mathrm{AT}}}\left\|\rho_{KI^nE'_nL_{|_{\Omega_{\mathrm{acc}}}}}-\frac{\id_{K_A}}{|\mathcal{K}_A|}\otimes\rho_{I^nE'_nL_{|_{\Omega_{\mathrm{acc}}}}}\right\|_1\nonumber\\
    =&\frac{1}{2}\pr{\Omega_{\mathrm{EV}}|\Omega_{\mathrm{AT}}}\pr{\Omega_{\mathrm{AT}}}\left\|\rho_{KI^nE'_nL_{|_{\Omega_{\mathrm{acc}}}}}-\frac{\id_{K_A}}{|\mathcal{K}_A|}\otimes\rho_{I^nE'_nL_{|_{\Omega_{\mathrm{acc}}}}}\right\|_1\nonumber\\
    \leq& 2^{\frac{1-\alpha}{\alpha}\left(2+nh+T_\alpha(f)-n\left(\frac{\alpha-1}{2-\alpha}\right)^2K(\alpha)-\lambdaEC-l-\ceil{\log(1/\eEV)}\right)}
    ,
\end{align}
Comparing this to the secrecy definition in Eq.~(\ref{eq:secrecy2}), we conclude that  the protocol is $\ePA$-secret as long as the key length satisfies the formula~\eqref{eq:keylength}.

Combining Eq.~(\ref{eq:keylength}) and Eq.~(\ref{eq:correction_apply}), we obtain the statements of the theorem.
\end{proof}

\subsection{Completeness}\label{subsec:completeness}

We now turn to the completeness parameter $\ecom$, i.e.~an upper bound on the probability of an ``accidental'' abort in the honest case. Recalling that the only points the protocol could abort are during the acceptance test (Step~\ref{step:AT}) or error verification (Step~\ref{step:ECandEV}), from the union bound we immediately see it suffices to choose
\begin{align}
\ecom = \ecomAT + \ecomEV,
\end{align}
where $\ecomAT$ and $\ecomEV$ are 
any upper bounds on 
the probabilities of the (honest) behaviour aborting during the acceptance test and error verification respectively. Since completeness is a property that is completely unrelated to the security of the protocol (it only involves how often the honest protocol ``accidentally'' aborts), the value of $\ecom$ can usually be chosen less stringently as compared to $\esecure$. In this work, we shall choose values on the order of $10^{-3}$.

We begin by considering the $\ecomAT$ term. First, we highlight that our above discussions regarding the security of the protocol were proven for arbitrary choices of the acceptance set $S_{\mathrm{acc}}$. However, to compute more explicit bounds on $\ecomAT$, we shall now focus on specific choices of $S_{\mathrm{acc}}$: namely, we consider protocols where the acceptance set is defined by ``entrywise'' constraints with respect to the honest behaviour, i.e.~letting $\mathbb{P}_{\mathcal{C}}$ denote probability distributions on $\mathcal{C}$, we have
\begin{align}\label{eq:acceptentrywise}
S_{\mathrm{acc}} = \left\{ \mbf{p}^\mathrm{acc} \in \mathbb{P}_{\mathcal{C}} \;\middle|\; \forall c \in \mathcal{C},\; \phonc -  \tlowc \leq p_c^\mathrm{acc} \leq \phonc + \tuppc \right\},
\end{align}
where $\tlow, \tupp \in \mathbb{R}_{\geq 0}^{|\mathcal{C}|}$, and $\phon$ denotes the distribution produced on a single-round register $C_j$ by the honest behaviour.\footnote{Note that by the structure specified in the protocol, we will always have $p^\mathrm{hon}_{\perp} = 1-\gamma$. 
}
Our analysis of $\ecomAT$ (and also one other step in our explicit key rate calculations later on, explained in Sec.~\ref{subsubsec:deltacomp}) will be focused on $S_{\mathrm{acc}}$ of the above form.

To bound $\ecomAT$ for such $S_{\mathrm{acc}}$, first focus on some specific value $c \in \mathcal{C}$. Now observe that since the honest behaviour is IID, the frequency of outcome $c$ in the string it produces simply follows a binomial distribution, with the single-trial ``success probability'' being $\phonc$. 
Hence the probability that the frequency of $c$ lies outside the acceptance interval $[\phonc -  \tlowc , \phonc + \tuppc]$ specified by $S_{\mathrm{acc}}$ can simply be written as
\begin{align}
\pr[X\sim\operatorname{Binom}\left(n,\phonc\right)]{\frac{X}{n} \notin [\phonc -  \tlowc , \phonc + \tuppc]},
\end{align}
where $X\sim\operatorname{Binom}(n,p)$ denotes a random variable $X$ following a binomial distribution with parameters $(n,p)$ (i.e.~$X$ is the sum of $n$ IID Bernoulli random variables $X_j$ with $\pr{X_j=1}=p$). 
Applying the union bound, we conclude that it suffices to take $\ecomAT$ to be any value such that
\begin{align}\label{eq:ecomATbnd}
\ecomAT \geq \sum_c \pr[X\sim\operatorname{Binom}\left(n,\phonc\right)]{\frac{X}{n} \notin [\phonc -  \tlowc , \phonc + \tuppc]}.
\end{align}
We note that the binomial-distribution probabilities in the above formula can be computed using inbuilt functions in Matlab, e.g.~the regularized beta function (with some special-case handling if $\phonc$ is too close to $0$ or $1$).
Hence in this work we use the above formula for $\ecomAT$; more specifically, we use the fixed choice $\ecomAT = 10^{-3}$ and optimize the values of $\tlow, \tupp$ while ensuring the above bound is satisfied (further details in Sec.~\ref{subsubsec:deltacomp}).

\begin{remark}
\label{remark:unique}
A number of past works on QKD have simplified their analysis by focusing on the case where $S_{\mathrm{acc}}$ only includes the single distribution asymptotically produced by the honest behaviour; such protocols are sometimes referred to as \term{unique-acceptance} protocols. Unique-acceptance protocols are not practical to implement, since in the finite-size regime they would almost always abort even under the honest behaviour, i.e.~the completeness parameter is trivial. However, for the purposes of comparison to previous works, in our later computations in Sec.~\ref{sec:Qubit BB84 with loss}--\ref{sec:Decoy-state with improved analysis} we also present some results for unique-acceptance protocols (though we then show how much our key rates change when instead using a realistic acceptance set with nontrivial completeness parameter).
\end{remark}

We now discuss the $\ecomEV$ term. Since the error verification step always accepts if Bob produces a correct guess $\widehat{S}_1^n$ for Alice's string $S_1^n$ after the error-correction procedure, we see that $\ecomEV$ is simply upper bounded by the probability that the error-correction procedure produces an incorrect guess for Bob. 
Fundamentally, this can be analyzed as a task of source compression with side information \cite{Slepian73,Tan13}, with Bob holding the registers $Y_1^n I_1^n$ as his side-information.
Since we are focusing on the honest behaviour, we can restrict to the IID case. 
For this, it is known~\cite{rennerthesis} that in the asymptotic limit, 
vanishing error probability can be achieved by having Alice send approximately $\lambdaEC \approx nH(S|YI)_\mathrm{hon}$ bits, where 
$H(S|YI)_\mathrm{hon}$ denotes the single-round conditional entropy in the honest behaviour.
Unfortunately, in the finite-size regime, many error-correction procedures used in practice do not have tight rigorously proven upper bounds on this probability, only heuristic estimates. For the purposes of this work, we follow a standard heuristic\footnote{We stress however that the heuristic nature of this step only affects the \emph{completeness} parameter, i.e.~the probability of the honest behaviour aborting --- it does not in any way affect the \emph{security} parameter we derived above against dishonest behaviour, which was fully rigorous and involved no heuristics.}~\cite{TMP+17,HL+06}: we take the length of the error-correction string to be
\begin{align}
    \label{eq:error_correction}
    \lambdaEC= nf_{\mathrm{EC}}H(S|YI)_\mathrm{hon},
\end{align}
where 
$f_{\mathrm{EC}}>1$ is an ``error-correction efficiency'' parameter that quantifies how much it differs from the asymptotic value, and we suppose that ``reasonable'' $\ecomEV$ values (say, $10^{-3}$) can be achieved by setting $f_{\mathrm{EC}} = 1.16$. We use this choice throughout all computations in this work. The computation of $H(S|YI)_\mathrm{hon}$ can usually be slightly simplified by using the structure of the public announcements; we describe this for specific protocols in Sec.~\ref{sec:Qubit BB84 with loss}--\ref{sec:Decoy-state with improved analysis}.

\subsection{Formulating the rate function}
\label{subsec:ratefunctions}
\newcommand{\channE}{\mathcal{E}}
\newcommand{\rhotJ}{\rho^t_J}
\newcommand{\rhogJ}{\rho^g_J}
\newcommand{\rhotE}{\rho^t_{\mathcal{E}}}
\newcommand{\rhogE}{\rho^g_{\mathcal{E}}}
\newcommand{\probst}{\mathbf{\Phi}}
\newcommand{\probt}{\Phi}
\newcommand{\Cnoperp}{\mathcal{C}_{\setminus \perp}}

We now describe more precisely the rate function for our protocol. Note that due to the techniques we apply in this work, we will not need to explicitly evaluate the rate function in the form described in this section; rather, we will only be computing modified versions of it that are described in Sec.~\ref{sec:finitesize}.

Consider a single round in Protocol~\ref{Prot:PM Protocol}. For brevity, in this section we omit the subscripts specifying individual rounds. By the source-replacement technique, we can equivalently describe Alice's preparation process in this round as follows: with probability $\gamma$ she prepares a ``test-round state'' \(\ket{\xi^t}_{AA'}\) and sends out the system \(A'\) (keeping the system \(A\) to be measured later with the POVM \(M^A\)), and otherwise she prepares a ``generation-round state'' \(\ket{\xi^g}_{AA'}\).
Eve then applies some channel\footnote{More precisely, in terms of the GEAT channels, this channel $\channE$ consists of Eve appending her past side-information to the register $A'$ and then applying the operation she performs in the GEAT channel $\mathcal{M}_i$. Strictly speaking, this means that the optimization~\eqref{eq:rateE} we define in this section is not in \emph{precisely} the same form as the optimization~\eqref{eq:ratefuncdef} in the definition of a rate function, as the former is an optimization over channels while the latter is an optimization over input states to the GEAT channels. However, it is clear from the structure of our prepare-and-measure protocol that the optimizations are basically equivalent once we consider the ``worst-case behaviour'' over all GEAT channels $\mathcal{M}_i$ that could describe Eve's possible attacks (see also~\cite{MR23} for similar discussion, under the term ``collective-attack bounds'').} $\channE: A' \rightarrow B$ on the  \(A'\) register;
let us denote the resulting states on $AB$ in the test and generation cases as (respectively)
\begin{align}\label{eq:rhotg_channel}
\begin{gathered}
\rhotE \defvar \left( \idmap_A \otimes \channE \right) \left( \ketbra{\xi^t}{\xi^t}_{AA'} \right) , \\
\rhogE \defvar  \left( \idmap_A \otimes \channE \right) \left( \ketbra{\xi^g}{\xi^g}_{AA'} \right) .
\end{gathered}
\end{align}
She then forwards the \(B\) register to Bob, who performs his measurements. Bob's measurements can be described as acting upon the state with his (potentially squashed) POVM \(M^B\). Alice and Bob then make some public announcement $I$, and Alice uses it to process her raw measurement outcome into some value $S$, as described by the protocol.

Let $W(\rhogE)$ denote the conditional entropy
$H(S|IE)$
of the state that would be produced at the end of this process if Eve keeps an arbitrary purification $E$ when applying the channel $\mathcal{E}$. Note that in our above notation $\rhogE$ is a state on only the registers $AB$ rather than $ABE$, but since all purifications are isometrically equivalent, the resulting value of $H(S|IE)$ can indeed be computed using only this state on $AB$, as discussed in~\cite{WLC18}. 
Furthermore, as shown in that work, $W$ can be expressed as a convex function (specifically, a relative entropy), which we present in more detail in Sec.~\ref{sec:Qubit BB84 with loss}--\ref{sec:Decoy-state with improved analysis}.
Also note that in our protocol, it depends only on $\rhogE$ and not $\rhotE$ because when the round is a test round, Alice simply announces her outcome and hence it does not contribute any entropy; we elaborate further on this when presenting the detailed formulas in those sections. 

Furthermore, conditioned on the round being a test round, let $\probt_c[\rhotE]$ denote the (conditional) probability of obtaining the value $c$ on the $C$ register when measuring $\rhotE$ --- put another way, the overall (unconditioned) distribution on the $C$ register is such that the probability of value $c$ is $\gamma\probt_c[\rhotE]$ for $c\neq\perp$, and $1-\gamma$ for $c=\perp$. For brevity, we introduce the notation 
\begin{align}
\Cnoperp \defvar \mathcal{C} \,\setminus \{\perp\},
\end{align}
i.e.~the alphabet $\mathcal{C}$ with the $\perp$ symbol excluded, and similarly for any probability distribution $\mbf{p}$ we let $\mbf{p}_{\setminus\perp}$ denote the components of $\mbf{p}$ with $p_\perp$ omitted. 
Also, we shall write the tuple of values $\left\{\probt_c[\rhotE] 
\,\middle|\, c\in\Cnoperp
\right\}$ as $\probst[\rhotE]$ (notice that this tuple does not include a $c=\perp$ term, since $\probt_c[\rhotE]$ is not defined for that case by our above discussion).

We thus see that a valid choice of rate function would be given by the following optimization problem:
\begin{align}\label{eq:rateE}
\begin{aligned}
\rate(\mbf{p}) \defvar\;
\inf_{
\channE} &\quad W(\rhogE) \\ 
\suchthat &\quad 
\gamma \probst[\rhotE] = \mbf{p}_{\setminus\perp}
,
\end{aligned}
\end{align}
where $\rhotE,\rhogE$ are defined in terms of the channel $\channE$ via~\eqref{eq:rhotg_channel}.
Note that for any distribution $\mbf{p}$ such that $p_\perp \neq 1-\gamma$ the optimization is infeasible (because the constraints enforce that $\sum_{c\neq\perp}p_c = \sum_{c\neq\perp}\probt_c[\rhotE] = \gamma$, so $p_\perp = 1-\gamma$ by normalization), and hence we have $\rate(\mbf{p})=+\infty$ for such $\mbf{p}$.
Any affine lower bound on this function $\rate(\mbf{p})$ would be a valid min-tradeoff function.

The above optimization takes place over arbitrary channels $\channE: A' \rightarrow B$, which may appear difficult to describe. However, we can in fact handle it by parametrizing the channel via its Choi state \(J\); specifically, note that the action of a channel can be expressed in terms of its Choi state via \(\channE(X) = \tr[A']{J\left( X^T \otimes \id_{B} \right)}\). This gives us
\begin{align}\label{eq:rateJ}
\begin{aligned}
\rate(\mbf{p}) =\;
\inf_{
J} &\quad W(\rhogJ) \\ 
\suchthat &\quad 
\gamma \probst[\rhotJ] = \mbf{p}_{\setminus\perp}
,
\end{aligned}
\end{align}
where
\begin{align}\label{eq:rhotg_Choi}
\begin{gathered}
\rhogJ \defvar \tr[A']{ \left( \id_A \otimes J \right) \left( \ketbra{\xi^g}{\xi^g}^{T_{A'}} \otimes \id_B \right) }, \\
\rhotJ \defvar \tr[A']{ \left( \id_A \otimes J \right) \left( \ketbra{\xi^t}{\xi^t}^{T_{A'}} \otimes \id_B \right) }.
\end{gathered}
\end{align}
Importantly, this parameterization yields a convex optimization problem (given that $W$ is a convex function of $\rhogJ$ and $\Phi_c$ are affine functions of $\rhotJ$). In the remainder of this work, we will mostly focus on using this formulation. Note that the $\probt_c[\rhotJ]$ terms in the constraints can be written more directly in terms of the Choi matrix $J$ by substituting in the expression for $\rhotJ$ to obtain
\begin{align}
\probt_c[\rhotJ] = \tr{\left( M_a^A \otimes J \right) \left( \ketbra{\xi^t}{\xi^t}_{AA'}^{T_{A'}} \otimes M_b^B\right) },
\end{align}
where \(c=(a,b)\).

The constraints in the above optimization have terms on the order of the testing probability $\gamma$, which can present numerical stability issues when $\gamma$ is small. To mitigate this, we use the fact that the channels we consider here are \term{infrequent-sampling channels} in the sense defined in~\cite{DF19}, i.e.~with probability $\gamma$ they perform a test round and record some nontrivial value in the $C$ register, while with probability $1-\gamma$ they simply set the $C$ register to $\perp$.\footnote{Still, it is worth keeping in mind that in the protocol description and this analysis, what we have chosen to call a test round is determined entirely by Alice. Bob's basis choice and announcement will be accounted for in the description of the POVMs on his side, with the subsequent sifting being accounted for in this context within the functions $W$ and $\probst$, rather than the conversion we now describe between crossover min-tradeoff functions and min-tradeoff functions.}
For such channels, we can follow~\cite{DF19,LLR+21} and define the notions of \term{crossover rate functions} and \term{crossover min-tradeoff functions}, which only bound the (single-round) entropy in terms of the distribution \emph{conditioned} on being a test round, and ``cross over'' the resulting entropy bound into a function of the full distribution on the $C$ register. Specifically, the following function (where $\mbf{q}$ is a distribution on $\Cnoperp$) is a crossover rate function~\cite{DF19,LLR+21}:
\begin{align}\label{eq:crossrateJ}
\begin{aligned}
r_\mathrm{cross}(\mbf{q}) \defvar\;
\inf_{J} &\quad W(\rhogJ) \\ 
\suchthat &\quad 
\probst[\rhotJ] = \mbf{q}
.
\end{aligned}
\end{align}
Note that for any distribution $\mbf{p}$ on $\mathcal{C}$, we can express the value of $\rate(\mbf{p})$ in~\eqref{eq:rateJ} in terms of the function $r_\mathrm{cross}$, 
as follows. First, if $\mbf{p}$ is such that $p_\perp = 1-\gamma$, it must be of the form $\mbf{p} = (\gamma\mbf{q},1-\gamma)$ (where the last term is to be interpreted as the $\perp$ probability) for some distribution $\mbf{q}$, and by comparing the optimizations~\eqref{eq:rateJ} and~\eqref{eq:crossrateJ} we see that we simply have $\rate(\mbf{p}) = r_\mathrm{cross}(\mbf{q})$. As for any other value of $\mbf{p}$, as noted previously the optimization~\eqref{eq:rateJ} is infeasible in that case, and we simply have $\rate(\mbf{p}) = +\infty$.

\newcommand{\cperp}{\Max(g)} 

In turn, any affine lower bound $g(\mbf{q})$ on $r_\mathrm{cross}(\mbf{q})$ serves as a valid crossover min-tradeoff function. 
Similar to the previous section, we note that since it is affine, it can always be expressed as
\begin{align}\label{eq:affineg}
g(\mbf{q}) = \mbf{g}\cdot\mbf{q} + k_g,
\end{align}
for some gradient vector $\mbf{g}$ and scalar $k_g$, and we will often use this formulation in our subsequent analysis.
As shown in~\cite{DF19}, we can directly convert any crossover min-tradeoff function into a min-tradeoff function, as follows: given a crossover min-tradeoff function $g$, a valid choice\footnote{In~\cite{LLR+21}, it was noted that there is an additional degree of freedom in this construction which could be optimized to slightly improve the key rates; however, we do not consider it here as the potential improvements appear to be small~\cite{TSB+22}. Informally, this may be because this degree of freedom appears in the exponent terms in $K(\alpha)$, and hence setting it to any value other than the one used in~\cite{DF19} has a large negative effect on the key rate. Furthermore, using that degree of freedom causes some complications in some of our subsequent calculations; e.g.~some expressions would depend on the scalar $k_g$ in~\eqref{eq:affineg} instead of only the gradient vector $\mbf{g}$.} of min-tradeoff function is given by the unique affine function $f$ specified by the following values:
\begin{align}
\fmin(\delta_c) = 
\begin{cases} 
\frac{1}{\gamma} \g(\delta_c) + \left(1-\frac{1}{\gamma}\right)\cperp & \text{ if } c\neq\perp \\
\cperp & \text{ if } c = \perp
\end{cases}
\, ,
\label{eq:g_to_f}
\end{align}
where $\delta_c$ denotes the point distribution on the value $c$ (i.e.~$\Pr[c]=1$ and all other probabilities are zero).
The above formula can be used to evaluate $f(\mbf{p})$ (in terms of $g$) for arbitrary $\mbf{p}$, using the linearity of affine functions with respect to convex combinations --- explicitly, we have:
\begin{align}
\fmin(\mbf{p}) = \sum_c p_c\fmin(\delta_c) 
&= \left(\sum_{c\neq\perp} p_c \left(\frac{1}{\gamma} \g(\delta_c) + \left(1-\frac{1}{\gamma}\right)\cperp\right)\right) + p_\perp
\cperp \notag\\
&= \frac{1}{\gamma} \mathbf{\g}\cdot\mbf{p}_{\setminus\perp} + (1-p_\perp)\left(\frac{k_g}{\gamma} + \left(1-\frac{1}{\gamma}\right)\cperp \right) + p_\perp\cperp \notag\\
&= \frac{1}{\gamma} \mathbf{\g}\cdot\mbf{p}_{\setminus\perp} + p_\perp\left(\frac{\cperp-k_g}{\gamma}\right) + \left(\left(1-\frac{1}{\gamma}\right)\cperp + \frac{k_g}{\gamma}\right),
\end{align}
where in the second line onwards we have further rewritten the expressions in terms of the formulation~\eqref{eq:affineg} for $g$.

By comparing the last line to~\eqref{eq:affinef}, we see that this choice of $\fmin$ has gradient
\begin{align}\label{eq:gradf_from_g}
\mbf{f} = \frac{1}{\gamma}(\mbf{\g},\cperp-k_g) = \frac{1}{\gamma}(\mbf{\g},\max(\mbf{g})),
\end{align}
where by $\max(\mbf{g})$ we simply mean the maximum of the vector components $g_c$. Hence the gradient of $f$ can be computed using only the gradient of $g$ (and the testing probability $\gamma$).
Similarly, note that for any distribution of the form $\mbf{p} = (\gamma\mbf{q},1-\gamma)$, this choice of $f$ satisfies
\begin{align}
\Var(\mbf{p},f)
&=\sum_cp_c\Big(f(\delta_c)-\sum_{c'}p_{c'}f(\delta_{c'})\Big)^2\cr
&=(1-\gamma)\Big(\Max(g)-g(\mbf{q})\Big)^2 + \sum_{c\neq\perp}\gamma q_c\Big(\Max(g)+\frac{1}{\gamma}\big(g(\delta_c)-\Max(g)\big)-g(\mbf{q})\Big)^2 \cr
&=(1-\gamma)\Big(\Max(g)-g(\mbf{q})\Big)^2 + \sum_{c\neq\perp}\frac{q_c}{\gamma}\big(g(\delta_c)-\Max(g)\big)^2 \cr
&\qquad + 2(\Max(g)-g(\mbf{q})) \sum_{c\neq\perp} q_c\big(g(\delta_c)-\Max(g)\big) + \gamma (\Max(g)-g(\mbf{q}))^2 \sum_{c\neq\perp} q_c \cr
&=(1-\gamma)\Big(\Max(g)-g(\mbf{q})\Big)^2 + \sum_{c\neq\perp}\frac{q_c}{\gamma}\big(g(\delta_c)-\Max(g)\big)^2 \cr
&\qquad + 2(\Max(g)-g(\mbf{q})) \big(g(\mbf{q})-\Max(g)\big) + \gamma (\Max(g)-g(\mbf{q}))^2 \cr
&=\sum_{c\neq\perp}\frac{q_c}{\gamma}\Big(\Max(g)-g(\delta_c)\Big)^2-\Big(\Max(g)-g(\mbf{q})\Big)^2\cr
&=\sum_{c\neq\perp}\frac{q_c}{\gamma}\Big(\max(\mbf{g})-g_c\Big)^2-\Big(\max(\mbf{g}) - \mbf{g}\cdot\mbf{q}\Big)^2,
\end{align}
which depends only on $\mbf{q}$ and the gradient\footnote{Intuitively, the fact that it does not depend on the scalar $k_g$ arises from the fact that $\Var(\mbf{p},f)$ is the \emph{variance} of the function $f$ with respect to the distribution $\mbf{p}$, which is inherently unchanged if the value of $f$ is shifted by an additive constant. The above calculation serves to verify that a similar property holds in terms of the crossover min-tradeoff function $g$ as well, when $f$ is defined from $g$ via~\eqref{eq:g_to_f} (in which shifting the value of $g$ by an additive constant also only shifts the value of $f$ by the same constant, without changing its gradient).} $\mbf{g}$ (and the test probability $\gamma$). Therefore, when $\mbf{p} = (\gamma\mbf{q},1-\gamma)$ we can write
\begin{align}\label{eq:Vf_from_g}
V(\mbf{p},f) = \widetilde{V}(\mbf{q},\mbf{g}) 
\defvar \Bigg(\log (1+2d_A^\CvsQ)+\sqrt{2+\sum_{c\neq\perp}\frac{q_c}{\gamma}\Big(\max(\mbf{g})-g_c\Big)^2-\Big(\max(\mbf{g}) - \mbf{g}\cdot\mbf{q}\Big)^2}\Bigg)^2,
\end{align}
i.e.~for $\mbf{p}$ of that form, $V(\mbf{p},f)$ can also be computed using only $\mbf{q}$, $\mbf{g}$ and $\gamma$. In the next section, we will use these formulas to analyze the second-order terms.

\section{Key rate computation techniques}
\label{sec:finitesize}

\subsection{Obtaining secure lower bounds}

We first explain how, \emph{given} a choice of min-tradeoff function (without yet discussing how we would arrive at one), we can explicitly compute a secure choice for the key length of the protocol.

To begin, we ease our subsequent analysis by expressing the key length formula~\eqref{eq:keylength} in terms of the expected probability distribution $\phon$ (on a single-round $C$ register) that would be produced by the honest behaviour (note that since each round is tested with probability $\gamma$, $\phon$ must be of the form $\phon = (\gamma\qhon,1-\gamma)$ for some distribution $\qhon$, which will be important in some later calculations). Specifically, letting $S_\mathrm{acc}$ denote a set that contains all frequency distributions accepted during parameter estimation, the constant $h$ in that formula can be bounded via
\begin{align}
h = \inf_{c^n\in\Omega} f(\freq_{c^n}) &= \inf_{\mbf{p}^\mathrm{acc} \in S_\mathrm{acc}} f(\mbf{p}^\mathrm{acc}) \notag\\
&= f(\phon) + \inf_{\mbf{p}^\mathrm{acc} \in S_\mathrm{acc}} \left(f(\mbf{p}^\mathrm{acc}) - f(\phon) \right)\notag\\
&= f(\phon) - \Delta_\mathrm{com}
\quad \text{for} \quad \Delta_\mathrm{com}
\defvar \sup_{\mbf{p}^\mathrm{acc} \in S_\mathrm{acc}} -\mbf{f}\cdot(\mbf{p}^\mathrm{acc} - \phon). \label{eq:deltacomp}
\end{align}
The quantity $\Delta_\mathrm{com}$ can be viewed as the penalty to our key rate caused by having some ``tolerance'' in the accept condition, rather than being a unique-acceptance protocol (see Remark~\ref{remark:unique}).
In order to securely upper bound $\Delta_\mathrm{com}$, we note that if for instance $S_\mathrm{acc}$ is a convex polytope, then the task of computing $\Delta_\mathrm{com}$ is simply a linear program (LP). In this work we will be focusing on protocols where this is indeed the case, hence $\Delta_\mathrm{com}$ can be securely bounded. (We give more details for optimizing the choice of $S_\mathrm{acc}$ in the subsequent section, as well as simplifications of the LP when $S_\mathrm{acc}$ is of the form~\eqref{eq:acceptentrywise}.)

With this in mind, we see that to find a secure key length via~\eqref{eq:keylength}, it would suffice to find a lower bound on the value of $f(\phon) + T_\alpha(f)$ (since by the above calculation we could just subtract $\Delta_\mathrm{com}$ to get $h + T_\alpha(f)$). To do so, we first note that since we are computing the rate function via the formula~\eqref{eq:rateJ}, we can substitute that formula into the definition of $T_\alpha(f)$ (Eq.~\eqref{eq:2ndorderfuncs}) 
to get the following simplifications (in which $\rhogJ,\rhotJ$ are defined via~\eqref{eq:rhotg_Choi} as before): 
\begin{align}
    &f(\phon) + T_\alpha(f) \nonumber\\
    =& f(\phon) + \inf_{J} \left(W(\rhogJ) -f(\mbf{p}_J)-\frac{\alpha-1}{2-\alpha}\frac{\operatorname{ln}(2)}{2}V(\mbf{p}_J,f) \right) 
    \text{ where } \mbf{p}_J\defvar \left(\gamma\probst[\rhotJ],1-\gamma\right) \nonumber\\
    =& \inf_{J} \left(W(\rhogJ) + \mbf{f}\cdot (\phon - \mbf{p}_J) -\frac{\alpha-1}{2-\alpha}\frac{\operatorname{ln}(2)}{2}V(\mbf{p}_J,f) \right) \nonumber\\
    =& \inf_{J} \left(W(\rhogJ) + \mbf{g}\cdot\left(\qhon - \probst[\rhotJ]\right) -\frac{\alpha-1}{2-\alpha}\frac{\operatorname{ln}(2)}{2}\widetilde{V}\!\left(\probst[\rhotJ],\mbf{g}\right) \right)  
    , 
    \label{eq:finalopt}
\end{align}
where in the last line we use the fact that $\phon = (\gamma\qhon,1-\gamma)$ and apply the formulas~\eqref{eq:gradf_from_g} and \eqref{eq:Vf_from_g} from the previous section. 
A critical feature of the above formula is that as shown in~\cite[Supplement, Lemma~2]{LLR+21}, the term $\widetilde{V}\!\left(\probst[\rhotJ],\mbf{g}\right)$ is concave with respect to $\probst[\rhotJ]$, which is in turn an affine function of $J$. This implies that the above optimization consists of minimizing a (differentiable) convex objective function over the set of Choi matrices $J$ --- this is a task that can be tackled using for instance the Frank-Wolfe (FW) algorithm~\cite{FW56}, as was done in~\cite{WLC18} (see however Remark~\ref{remark:FW} regarding its convergence rates). In such contexts, this algorithm yields reliable lower bounds on the true minimum value, up to the precision of the SDP solver used in the last step of the algorithm. 

\begin{remark} \label{remark:FW}
A common problem encountered by FW methods is known as zigzagging, where the iterates move only very slowly closer towards the optimal point. This often occurs if the optimal point is on the boundary or close to it. In order to circumvent this problem we implemented improved FW methods from \cite{LJ15}. Those improved FW methods minimize the zigzagging and can reach higher convergence rates than the \(\mathcal{O}(1/N)\) rate of the original FW algorithm, where \(N\) is the number of iterations. We observed that these improved FW methods increased the convergence speed drastically.
\end{remark}

In summary, substituting~\eqref{eq:deltacomp} and~\eqref{eq:finalopt} into the key length formula~\eqref{eq:keylength}, we obtain
\begin{align}
    l&\leq n\left(\inf_{J} \left(W(\rhogJ) + \mbf{g}\cdot\left(\qhon - \probst[\rhotJ]\right) -\frac{\alpha-1}{2-\alpha}\frac{\operatorname{ln}(2)}{2}\widetilde{V}\!\left(\probst[\rhotJ],\mbf{g}\right) \right) - \Delta_\mathrm{com} \right) \nonumber\\
    &\qquad -n\left(\frac{\alpha-1}{2-\alpha}\right)^2K(\alpha)-\lambdaEC-\ceil{\log\frac{1}{\eEV}}-\frac{\alpha}{\alpha-1}\log\frac{1}{\ePA}+2, \label{eq:keylengthfinal}
\end{align}
where $\rhogJ,\rhotJ$ are defined via~\eqref{eq:rhotg_Choi}, and hence the minimization over $J$ is a convex optimization that can be bounded using for instance the Frank-Wolfe algorithm (again, see also Remark~\ref{remark:FW}). An interesting feature of the above expression is that the only way in which it depends on the crossover min-tradeoff function $g$ is via the gradient $\mbf{g}$, \emph{not} the scalar term $k_g$. In principle, any choice of gradient $\mbf{g}$ corresponds to some valid crossover min-tradeoff function (simply by choosing the scalar term $k_g$ appropriately; we discuss this in Appendix~\ref{app:modprimal}). Hence, the above formula in fact allows us to compute a secure key rate for any choice of $\mbf{g}$, without having to explicitly find the crossover min-tradeoff function $g$ itself. However, optimizing the choice of $\mbf{g}$ to obtain the best possible key length from the above formula is still a nontrivial task, as $\mbf{g}$ can be a high-dimensional vector in the case of protocols with many measurement choices and outcomes (e.g.~decoy-state protocols). In the next section, we discuss a method to obtain an (approximately) optimal such choice.

\subsection{Optimizing parameter choices}
\label{subsec:optparams}

\subsubsection{Min-tradeoff function}\label{subsubsec:optfmin}
\newcommand{\newT}{\hat{T}}
\newcommand{\lambdavec}{\boldsymbol{\lambda}}

We first describe how, given fixed values for $\gamma,\alpha$, there is a systematic method to find a choice of $\mbf{g}$ that approximately maximizes the key length given by~\eqref{eq:keylengthfinal}. The techniques described here are based on those in~\cite{arx_GLvH+22}; however, we shall directly use Lagrange duality instead of invoking Fenchel duality as in that work. 

To begin, note that the only terms in the key length formula~\eqref{eq:keylengthfinal} that depend on $\mbf{g}$ are the infimum over $J$ (i.e.~the expression~\eqref{eq:finalopt}), the $K(\alpha)$ term, and the $\Delta_\mathrm{com}$ term. 
For ease of analysis, we do not consider the dependencies on\footnote{For the $K(\alpha)$ term at least, we can say it should only have a very small effect on the key length, as it is has a prefactor of order $(\alpha-1)^2$, which scales asymptotically as $O(1/n)$~\cite{DF19} and hence quickly becomes small.} $K(\alpha)$ and $\Delta_\mathrm{com}$ for the purposes of approximately optimizing $\mbf{g}$, focusing only on finding the choice of $\mbf{g}$ that maximizes the expression~\eqref{eq:finalopt}. Now, while it would be ideal to tackle that term ``directly'', we did not find a method to do so via the following techniques (basically because these techniques rely on computing explicit Legendre-Fenchel conjugates, which we were unable to find for the expression~\eqref{eq:finalopt} itself), and hence we shall instead begin by first reducing that expression to an approximate lower bound with a simpler formula.\footnote{We stress that while this approximation is not necessarily a secure lower bound on the key length, we \emph{only} use it for the purposes of finding an approximately optimal choice of $\mbf{g}$ to substitute into the actual key length formula~\eqref{eq:keylengthfinal}. The latter does indeed yield a secure lower bound, and is what we use in our final computations, so we are not in danger of overestimating the true secure key length.}

By considering the second line in the computation~\eqref{eq:finalopt}, we see it can be lower-bounded as follows (as was done in~\cite{DF19}):
\begin{align}
&f(\phon) + \inf_{J} \left(W(\rhogJ) -f(\mbf{p}_J)-\frac{\alpha-1}{2-\alpha}\frac{\operatorname{ln}(2)}{2}V(\mbf{p}_J,f) \right) \nonumber\\
\geq& f(\phon) + \inf_{J} \left(-\frac{\alpha-1}{2-\alpha}\frac{\operatorname{ln}(2)}{2}V(\mbf{p}_J,f) \right) \nonumber\\
\geq& f(\phon) -\frac{\alpha-1}{2-\alpha}\frac{\operatorname{ln}(2)}{2}\left(\log(1+2d_A^{\CvsQ})+\sqrt{2+\frac{1}{\gamma}\left(\max(\mbf{g})-\min(\mbf{g})\right)^2}\right)^2 ,
\end{align}
where the first inequality holds because $W(\rhogJ) -f(\mbf{p}_J) \geq 0$, and the second inequality is from~\cite{DF19} Lemma~V.5.
Expanding the square in the above expression and recalling that $f(\phon)=g(\qhon)$, we see the above is equal to
\begin{align}
& g(\qhon) -\frac{\alpha-1}{2-\alpha}\frac{\operatorname{ln}(2)}{2}\bigg( \left(\log(1+2d_A^{\CvsQ})\right)^2 + 2\log(1+2d_A^{\CvsQ})\sqrt{2+\frac{1}{\gamma}\left(\max(\mbf{g})-\min(\mbf{g})\right)^2} + \nonumber\\
&\qquad + 2 + \frac{1}{\gamma}\left(\max(\mbf{g})-\min(\mbf{g})\right)^2 \bigg) \nonumber\\
\approx\,& g(\qhon) - \varphi_0(\left(\max(\mbf{g})-\min(\mbf{g})\right)^2 - \varphi_1\left(\max(\mbf{g})-\min(\mbf{g})\right) - \varphi_2 \nonumber\\
=\,& g(\qhon) - \newT(\mbf{g}) - \varphi_2 \qquad \text{for } \newT(\mbf{g}) \defvar \varphi_0(\left(\max(\mbf{g})-\min(\mbf{g})\right)^2 + \varphi_1\left(\max(\mbf{g})-\min(\mbf{g})\right)
, \label{eq:approxrate}
\end{align}
where the approximation in the second line is obtained by dropping the $2$ in the square root term\footnote{This approximation is reasonably accurate whenever $\left(\max(\mbf{g})-\min(\mbf{g})\right)^2/\gamma \gg 2$, which is typically the case for the parameter regimes used in this work, where $\gamma$ is small. Since all the other computations here are genuine lower bounds, this at least suggests it is unlikely that we will end up with some choice of $\mbf{g}$ that is ``good'' for~\eqref{eq:approxrate} but yields worse results in the actual key length formula~\eqref{eq:keylengthfinal}.}, and writing
\begin{align}
\begin{aligned}
\varphi_0 &\defvar \frac{\alpha-1}{2-\alpha}\frac{\operatorname{ln}(2)}{2} \frac{1}{\gamma}, \\
\varphi_1 &\defvar \frac{\alpha-1}{2-\alpha}\frac{\operatorname{ln}(2)}{2} \frac{2\log(1+2d_A^{\CvsQ})}{\sqrt{\gamma}}, \\
\varphi_2 &\defvar \frac{\alpha-1}{2-\alpha}\frac{\operatorname{ln}(2)}{2} \left(2 + \left(\log(1+2d_A^{\CvsQ})\right)^2\right).
\end{aligned}
\end{align}
Hence in the following analysis for optimizing the choice of min-tradeoff function, we will use the expression~\eqref{eq:approxrate} as a loose substitute for~\eqref{eq:finalopt}. 

Our goal is now to find the choice of $g$ that maximizes~\eqref{eq:approxrate}. First observe that since the $\varphi_2$ term is independent of $g$, we can ignore it for the purposes of the subsequent analysis, and just aim to maximize $g(\qhon) - \newT(\mbf{g})$. We now introduce a new function $\newT^*$ (which is in fact the Legendre-Fenchel conjugate (see e.g.~\cite{BV04v8} Chapter~3.3) of $\newT$, as we prove in Appendix~\ref{app:modprimal}):
\begin{align}\label{eq:Tconj}
\newT^*(\lambdavec) \defvar 
\begin{cases} 
s({\norm{\lambdavec}_1}/{2})& \text{ if } \sum_c \lambda_c = 0 \\
+\infty & \text{ otherwise }
\end{cases}
, \quad \text{where }
s(x) \defvar
\begin{cases} 
\frac{(x-\varphi_1)^2}{4\varphi_0} & \text{ if } x \geq \varphi_1 \\
0 & \text{ otherwise }
\end{cases}
.
\end{align}
Note that $\newT^*$ is a convex function, for instance by observing that the $1$-norm is a convex function and $s$ is a nondecreasing convex function on $\mathbb{R}$ (alternatively, since it is a Legendre-Fenchel conjugate, it is automatically convex as it is a supremum of affine functions).
In terms of the above functions $\newT^*$ and $s$, we have the following result, whose proof we give in Appendix~\ref{app:modprimal}:
\begin{theorem}\label{th:modprimal}
Let $r_\mathrm{best}$ be defined as the value
\begin{align}\label{eq:bestapproxrate}
r_\mathrm{best} \defvar \sup_{g} \left(g(\qhon) - \newT(\mbf{g}) \right),
\end{align}
where the supremum is taken over all valid crossover min-tradeoff functions $g$. Then we have 
\begin{align}
r_\mathrm{best} = \;
\inf_{J,\lambdavec} &\quad W(\rhogJ) + \newT^*(\lambdavec) \nonumber\\
\suchthat
&\quad \qhon - \probst[\rhotJ] - \lambdavec = \mbf{0} \label{eq:modprimal1} \\
= \;
\inf_{J,\lambdavec,\boldsymbol{\svar}} &\quad W(\rhogJ) + s\!\left(\sum_c \svar_c/{2}\right) \nonumber\\
\suchthat
&\quad \qhon - \probst[\rhotJ] - \lambdavec = \mbf{0} \label{eq:modprimal2} \\
&\quad -\boldsymbol{\svar} \leq \lambdavec \leq \boldsymbol{\svar} , \nonumber
\end{align}
where the domain in each infimum is to be understood as the set of all Choi matrices $J$ and vectors $\lambdavec,\boldsymbol{\svar} \in \mathbb{R}^{\left|\Cnoperp\right|}$. Furthermore, suppose $L(\rhogJ,\boldsymbol{\svar})$ is an affine lower bound on $W(\rhogJ) + s\left(\sum_c \svar_c/{2}\right)$, and define
\begin{align}\label{eq:linmodprimal}
\begin{aligned}
r_\mathrm{SDP} \defvar \;
\inf_{J,\lambdavec,\boldsymbol{\svar}} &\quad L(\rhogJ,\boldsymbol{\svar}) \\
\suchthat
&\quad \qhon - \probst[\rhotJ] - \lambdavec = \mbf{0}\\
&\quad -\boldsymbol{\svar} \leq \lambdavec \leq \boldsymbol{\svar} .
\end{aligned}
\end{align}
Then the above optimization is an SDP where the optimal 
dual solution\footnote{Note that for equality constraints, there is an arbitrary sign convention to choose when defining the dual variables. This result is to be understood with respect to the sign convention we used, namely that defined in~\eqref{eq:maxminprob}--\eqref{eq:modprimal3} in Appendix~\ref{app:modprimal}.} to the constraint $\qhon - \probst[\rhotJ] - \lambdavec = \mbf{0}$ is a vector $\mbf{g}^\star \in \mathbb{R}^{\left|\Cnoperp\right|}$ with the following property: there exists a crossover min-tradeoff function $g^\star$ with gradient $\mbf{g}^\star$ such that
\begin{align}\label{eq:FWoptimality}
g^\star(\qhon) - \newT(\mbf{g}^\star) \geq r_\mathrm{SDP},
\end{align}
i.e.~it is a feasible point of the supremum in~\eqref{eq:bestapproxrate} that attains value at least $r_\mathrm{SDP}$.
\end{theorem}

The results in the above theorem can be interpreted as follows. We first remark that if we were only interested in the value $r_\mathrm{best}$, it could in principle be computed by solving the first optimization~\eqref{eq:modprimal1}, which is a convex optimization --- note that while the $\newT^*(\lambdavec)$ term in the objective is abstractly defined to have value $+\infty$ whenever $\sum_c \lambda_c \neq 0$, this is not a major obstacle since the $\qhon - \probst[\rhotJ] - \lambdavec = \mbf{0} $ constraint ensures we indeed always have (as long as all variables are correctly normalized)
\begin{align}\label{eq:lambdasum}
\sum_c \lambda_c = \left(\sum_c \qhonc\right) - \left(\sum_c \probt_c[\rhotJ]\right) = 1-1 = 0,
\end{align}
hence we can replace the $\newT^*(\lambdavec)$ term in the objective with $s({\norm{\lambdavec}_1}/{2})$, which is always finite.
However, even with this replacement, the objective is still not differentiable at some points (e.g.~for $\lambdavec$ where one component is zero and ${\norm{\lambdavec}_1}/{2} \geq \varphi_1$). This creates two difficulties: first, it is no longer straightforward to apply gradient-based algorithms such as Frank-Wolfe; second, even if we could find the optimal solution to that optimization, it is not obvious\footnote{A potential approach may have been to try solving the KKT conditions (\cite{BV04v8} Chapter~5.5.3) to obtain a dual optimal solution from a primal optimal solution. However, we were unable to find a method to solve the KKT conditions for this problem, again basically due to the fact that the objective is not differentiable.} 
how to use it to extract the optimal $\mbf{g}$ in~\eqref{eq:bestapproxrate}, which is the quantity we are actually interested in here.

Those difficulties are resolved by the subsequent properties listed in the theorem. First, the optimization~\eqref{eq:modprimal2} is a convex optimization where the objective function is indeed differentiable, allowing us to apply the Frank-Wolfe algorithm. Furthermore, the Frank-Wolfe algorithm inherently yields a sequence of affine lower bounds $L(\rhogJ,\boldsymbol{\svar})$ such that the corresponding $r_\mathrm{SDP}$ values converge towards $r_\mathrm{best}$. Therefore, the bound~\eqref{eq:FWoptimality} ensures that for the last SDP evaluated when the algorithm terminates (which should have $r_\mathrm{SDP}$ close to $r_\mathrm{best}$), we have the property that its dual solution yields a $\mbf{g}$ that achieves objective value at least $r_\mathrm{SDP}$ in the supremum in~\eqref{eq:bestapproxrate}, which is what we are interested in. 

Finally, we remark that when solving the optimization~\eqref{eq:modprimal2}, it is not necessary to explicitly impose all the constraints --- instead, it is usually more numerically stable to eliminate the $\lambdavec$ variables with the $\qhon - \probst[\rhotJ] - \lambdavec = \mbf{0}$ constraint, i.e.~by substituting $\lambdavec = \qhon - \probst[\rhotJ]$. With this, the optimization~\eqref{eq:modprimal2} simplifies to
\begin{align}
\begin{aligned}
\inf_{J,\boldsymbol{\svar}} &\quad W(\rhogJ) + s\!\left(\sum_c \svar_c/{2}\right) \\
\suchthat
&\quad -\boldsymbol{\svar} \leq \qhon - \probst[\rhotJ] \leq \boldsymbol{\svar} .
\end{aligned}
\end{align}
A solution to this optimization can easily be converted into one for~\eqref{eq:modprimal2} by reversing the $\lambdavec$ substitution. Also, applying the Frank-Wolfe algorithm to this optimization still allows us to obtain affine lower bounds $L(\rhogJ,\boldsymbol{\svar})$ with the desired properties.

\subsubsection{Completeness penalty term}\label{subsubsec:deltacomp}

Recall that given any $\mbf{f}$, $\phon$, and $S_\mathrm{acc}$, the formula~\eqref{eq:deltacomp} for $\Delta_\mathrm{com}$ is completely specified, and if $S_\mathrm{acc}$ is a convex polytope, then it can be securely computed as an LP. 
This indeed holds for $S_\mathrm{acc}$ of the form we focus on (see~\eqref{eq:acceptentrywise}).
Hence for any desired value of $\ecomAT$, we could (heuristically) find the choice of $S_\mathrm{acc}$ of that form that minimizes the completeness penalty $\Delta_\mathrm{com}$, simply by optimizing over the tolerance values $\tlow, \tupp \in \mathbb{R}_{\geq 0}^{|\mathcal{C}|}$ to minimize $\Delta_\mathrm{com}$ (evaluated as an LP) while subject to the constraint imposed by~\eqref{eq:ecomATbnd} (in which the binomial-distribution probabilities can be computed as a function of $\tlow, \tupp$).

We further note that in fact, for $S_\mathrm{acc}$ of the form~\eqref{eq:acceptentrywise}, we can straightforwardly compute the solution to the LP for $\Delta_\mathrm{com}$ by the following procedure, which is faster than running a generic LP solver. We first slightly simplify the LP, recalling that $\mbf{p}^\mathrm{acc},\mbf{p}^\mathrm{hon}$ are normalized distributions:\footnote{Technically this means we should also have additional constraints $\mbf{0} \leq \mbf{p}^\mathrm{acc} \leq \mbf{1}$ in the LP. However, observe that the optimal choices of $\tlow, \tupp$ for minimizing $\Delta_\mathrm{com}$ will always be such that this constraint is imposed, because if e.g.~we have $\phonc -  \tlowc < 0$ for some $c$, then for the purposes of minimizing $\Delta_\mathrm{com}$ we can always decrease $\tlowc$ ``for free'', in the sense that this tightens the LP constraints (thereby reducing $\Delta_\mathrm{com}$) while not affecting the $\ecomAT$ condition~\eqref{eq:ecomATbnd} (because the corresponding binomial-distribution term in that formula is constant for all $\tlowc$ values such that $\phonc -  \tlowc \leq 0$). Hence in the LP here, we do not separately impose these additional constraints. 
} 
\begin{align}
\sup_{\mbf{p}^\mathrm{acc}} &\quad -\mbf{f}\cdot(\mbf{p}^\mathrm{acc} - \phon) \nonumber\\
\suchthat
&\quad \phon - \tlow \leq \mbf{p}^\mathrm{acc} \leq \phon + \tupp \\
&\quad \sum_c p_c^\mathrm{acc} = 1 \nonumber \\
= \;
\sup_{\mbf{v}} &\quad \overline{\mbf{f}}\cdot\mbf{v} \nonumber\\
\suchthat
&\quad - \tlow \leq \mbf{v} \leq \tupp \\
&\quad \sum_c v_c = 0, \nonumber
\end{align}
where in the second line we introduce $\overline{\mbf{f}} \defvar -\mbf{f}$ to avoid inconvenient sign flips in the following explanation. 
This LP is feasible if and only if $\sum_c - \tlowc \leq 0 \leq \sum_c \tuppc$; 
recall that in our case we have chosen $\tlow, \tupp \in \mathbb{R}_{\geq 0}^{|\mathcal{C}|}$, so this trivially holds.
Given that the LP is feasible, our method to compute the optimal solution (or an optimal solution, if it is not unique) is basically to first initialize all $v_c$ values to their extremal values that maximize $\overline{\mbf{f}}\cdot\mbf{v}$, and then adjust the $v_c$ values iteratively until the $\sum_c v_c = 0$ constraint holds, starting with those that have the ``smallest effect'' on the objective. 

Explicitly, we first generate a ``candidate solution'' as follows:
letting $\mathcal{C}_{\geq 0}$ be the set of $c\in\mathcal{C}$ with $\overline{f}_c \geq 0$ and $\mathcal{C}_{<0}$ be the rest,
set $v_c = \tuppc$ for all $c\in\mathcal{C}_{\geq 0}$ and $v_c = \tlowc$ for all $c\in\mathcal{C}_{< 0}$. This candidate solution clearly maximizes $\overline{\mbf{f}}\cdot\mbf{v}$ subject to $- \tlow \leq \mbf{v} \leq \tupp$, but may not satisfy $\sum_c v_c = 0$ (if it does, then we are done). If $\sum_c v_c > 0$, then we find the value $c\in\mathcal{C}_{\geq 0}$ with the smallest value of $\overline{f}_c$ (possibly zero), and decrease the corresponding $v_c$ value continuously towards $-\tlowc$; if at some point this achieves $\sum_c v_c = 0$ then we stop and take this as the solution, and otherwise we iteratively repeat with the next-smallest $\overline{f}_c$ value over $c\in\mathcal{C}_{\geq 0}$ until we reach $\sum_c v_c = 0$ and stop. 
This will indeed terminate as long as the LP is feasible (i.e.~$\sum_c - \tlowc \leq 0 \leq \sum_c \tuppc$ as mentioned above), by observing that since we started from a candidate solution with $\sum_c v_c > 0$, this process must achieve $\sum_c v_c = 0$ at some point by continuity (since if all the $v_c$ values for $c\in\mathcal{C}_{\geq 0}$ were decreased to $-\tlowc$, we would have $\sum_c v_c = \sum_c - \tlowc \leq 0$).
Analogously, if instead the initial candidate solution had $\sum_c v_c < 0$, then we find the value $c\in\mathcal{C}_{< 0}$ with the smallest value of $|\overline{f}_c|$ and increase the corresponding $v_c$ value towards $\tuppc$, iterating until $\sum_c v_c = 0$. 
As a final check on this method, we verified the results obtained from it against the results returned by a standard LP solver for the final data points computed in this work (though not for the intermediate computations used to heuristically optimize the choices of $\tlow, \tupp$, to speed up that optimization).

\subsubsection{Other parameters and overall summary}\label{subsubsec:otherparams}
To further optimize the key length formula~\eqref{eq:keylength}, in this section, we optimize the choice of security parameters. More precisely, as mentioned earlier the security parameter $\esecure$ consists of the sum of two terms: $\esecure=\ePA+\eEV$. Based on the key length formula in Eq.~(\ref{eq:keylength}), let us define the following function:
\begin{align}
    \label{eq:epsilons}
    f_\alpha(
    \esecure,
    \ePA)\coloneqq\frac{\alpha}{\alpha-1}\log\left(\frac{1}{\ePA}\right)+\log\left(\frac{2}{\esecure-\ePA}\right).
\end{align}
Then by simple calculus (finding the choice of $\ePA \in (0,\esecure)$ that
minimizes the above value, thereby maximizing the finite-size key rate in Eq.~(\ref{eq:keylength})), we find the following optimal choice of secrecy and correctness parameters:
\begin{align}
    \label{eq:optimal_PA}
    \ePA=\left(\frac{\alpha}{2\alpha-1}\right)\esecure\ ,\qquad\qquad\eEV=\left(\frac{\alpha-1}{2\alpha-1}\right)\esecure.
\end{align}

In summary, our computational procedure to heuristically optimize the parameter choices (for some desired $\esecure$ and $\ecomAT$ values, recalling that we suppose  ``reasonable'' $\ecomEV$ values can be achieved by choosing $f_{\mathrm{EC}} = 1.16$ in~\eqref{eq:error_correction}) has the following logical structure, where we note that the computations in each function have no dependency on any parameters that have already been optimized out by a ``lower-level'' function:
\begin{itemize}
\item Write a ``base'' function that, given $\tlow, \tupp,
\gamma,\alpha,\mbf{g}$, computes the key rate according to~\eqref{eq:keylengthfinal}, by using the Frank-Wolfe algorithm to securely bound the infimum over $J$, and setting $\esecret,\eEV$ as a function of $\esecure$ according to~\eqref{eq:optimal_PA}. 
\item Write a function that (given $\gamma,\alpha,\mbf{g}$) finds the optimal $\tlow, \tupp$ for the preceding function subject to the desired $\ecomAT$ value, by applying the procedure in Sec.~\ref{subsubsec:deltacomp} 
(note that this procedure handles all dependencies on $\tlow, \tupp$).
\item Write a function that (given $\gamma,\alpha$) finds an approximately optimal choice of $\mbf{g}$ for the preceding function, 
by applying the procedure in Sec.~\ref{subsubsec:optfmin}. 
\item Write a function that finds an approximately optimal choice of $\gamma,\alpha$ for the preceding function, by simply evaluating it over a grid of $\gamma,\alpha$ values and taking the best result.
\end{itemize}
Technically, all of the above computations also require a specification of the honest device behaviour (in order to specify $\phon,H(S|YI)_\mathrm{hon}$ and the functions $W,\probst$), which may in turn contain some other parameters that could be optimized over --- for instance, the decoy-state intensities and distribution of intensity choices in Sec.~\ref{sec:Decoy-state with improved analysis} (though the simpler qubit BB84 protocol in Sec.~\ref{sec:Qubit BB84 with loss} does not have such parameters). For the purposes of this work, we simply set such parameters to the heuristically optimized values found in \cite{arx_KTL25} for the IID case (for each $n$ and each loss value).

\section{Qubit BB84 with loss}\label{sec:Qubit BB84 with loss}

\newcommand{\depol}{p^\mathrm{depol}_\mathrm{hon}}
\newcommand{\lossparam}{\zeta_\mathrm{hon}}

In this section we consider a version of the BB84 protocol \cite{BB84} using a perfect qubit source or equivalently a perfect single photon source. 
We take the following model for the honest behaviour, described by two parameters: a depolarizing-noise parameter $\depol \in [0,1]$, and a loss-in-decibels parameter $\lossparam \in [0,\infty)$. Namely, after Alice uses the source to prepare perfect qubit states (in any specified basis), they are first subject to a depolarizing channel
\begin{align}
\mathcal{E}_\mathrm{depol}[\rho] \defvar (1-\depol) \rho + \depol \frac{\id}{2}.
\end{align}
Then, with probability $1-10^{-\frac{\lossparam}{10 \mathrm{dB}}}$ the qubit is lost (in which case Bob always obtains a no-detection outcome), and otherwise Bob performs an ideal projective measurement on the qubit, in some specified basis. 

Let us first discuss the single-round steps in this protocol in more detail (for brevity, we again omit the subscripts specifying individual round numbers during this discussion). In each round Alice decides with probability \(\gamma\) whether it is a test or generation round. If it is a generation round, Alice always sends one of the states \(\ket{0}, \ket{1}\) in the \(Z\)-basis. Otherwise, she always sends \(\ket{+}, \ket{-}\) in the \(X\)-basis. Alice records her choice of basis and signal state in a classical register $X$ with alphabet $\{(\mathtt{Z},0), (\mathtt{Z},1), (\mathtt{X},0), (\mathtt{X},1)\}$.
Upon arrival, Bob measures the incoming states with probability \(\gamma\) in the \(X\)-basis and with \(1-\gamma\) in the \(Z\)-basis, recording the basis choice and outcome in a classical register $Y$ with alphabet $\{(\mathtt{Z},0), (\mathtt{Z},1), (\mathtt{Z},\text{\texttt{no-det}}), (\mathtt{X},0), (\mathtt{X},1), (\mathtt{Z},\text{\texttt{no-det}})\}$.

Then Alice announces a classical register \(C^A\) that is set to $\perp$ if it was a generation round, and otherwise set to the value of the $X$ register. Bob then announces a classical register \(C^B\) as follows: if it was a generation round (i.e.~$C^A=\perp$) \emph{and} he measured in the $Z$ basis, then he sets $C^B$ to either $(\mathtt{Z},\text{\texttt{det}})$ or $(\mathtt{Z},\text{\texttt{no-det}})$ depending on whether there was a detection; otherwise he sets $C^B = Y$. These form all the single-round public announcements in this protocol, i.e.~we have $I=(C^A,C^B)$ for the purposes of our protocol description.
With these announcements, Alice and Bob set $C=\perp$ if $C^A = \perp$, and otherwise set $C=(C^A, C^B)$ (note that this means all signals sent in the \(X\)-basis are used for testing, independent of Bob's basis choice). 
Alice then applies sifting: 
if the round was a generation round that was also measured in the $Z$ basis and successfully detected
(i.e.~$C^A = \perp$ and $C^B = (\mathtt{Z},\mathtt{det})$), she sets $S$ to be equal to the second entry of $X$ (i.e.~just the choice of eigenstate, not the basis choice), otherwise she sets $S$ to be a fixed value of $0$; with this the alphabet of $S$ is $\{0,1\}$ and so we have $\dim(S)=2$.
Note that since this sifting process is based only on the public announcements $I=(C^A,C^B)$, Bob knows which rounds have been sifted out, though for the purposes of the subsequent calculations this is simply accounted for by the inclusion of $I$ in the conditioning registers.

\newcommand{\pdetZZ}{{p^{\mathtt{det}|\mathtt{ZZ}}_\mathrm{hon}}} 
With this description, we can slightly simplify the computation of the single-round honest value $H(S|YI)_\mathrm{hon}$ that determines the error-correction term~\eqref{eq:error_correction}. Specifically, since $S$ is set to a deterministic value $0$ whenever the round is sifted out, we can write
\begin{align}\label{eq:simplifyEC}
    H(S|YI)_\mathrm{hon} = 
    (1-\gamma)^2 \pdetZZ H(S|Y; C^A = \perp \land\, C^B = (\mathtt{Z},\mathtt{det}))_\mathrm{hon},
\end{align}
where by $\pdetZZ$ we mean the single-round probability of a detection event conditioned on sending and measuring in the $Z$ basis (in the honest case), and we can omit the $I$ register on the right-hand-side because it takes a fixed value when conditioned on $C^A = \perp$ and $C^B = (\mathtt{Z},\mathtt{det})$. For the model considered here, we simply have $\pdetZZ = 10^{-\frac{\lossparam}{10 \mathrm{dB}}}$ since the loss probability is independent of Alice and Bob's basis choices.

Previously we kept the exact definition of \(W (\rhogJ)\) rather general because it depends on the exact protocol. We will apply the formalism of \cite{WLC18} to represent the protocol and its announcement structure with two maps, the \(\mathcal{G}\)- and \(\mathcal{Z}\)-map. The \(\mathcal{G}\)-map applies the key map and incorporates the announcements, while the \(\mathcal{Z}\)-map is a pinching channel which is needed for technical reasons. Additionally, we will make use of the simplifications to this formalism that were presented in \cite[Appendix A]{LUL19}.

With the above protocol description, the appropriate conditional entropy $H(S|IE)$ to consider here can be written in the following form:\footnote{The prefactor of \(\left(1- \gamma \right)^2\) is to account for the probability of Alice and Bob both measuring in the $Z$ basis; we separately account for the probability of detection in the $W$ map, since that probability has a dependence on the state $\rhogJ$.} 
\begin{equation}\label{eq:W qubit BB84}
    W (\rhogJ) := \left(1- \gamma \right)^2 D\left(\mathcal{G}(\rhogJ) || \mathcal{Z} \circ \mathcal{G}(\rhogJ) \right),
\end{equation}
where 
\begin{equation}
    \rho^{t/g}_{J} = \tr[A']{ \left( \id_A \otimes J \right) \left( \ketbra{\xi^{t/g}}{\xi^{t/g}}^{T_{A'}} \otimes \id_B \right) },
\end{equation}
and 
\begin{equation}
    \ket{\xi^{t/g}}_{AA'} = \frac{1}{\sqrt{2}} \left( \ket{00} + \ket{11} \right).
\end{equation}

After applying the condition on the generation rounds, the only remaining Kraus operator for the \(\mathcal{G}\)-map is
\begin{equation}\label{eq:Kraus Ops G map}
	\begin{aligned}
		K_Z = &\left[
		\begin{pmatrix} 1 \\ 0 \end{pmatrix}
		_S  \otimes \begin{pmatrix} 1 & 0\\ 0 & 0 \end{pmatrix}_A +
		\begin{pmatrix} 0 \\ 1 \end{pmatrix}
		_S \otimes  \begin{pmatrix} 0 & 0\\ 0 & 1 \end{pmatrix}_A \right] 
		\otimes 
		\begin{pmatrix} 1 & & \\ & 1 & \\& & 0 \end{pmatrix}
		_B \otimes 1_I
	\end{aligned},
\end{equation}
where \(1_I \) is just a scalar. Finally, the Kraus operators for the pinching map \(\mathcal{Z}\) are
\begin{equation}\label{eq:Kraus Ops Z map}
	\begin{aligned}
		Z_1 &=
		\begin{pmatrix} 1 & \\ & 0 \end{pmatrix}
		\otimes _{\id_{\dim_A} \times\dim_B}, \\
		Z_2 &=
		\begin{pmatrix} 0 & \\ & 1 \end{pmatrix}
		\otimes _{\id_{\dim_A} \times\dim_B}.
	\end{aligned}
\end{equation}

Now, only the derivatives with respect to the Choi state \(J\) are missing to implement the Frank-Wolfe algorithm, which we will state for convenience below. Therefore, let us define the completely positive map \(\chi_g(X)\) as
\begin{equation}
    \chi_g(X) := \tr[A']{ \left( \id_A \otimes X^{T_{A'}} \right) \left( \ketbra{\xi^g}{\xi^g} \otimes \id_B \right) },
\end{equation}
where the state \(\ketbra{\xi^g}{\xi^g}\) plays the role of a Choi state. Thus, for \(X=J\), we find \(\rhogJ = \chi_g(J)\) and we can rewrite \(W(\rhogJ)\) as
\begin{equation}
    W(\rhogJ) = \left(W \circ \chi_g\right) (J).
\end{equation}
Then, keeping in mind that \(W(\rhogJ) = \left(1-\gamma\right)^2 D\left( \mathcal{G}(\rhogJ) || \mathcal{Z}\circ \mathcal{G}(\rhogJ) \right)\) as defined in Eq.~\eqref{eq:W qubit BB84}, we can write the derivative of \(W \) with respect to \(J\) as
\begin{align}
    \nabla_J W(\rhogJ) = \left(\chi_g^{\dagger} \circ \mathcal{G}^{\dagger} \circ \log \circ\; \mathcal{G} \circ \chi_g \right) \left(J \right) - \left(\chi_g^{\dagger} \circ \mathcal{G}^{\dagger} \circ \log \circ \mathcal{Z} \circ \mathcal{G} \circ \chi_g \right) \left(J \right).
\end{align}
The gradient of the remaining terms can be written as
\begin{equation}
\begin{aligned}
    &\nabla_J \left(\mbf{g}\cdot\left(\qhon - \probst[\rhotJ]\right) -\frac{\alpha-1}{2-\alpha}\frac{\operatorname{ln}(2)}{2}\widetilde{V}\!\left(\probst[\rhotJ],\mbf{g}\right) \right) \\
    = - &\left( \sum_c B_c^T \left( \mbf{g}_c + \frac{\alpha-1}{2-\alpha} \frac{\ln(2)}{2} \frac{\partial{\tilde{V}}}{\partial{\mbf{q_J}_c}} \right) \right),
\end{aligned}
\end{equation}
where \(B_c = B_{ij}\) for \((i,j) \in \Cnoperp \) and \(\mbf{q_J}\) are defined via
\begin{align}
    B_{ij} &\defvar \tr[A]{\left( M_i^A \otimes \id_{A'} \otimes M_j^B \right) \left( \ketbra{\xi^t}{\xi^t}^{T_{A'}} \otimes \id_B \right) }, \\
    \mbf{q}_J &\defvar \probst[\rhotJ].
\end{align}

\subsection{Numerical results}

First, in Fig.~\ref{fig:Qubit BB84 iid comparison}, 
we present the finite-size key rates for the qubit BB84 protocol described above plotted against loss in \(\mathrm{dB}\). In order to represent a realistic setup, we chose a depolarization of $\depol =0.01$ for all signals and a security parameter of $\esecure=10^{-8}$, though for simplicity we first only show the unique-acceptance scenario (i.e.~a single-point acceptance set \(S_{\mathrm{acc}} = \{\mbf{\phon} \}\); see Remark~\ref{remark:unique}). The testing probability \(\gamma\) and the \Renyi\ parameter have been optimized for each data point. 
We compare our results with the key rates resulting from the unique-acceptance IID analysis as presented in Ref.~\cite{arx_KTL25}. Our key rates compared to that scenario are worse, but this is to be expected since this work covers a much more general attack, compared to the IID attack in Ref.~\cite{arx_KTL25}. (In principle, another method to obtain key rates secure against coherent attacks would be to apply the postselection technique~\cite{CKR09,Nahar2024PRXQuantum} to that IID analysis. However, we leave such an analysis for separate work, since here we focus on the GEAT.) We see that in the low-loss regime, we in fact obtain basically the same key rates as the IID scenario for all $n$ values. However, our key rates drop off significantly faster as the loss increases.
\begin{figure}[h]
    \centering
    \includegraphics[width=1\textwidth]{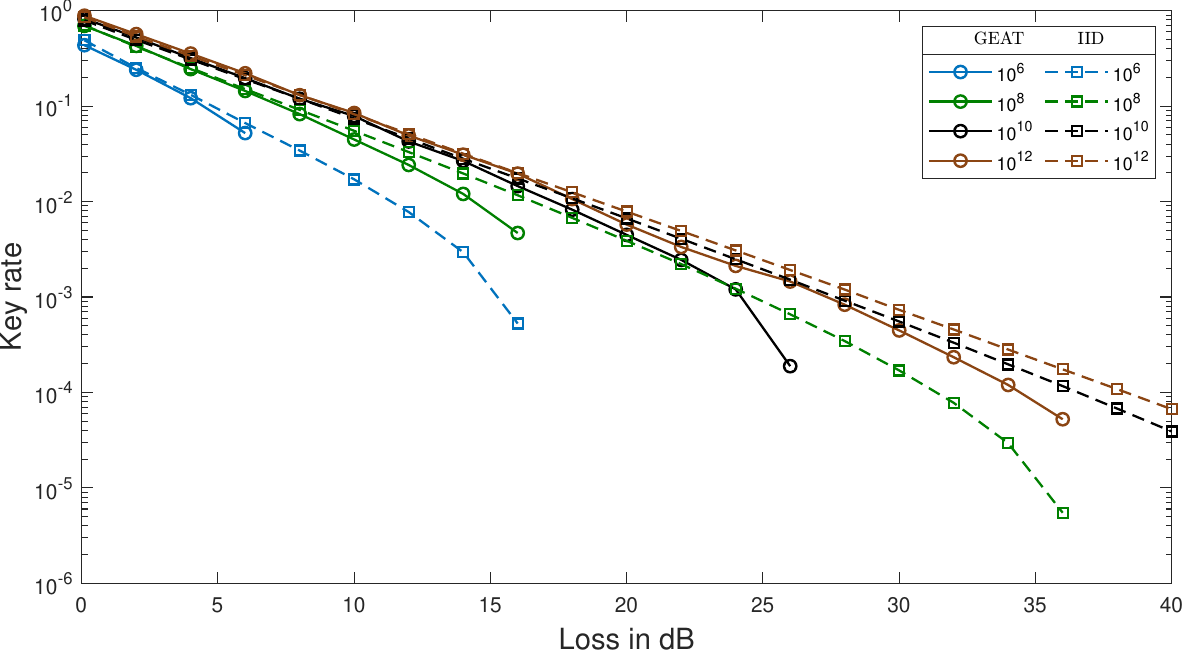}
    \caption{Key rate comparison for the qubit BB84 protocol with loss, between the GEAT analysis in this work (solid lines), and the IID setting (dashed lines). In all cases, the signals are depolarized with \(\depol = 0.01\), the security parameter is $\esecure=10^{-8}$, and for simplicity we show only the results for unique acceptance (see Remark~\ref{remark:unique}). The testing probability $\gamma$ and the \Renyi\ parameter $\alpha$ are optimized for the EAT analysis, and the former was also optimized in the IID analysis.}
    \label{fig:Qubit BB84 iid comparison}
\end{figure}

\begin{figure}
    \centering
    \includegraphics[width=1\textwidth]{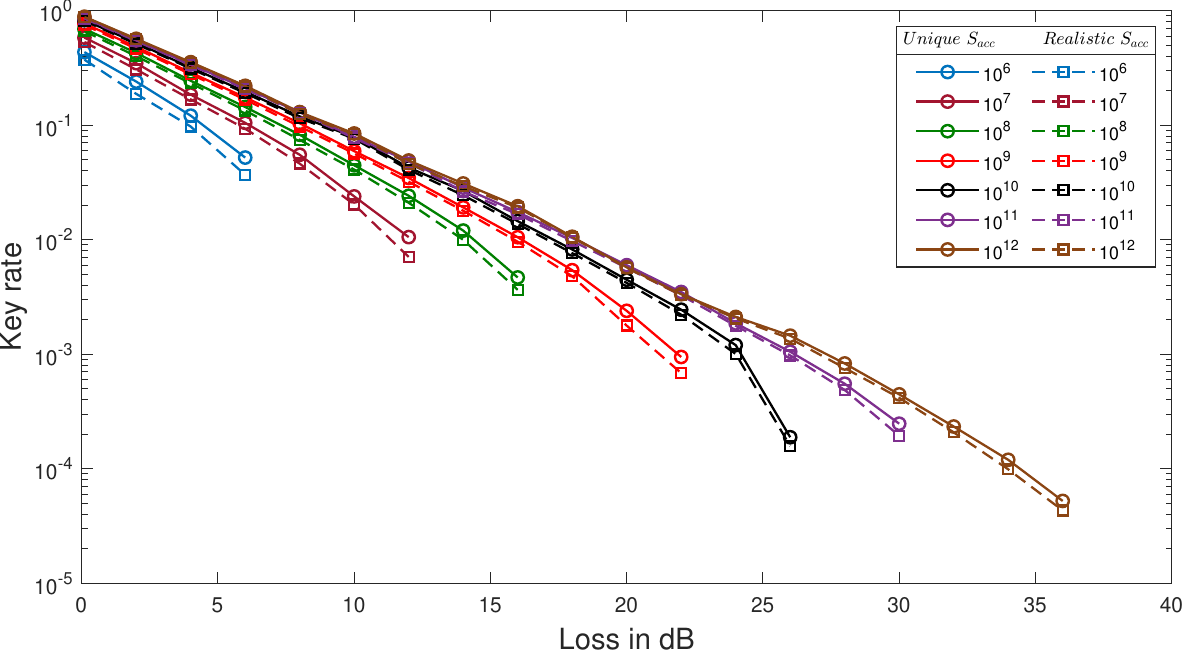}\\
    \includegraphics[width=1\textwidth]{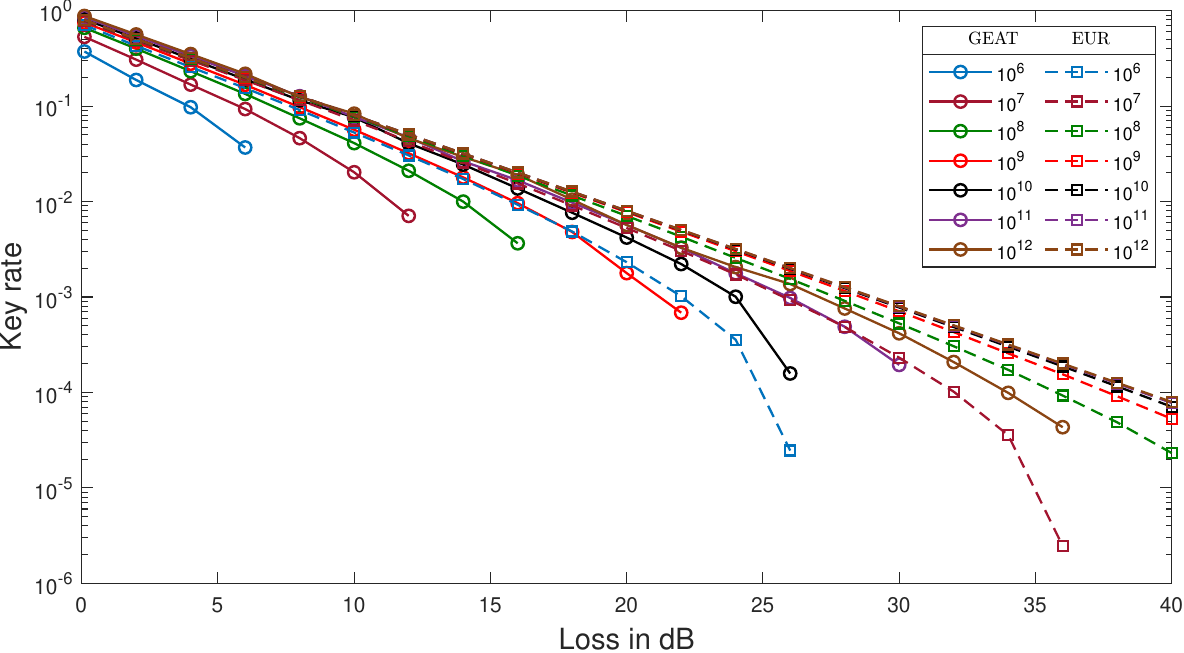}
    \caption{Key rates for the qubit BB84 protocol with loss. In the upper plot, the solid lines correspond to the unique-acceptance condition (see Remark~\ref{remark:unique}), whereas the dashed lines are for a realistic acceptance condition where there is a tolerance interval around the expected distribution as noted in Eq.~\eqref{eq:acceptentrywise}. In all cases, signals are depolarized with $\depol =0.01$; the security parameter is $\esecure=10^{-8}$; the completeness parameter for the realistic tolerance intervals is $\ecom=10^{-3}$; and the testing probability $\gamma$ and \Renyi\ parameter $\alpha$ are optimized for each data point. 
    In the lower plot, the solid lines show again our key rates for the realistic acceptance condition, while the dashed lines show the expected key rates obtained from a recent complementarity-based analysis in~\cite{arx_TNS+24} (incorporating recent improvements developed in~\cite{arx_MZC24}). The same $\depol,\esecure,\ecom$ values were used, and other parameters were again optimized. We find that for this specific protocol, our approach performs somewhat worse, though the difference becomes smaller at large $n$.
    }
    \label{fig:Qubit BB84 unique vs realistic acceptance}
\end{figure}

Next, in Fig.~\ref{fig:Qubit BB84 unique vs realistic acceptance}, we show a comparison between unique acceptance and a realistic acceptance set.
The latter was chosen such that it fulfils the completeness condition \(\Pr[\freq(c^n) \notin S_{\mathrm{acc}}  ] \leq \ecom\), where we chose \(\ecom = 10^{-3}\), as laid out in Sec.~\ref{subsec:completeness} and Sec.~\ref{subsubsec:deltacomp}.
One can see that to achieve positive key rates up to \(\unit[30]{dB}\), we require at least \(n=10^{11}\) signals sent. 
However, there is only a small difference in the key rates between the unique-acceptance case and the realistic acceptance set. In practical protocols, one would need to use the latter option, since otherwise the protocol would just almost always abort. Hence, the fact that we incur such a small penalty when making the protocol robust against statistical fluctuations is an important feature for real-world implementations.

Also in Fig.~\ref{fig:Qubit BB84 unique vs realistic acceptance}, we show a comparison of our results for the realistic acceptance set with the key rates obtained from a complementarity-based security proof, combining results from recent works~\cite{arx_TNS+24,arx_MZC24}. We find that the latter generally performs better than our approach for this protocol, though the difference becomes smaller at large $n$.

\section{Decoy-state with improved analysis}\label{sec:Decoy-state with improved analysis}

In this section we will use our methods presented in Sec.~\ref{sec:finitesize}, extend them to decoy-state protocols, and then apply the results to a decoy-state version of the BB84 protocol.

\subsection{Protocol details}\label{sec:Protocol Details Decoy}
We present a decoy-state version of the BB84 protocol which is effectively a decoy-state version of our qubit example in Sec.~\ref{sec:Qubit BB84 with loss} above. For this decoy-state protocol, Alice can choose between fully phase-randomized weak coherent pulses (WCP) with intensities \(\{\musig, \mu_2,\mu_3 \}\) to send her signals. The intensity \(\musig\) is the so-called signal intensity and will be used predominantly for key generation. Furthermore, we assume that the information is encoded in the polarization degree of freedom, i.e. Alice will send states from the set \(\{\rho_H, \rho_V , \rho_D,  \rho_A \}\), which are mixed because fully phase-randomized states are a 
classical mixture
of photon-number states.

There are a few additional differences to the qubit BB84 protocol presented above. 
In each round Alice still decides with probability \(\gamma\) whether it is a test or generation round. If it is a test round, Alice now selects an intensity \(\mu_i \in \{\musig, \mu_2,\mu_3 \} \) with some specified probability \(p(\mu_i|t)\), and sends a uniformly random choice out of \(\rho_D, \rho_A\) in the \(X\)-basis with that intensity. 
Otherwise, she uses the signal intensity \(\musig\) to send a uniformly random choice out of the states \(\rho_H, \rho_V\) in the \(Z\)-basis.
In either case, she uses a classical register $X$ to record her choice of intensity, basis, and signal state.\footnote{With the procedure specified here, this $X$ value is also sufficient to determine whether it is a test or generation round, just as in the previous qubit BB84 protocol (since Alice uses different basis choices in test or generation rounds). For more general protocols we may need to have the $X$ register also include a specification of whether it is a test or generation round.}

We consider an active measurement setup where Bob has control over a polarization-rotator, which then determines the basis to be measured in. This choice is just out of convenience for our optimization problem, to keep the dimensions of the involved Choi states small. There are no inherent issues with using a passive detection setup, unlike methods relying on the entropic uncertainty relation, e.g. \cite{LCW+14}.

In case of a passive detection setup, one would be required to use alternative squashing maps unless \(\gamma = 1/2\), since the squashing map presented in \cite{BML08,GBN+14} is only valid for passive setups with symmetric basis choices. There may be some technical points to handle when applying alternatives --- for example, the flag-state squasher of \cite{ZCW+21} (which allows for asymmetric basis choices of any detection setup) requires a subspace estimation\footnote{A generic method for this subspace estimation for passive linear optical detection setups was presented in \cite{KL24}.} technique, which we discuss how to handle within our framework in Appendix~\ref{app:squash} (Theorem~\ref{thrm:Squashed mod primal}). With this approach, using the flag-state squasher would result in one additional constraint in \eqref{eq:bestapproxrate decoy} corresponding to the subspace estimation. Furthermore, Bob's POVM elements would change due to his passive detection setup and by adding so-called flags necessary for the flag-state squasher. In summary, our framework can incorporate any passive linear optical detection setup, but for this work we focus on an active detection setup to keep the dimension of our optimization variables small.

Applying the active detection setup, Bob still measures the incoming states with probability \(\gamma\) in the \(X\)-basis and with \(1-\gamma\) in the \(Z\)-basis, recording his basis choice and outcome. Bob then applies the post-processing of \cite{BML08,GBN+14} for active BB84 detection setups to his outcomes. This converts his outcomes to an equivalent qubit detection scheme. Finally, Bob stores his post-processed outcomes in a classical register $Y$.

Then, the public announcements are also analogous to the previous qubit BB84 protocol. Alice first announces a classical register \(C^A\) that is set to $\perp$ if it was a generation round, and otherwise set to the value of the $X$ register. 
Bob then announces a classical register \(C^B\) as before: if it was a generation round (i.e.~$C^A=\perp$) \emph{and} he measured in the $Z$ basis, then he sets $C^B$ to either $(\mathtt{Z},\text{\texttt{det}})$ or $(\mathtt{Z},\text{\texttt{no-det}})$ depending on whether there was a detection; otherwise he sets $C^B = Y$. 
Just as in the qubit BB84 protocol, these form all the single-round public announcements, i.e.~we have $I=(C^A,C^B)$; similarly, with these announcements Alice and Bob set $C=\perp$ if $C^A = \perp$, and otherwise set $C=(C^A, C^B)$. 
Then, they apply an analogous sifting procedure: if the round was a generation round that was also measured in the $Z$ basis and successfully detected (i.e.~$C^A = \perp$ and $C^B = (\mathtt{Z},\mathtt{det})$), then Alice sets $S$ to be $0$ or $1$ depending on which $Z$-eigenstate she sent; otherwise she fixes $S=0$.

For this analysis, we model the honest behaviour following the description in~\cite{WL22}, which is specified by a misalignment parameter $\theta^{\text{misalign}}_\mathrm{hon}$ and a loss-in-decibels parameter $\lossparam$. Note that the latter is used to implement a beamsplitter model for loss from multi-photon states, rather than the single-photon loss model described previously in Sec.~\ref{sec:Qubit BB84 with loss}. With this, the overall probability of a round 
passing the sifting stage in the honest case is again \(\left(1- \gamma \right)^2 \pdetZZ \), though $\pdetZZ$ has to be computed under the beamsplitter model instead. 
The value of $H(S|YI)_\mathrm{hon}$ in the error-correction term can then be computed using the formula~\eqref{eq:simplifyEC} as before.

\subsection{General decoy formulation}

We now lay out the the theoretical foundations for bounding the entropy against Eve. 
For flexibility in potential applications, in this description we will consider a slightly more general scenario than that described in Sec.~\ref{sec:Protocol Details Decoy} above. Specifically, we shall allow Alice to choose from multiple intensities \(\musig, \mu_2, \dots \) in the generation rounds as well as the test rounds, according to some arbitrary distribution that may be different in the two cases.
The protocol we described above can be viewed as the special case where the distribution in the generation rounds is the trivial distribution that always uses a single intensity $\musig$.

When using decoy states from WCP sources, we can assume without loss of generality that Eve performs a QND measurement of the photon number first and then applies an attack based on the photon number \cite{LL20}. Thus, we can write Eve's attack channel \(\channE\) as a direct sum acting on each photon number separately, \(\channE = \bigoplus_{n=0}^{\infty} \channE_n \). The same holds true for the Choi states of the channel, i.e. \(J = \bigoplus_{n=0}^{\infty} J_n \). Therefore, the states \(\rhogJ\) conditioned on a generation round and \(\rho^{t,\mu}_J\) conditioned on a test round with intensity \(\mu\) satisfy
\begin{equation}
    \rhogJ = \sum_n p(n) \rho_{J_n}^{g}, \qquad 
    \rhotmu = \sum_n p_{\mu}(n) \rhotmun.
\end{equation}

Incorporating these properties of the channel and the states, we rewrite the crossover rate function as
\begin{align}\label{eq:crossover decoy infinite}
\begin{aligned}
r_\mathrm{cross}(\mbf{q}) =
\inf_{J = \bigoplus_{n=0}^{\infty} J_n} &\quad \sum_{n=0}^{\infty} p(n) W(\rho_{J_n}^{g}) \\ 
\suchthat &\quad 
\probst[\rhotJ] = \mbf{q}.
\end{aligned}
\end{align}
which still contains infinitely many terms in both the objective function and the constraints. However, we note that this optimization problem is already convex and does not require any relaxations in order to achieve convexity. Next, we aim to find a lower bound requiring only a finite number of optimization variables. We begin by first bounding the objective function, as follows.

For any protocol we can find a lower bound on the objective in the crossover rate function \eqref{eq:crossrateJ} by ignoring all contributions apart from those of single photons sent by Alice. In the case of the BB84 protocol this will be the main contribution, whereas vacuum states would only contribute on the order of dark counts. All higher photon numbers will not contribute to increasing the key rate because Eve could perform a PNS attack \cite{BBB+92,BLM+00}. Nevertheless, our technique allows for including higher photon numbers in principle, which could be beneficial for other protocols.

As the single-photon contribution to the objective function only depends on the Choi state $J_1$, from the above consideration we see the crossover rate function can be lower bounded via
\begin{align}
\begin{aligned}
r_\mathrm{cross}(\mbf{q}) \geq
\inf_{J = \bigoplus_{n=0}^{\infty} J_n} &\quad p(1) W(\rho_{J_1}^{g}) \\ 
\suchthat &\quad 
\probst[\rhotJ] = \mbf{q},
\end{aligned}
\end{align}
whereas if we were to consider additional photon numbers up to some cut-off \(N_c\), we would replace \(p(1) W(\rho_{J_1}^{g})\) with 
\begin{equation}\label{eq:multiphotonobjective}
    \sum_{n\leq N_c} p(n) W(\rho_{J_n}^{g}),
\end{equation}
and the optimization variables appearing in the objective function would include all Choi states up to \(N_c\).

Similarly to the objective function, one can exploit the block-diagonal structure of the Choi states for \(\probst[\rhotmu]\). Here we can equivalently write for each intensity \(\mu \)
\begin{equation}\label{eq:allChoiconstraints}
    \probst[\rhotmu] = \sum_{n=0}^{\infty} p_{\mu}(n) \probst[\rhotmun].
\end{equation}
For simplicity let us write \(\mbf{q}^{\mu}\) to denote the component of \(\mbf{q}\) corresponding to the intensity \(\mu\), i.e.~so we have \(\mbf{q} = \left(\mbf{q}^{\musig}, \mbf{q}^{\mu_2}, \dots \right)\) and \( \sum_c \mbf{q}^{\mu}_c = p(\mu |t)\).

Next, for each \(c=(a,b) \in \Cnoperp\) we can make the following rearrangements by writing \(\probst[\rhotmun]_{ab} = p(a,b|\mu,n) =: Y_n^{ab}\): 
\begin{equation}
    \mbf{q}^{\mu} =  p(\mu|t)\probst_{ab}[\rhotmu] =  p(\mu|t) \sum_{n=0}^{\infty} p_{\mu}(n) \probst_{ab}[\rhotmun] =  p(\mu|t) \sum_{n=0}^{\infty} p_{\mu}(n) Y_n^{ab},
\end{equation}
where we chose the definition of the \(n\)-photon yield \(Y_n^{ab} \) in line with the common one used for decoy-state protocols. These yields actually do not depend on the intensity \(\mu\), because based on the photon number Eve is unable to distinguish the intensities \cite{LL20}.

With this formulation, observe that for all \(n \neq 1 \), the only dependence of our optimization on the Choi states $J_n$ is via the corresponding yields \(\mbf{Y}_n\). Therefore, we can optimize over the yields \(\mbf{Y}_n\) in place of those Choi states. 
However, since this would still contain an infinite number of optimization variables, we introduce a photon number cut-off \(\Nph\) and characterize the remainder of the sum by \(\deltavec^{\mu}\). The remainder \(\deltavec^{\mu}\) cannot be arbitrarily large; it needs to satisfy for each intensity \(\mu\):
\begin{equation}
    0 \leq \delta_{ab}^{\mu} \leq 1- \sum_{n \leq \Nph} p_{\mu}(n) =: 1- \ptot(\mu) \quad \forall (a,b) \in \mathcal{C}.
\end{equation}
Therefore, for some \(\delta_{ab}^{\mu}\) satisfying the above constraint, we can write 
\begin{equation}\label{eq:probst per intensity}
    \probst_{ab}[\rhotmu] = \sum_{n=0}^{\Nph} p_{\mu}(n) Y_n^{ab} + \delta_{ab}^{\mu},
\end{equation}
for all intensities \(\mu\). Hence, we can recast the final optimization problem for the crossover rate function as
\begin{align}\label{eq:crossover rate decoy}
\begin{aligned}
r_\mathrm{cross}(\mbf{q}) \geq
\inf_{J_1, \mbf{Y}_0, \dots \mbf{Y}_{\Nph}, \deltavec^{\mu}} &\quad p(1) W(\rho_{J_1}^{g}) \\ 
\suchthat &\quad 
p(\mu|t)\left( \sum_{n \leq \Nph} p_{\mu}(n) \mbf{Y}_n+ \deltavec^{\mu} \right) = \mbf{q}^{\mu} \; \forall \mu, \\
&\quad  0 \leq \deltavec^{\mu} \leq 1- \ptot(\mu)\; \forall \mu, \\
&\quad  p(\mu|t) \mbf{Y}_1 = p(\mu|t) \probst[\rhotmun[1]]\; \forall \mu,\\
&\quad  \sum_b Y^{ab}_n = p(a|t,n) \; \forall a,b.
\end{aligned}
\end{align}

Before we continue to the decoy version of the optimization problem in Theorem~\ref{th:modprimal}, a few remarks about this crossover rate function are in line. As mentioned earlier, the crossover rate function yields the secret key rate for collective attacks, and hence can be used to compute valid key rates in the asymptotic limit as well. Thus, one can draw simple comparisons to previous decoy-state methods.

For example, in \cite{WL22}, asymptotic key rates were calculated for decoy-state protocols. In that work a two-step process was performed, where first the single-photon yields were bounded from above and below using an LP, and then these bounds were used to calculate the secret key rate. Additionally, in \cite{NUL23, KL24} improved methods for the decoy-state analysis of this two-step process were developed. Those works already used the Choi state to characterize the channel, but only for the purposes of bounding the single-photon yields, i.e.~the overall key rate calculation was still a two-step process. This was then extended to a finite-size analysis in~\cite{arx_KTL25}.

More specifically, in terms of our notation, the first step of the two-step process in those works consists of computing bounds that constrain the values \( \mbf{Y}_1 = \probst[\rhotmun[1]]\). Then, the second step consists of minimizing over the Choi state, subject only to those constraints. Qualitatively, the drawback of that approach is that it allows feasible points in the optimization that essentially correspond to Eve achieving the ``worst-case'' value on all yields simultaneously, which results in suboptimal bounds on the key rate.
In contrast, our methods here combine these two steps into a single one, which should provide better bounds than~\cite{WL22,NUL23,KL24,arx_KTL25}.

\begin{remark} \label{remark:twostep}
Note that the above claim can be rigorously justified as follows: every feasible point in our optimization in Eq.~\eqref{eq:crossover rate decoy} yields a feasible point in the two-step optimizations used in~\cite{WL22,NUL23,KL24,arx_KTL25}, but not vice versa --- for instance, in our optimization, there would be no feasible points that correspond to Eve simultaneously achieving the ``worst-case'' value on all yields. Since these optimizations are minimizations, this immediately implies on a mathematical level that applying our formulation would yield a higher optimal value as compared to those works. Still, it is true that this does not give an estimate of the extent by which it is higher, hence we present some empirical results in the next section (see the discussion of Fig.~\ref{fig:Decoy BB84 iid comparison}). We also note that the optimization we have presented is convex, and we find that the improved Frank-Wolfe methods we implemented converge in practice; similarly, the two-step procedures described in those works consist of an LP (which is rapidly solvable by standard algorithms) followed by a convex optimization that is similarly solved using Frank-Wolfe. 

We again clarify that the scope of this claim is not intended to encompass complementarity-based proof techniques, for which it is harder to provide a fully theoretical comparison. For such techniques, we provide an empirical comparison of the resulting key rates in Fig.~\ref{fig:Decoy BB84 unique vs realistic acceptance}.
\end{remark}

In fact, we emphasize that the optimization in Eq.~\eqref{eq:crossover decoy infinite} is \emph{exactly} a reparametrization of an optimization over Eve's single-round attack on the decoy-state protocol. More formally: there is a bijection between feasible points of the optimization and attacks Eve can physically perform, described by the Choi states. Hence in the context of the asymptotic key rates, that optimization would yield an \emph{exactly} tight bound on the Devetak-Winter asymptotic key rate formula, i.e.~there was no loss of tightness in constructing that optimization, because every feasible point corresponds to a valid attack that could actually be performed.

The only differences between that tight bound and the optimization that we numerically implement (Eq.~\eqref{eq:crossover rate decoy}) are the following. First, we instead have to set a finite value for the photon cut-off \(\Nph\), effectively grouping together all photon-number terms above the cutoff into a single term that might theoretically not correspond to a valid attack. Second, we have relaxed the objective function slightly in that we have excluded the vacuum-state entropy contribution.\footnote{While the objective function formula in Eq.~\eqref{eq:crossover rate decoy} also does not have multi-photon contributions, this is simply because these contributions are exactly zero in the case of decoy-state BB84 as previously discussed, and hence nothing is lost by excluding them.} Third, we relaxed some of the constraints slightly by replacing the multi-photon Choi states with yields, which loses some ``quantum correlations'' in those terms. It appears difficult to give a fully formal analysis of the effect of these relaxations --- however, on a purely informal basis, we remark that it seems unlikely that this introduces a significant suboptimality. This is because the higher-photon-number components usually have low weight in typical parameter regimes for decoy-state protocols, and the vacuum-state entropy contribution is only on the order of the dark counts, as mentioned previously. However, we do not intend this to be a formal claim, and perhaps future work could study in more detail the effect of this relaxation. Most importantly, we highlight that even the numerical implementation in Eq.~\eqref{eq:crossover rate decoy} can be made arbitrarily tight with respect to the Devetak-Winter formula, by including more Choi states.

Next we turn our attention to finding an optimal crossover min-tradeoff function as in Theorem~\ref{th:modprimal} (or Appendix~\ref{app:squash}, Theorem~\ref{thrm:Squashed mod primal} if squashing maps are required). We still define \(r_\mathrm{best}\) in the same way, but instead we use the modified crossover rate function of Eq.~\eqref{eq:crossover rate decoy} from above. After following the same steps as in Appendix~\ref{app:modprimal}, we find
\begin{align}\label{eq:bestapproxrate decoy}
\begin{aligned}
r_\mathrm{best} =
\inf_{\substack{J_1, \mbf{Y}_0, \dots \mbf{Y}_{\Nph}, \\\deltavec^{\mu} ,\boldsymbol{\svar}^{\mu}}} &\quad p(1) W(\rho_{J_1}^{g}) + s\!\left(\sum_c \svar_c/{2}\right) \\ 
\suchthat &\quad 
-\boldsymbol{\svar}^{\mu} \leq \mbf{q}^{\mu} - p(\mu|t)\left( \sum_{n \leq \Nph} p_{\mu}(n) \mbf{Y}_n + \deltavec^{\mu} \right) \leq \boldsymbol{\svar}^{\mu} \; \forall \mu , \\
&\quad  0 \leq \deltavec^{\mu} \leq 1- \ptot(\mu)\; \forall \mu, \\
&\quad  p(\mu|t) \mbf{Y}_1 = p(\mu|t) \probst[\rhotmun[1]]\; \forall \mu,\\ 
&\quad  \sum_b Y^{ab}_n = p(a|t,n) \; \forall a,b.
\end{aligned}
\end{align}

As before, the gradient of the crossover min-tradeoff function can be extracted as the dual variable to the constraint 
\(\qhon - \probst[\rhotJ] - \lambdavec = \mbf{0}\), where \(\qhon = \left(\mbf{q}^{\musig}, \dots \right)^T\), \(\probst[\rhotJ] = \left( p(\mu|t)\probst[\rhotmu] , \dots \right)^T\) and each \(\probst[\rhotmu]\) is identified by Eq.~\eqref{eq:probst per intensity}.
This constraint is, as expected, equivalent to the first constraint of our optimization problem for \(r_\mathrm{best}\):
\begin{equation}
    \mbf{0} = \mbf{q}^{\mu} - p(\mu|t)\left( \sum_{n \leq \Nph} p_{\mu}(n) \mbf{Y}_n + \deltavec^{\mu} \right) - \lambdavec^{\mu} \quad \forall \mu.
\end{equation}

After we find the gradient \(\mbf{g}\) of a crossover min-tradeoff function, we can again apply the key length formula from Eq.~\eqref{eq:keylength}:
\begin{align}
    l&\leq n\left(\inf_{\substack{J_1, \mbf{Y}_0, \dots \mbf{Y}_{\Nph}, \\\mbf{\delta}^{\mu} }} \left(p(1)W(\rho_{J_1}^{g}) + \mbf{g}\cdot\left(\qhon - \probst[\rhotJ]\right) -\frac{\alpha-1}{2-\alpha}\frac{\operatorname{ln}(2)}{2}\widetilde{V}\!\left(\probst[\rhotJ],\mbf{g}\right) \right) - \Delta_\mathrm{com} \right) \nonumber\\
    &\qquad -n\left(\frac{\alpha-1}{2-\alpha}\right)^2K(\alpha)-\lambdaEC-\frac{\alpha}{\alpha-1}\log\frac{1}{\esecret}+2, 
\end{align}
where now \(\qhon = \left(\mbf{q}^{\musig}, \dots \right)^T\), \(\probst[\rhotJ] = \left( p(\mu|t)\probst[\rhotmu] , \dots \right)^T\) and 
\begin{equation}
    p(\mu|t)\probst[\rhotmu] = p(\mu|t)\left( \sum_{n \leq \Nph} p_{\mu}(n) \mbf{Y}_n + \deltavec^{\mu} \right) 
\end{equation}
is a function of \( \mbf{Y}_0, \dots, \mbf{Y}_{\Nph}, \deltavec^{\mu}\). Thus, we have found a formulation in line with Sec.~\ref{sec:finitesize} and can apply those methods to calculate the finite-size secret key rates.
Again, if we were to include higher photon numbers, the term \(p(1)W(\rho_{J_1}^{g})\) will be replaced with \(\sum_n p(n)W(\rho_{J_n}^{g})\).

\subsection{Numerical results}
In this section we present our numerical results for the decoy-state BB84 protocol as described in Sec.~\ref{sec:Protocol Details Decoy}. Again, note that we only consider the single-photon contribution (any higher photon number will give zero key rate) as shown in the derivation of \eqref{eq:bestapproxrate decoy}.
Hence, \(W(\rho_{J_1}^{g})\) is equal to \(W(\rhogJ)\) in Eq.~\eqref{eq:W qubit BB84} of the qubit BB84 protocol. Therefore, also the Kraus operators for the \(\mathcal{G}\)-map of eq.~\eqref{eq:Kraus Ops G map} and the \(\mathcal{Z}\)-map of eq.~\eqref{eq:Kraus Ops Z map} apply for the decoy-state protocol as well.

\begin{figure}
    \centering
    \includegraphics[width=1\textwidth]{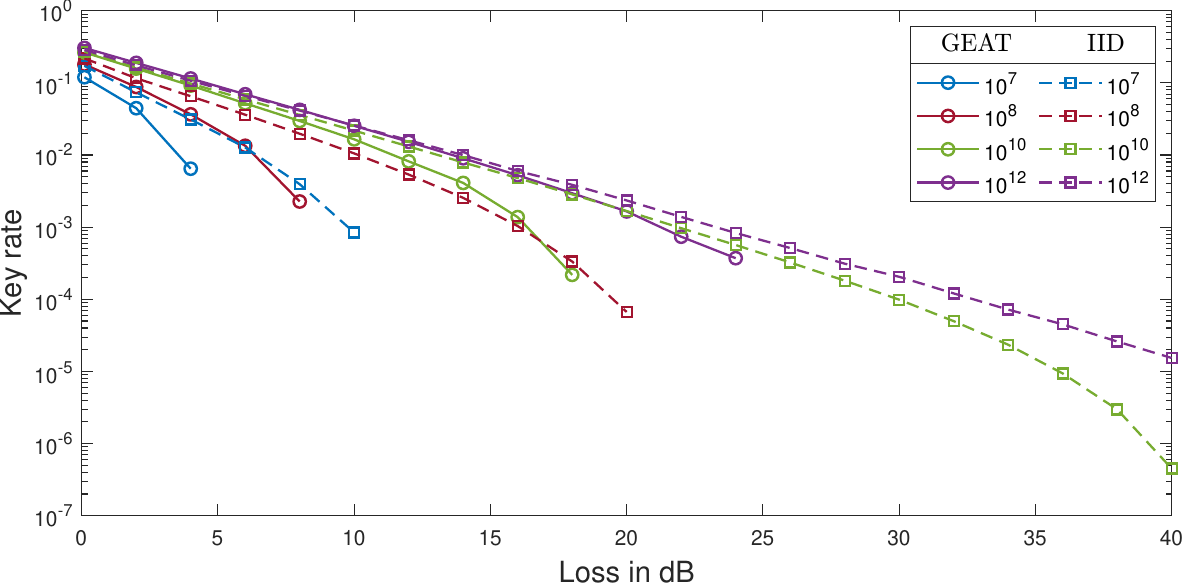}
    \caption{
    Key rate comparison for the decoy-state protocol, between the GEAT analysis in this work (solid lines), and the IID setting (dashed lines). 
    In all cases, fully phase-randomized weak coherent pulses with intensities $\{\musig = 0.9, \mu_2 = 2\times10^{-2},\mu_3 = 10^{-3}\}$ were used. Signals are misaligned with \(\theta^{\text{misalign}}_\mathrm{hon} = \asin(0.1)\), the photon number cut-off is $\Nph=10$, the security parameter is $\esecure=10^{-8}$,
    and for simplicity we show only the results for unique acceptance (see Remark~\ref{remark:unique}). The testing probability $\gamma$ and the \Renyi\ parameter $\alpha$ are optimized for the EAT analysis, and the former was also optimized for the IID analysis.}
    \label{fig:Decoy BB84 iid comparison}
\end{figure}

\begin{figure}
    \centering
    \includegraphics[width=.97\textwidth]{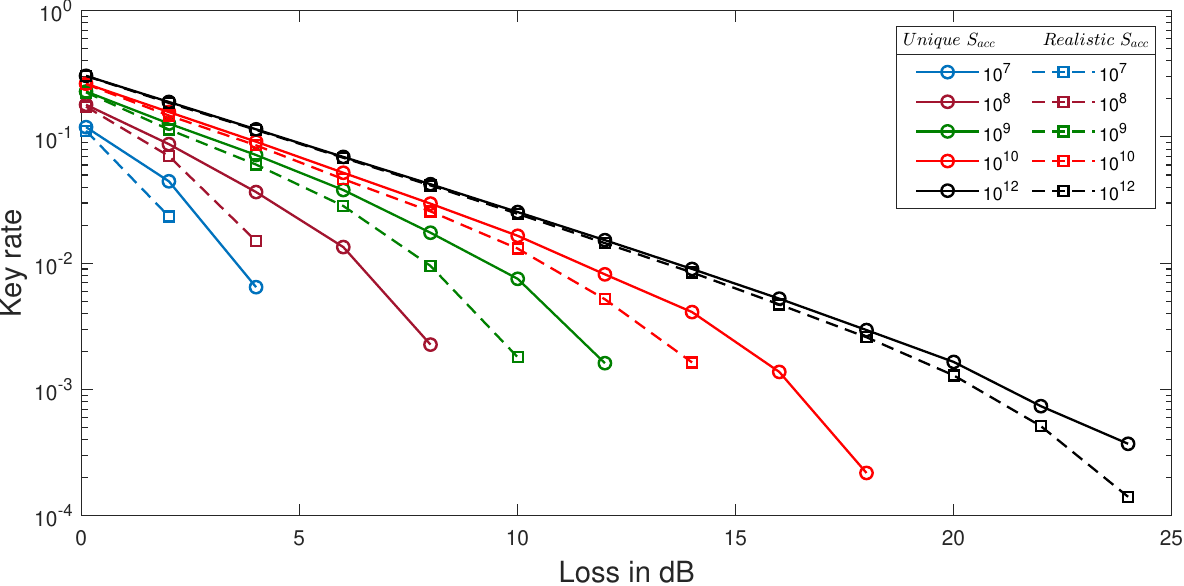}\\
    \includegraphics[width=.97\textwidth]{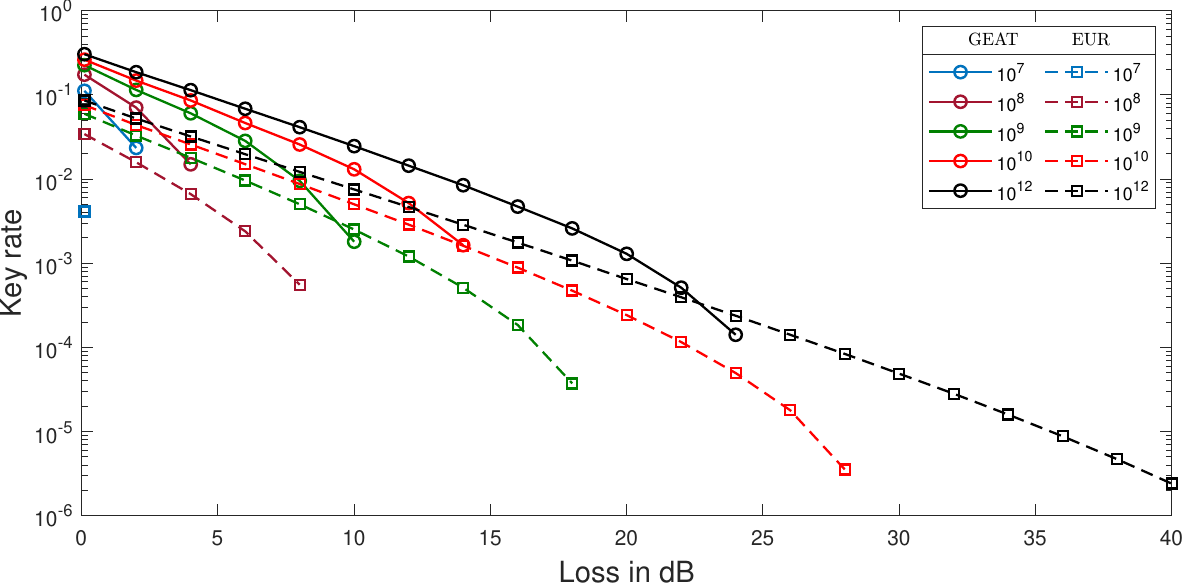}
    \caption{
    Key rates for the decoy-state protocol. In the upper plot, the solid lines correspond to the unique-acceptance condition (see Remark~\ref{remark:unique}), whereas the dashed lines are for a  realistic acceptance condition where there is a tolerance interval around the expected distribution as noted in Eq.~\eqref{eq:acceptentrywise}. 
    In all cases, two decoy signals are used with intensities $\mu_2=2\times 10^{-2}$ and $\mu_3=10^{-3}$, while the actual signal has the intensity of $\musig=0.9$, and the signals are chosen uniformly at random in each round. Signals are misaligned with \(\theta^{\text{misalign}}_\mathrm{hon} = \asin(0.1)\), the photon number cut-off is $\Nph=10$, the security parameter is chosen to be $\esecure=10^{-8}$, and the completeness parameter for the realistic tolerance intervals is $\ecom=10^{-3}$. The testing probability $\gamma$ and the \Renyi\ parameter $\alpha$ are optimized for each data point. From the plot, it can be seen that similar to the qubit BB84 protocol with loss in Sec.~\ref{sec:Qubit BB84 with loss}, as the number of signals is increased the difference between the two cases decreases. 
    In the lower plot, the solid lines show again our key rates for the realistic acceptance condition, while the dashed lines show the expected key rates obtained from a recent complementarity-based analysis in~\cite{arx_TNS+24} (incorporating recent improvements developed in~\cite{arx_MZC24}). The same $\theta^{\text{misalign}}_\mathrm{hon},\esecure,\ecom$ values were used, and other parameters were again optimized. We find that for this specific protocol, our approach appears to perform better at low loss values, while the complementarity-based approach appears to perform better at high loss values; however, these conclusions do not necessarily generalize to other protocols, for which similar comparisons would have to be performed on a case-by-case basis.
    }
    \label{fig:Decoy BB84 unique vs realistic acceptance}
\end{figure}

In Fig.~\ref{fig:Decoy BB84 iid comparison} we show the secret key rates for the decoy-state BB84 protocol, plotted against loss in \(\mathrm{dB}\). 
We suppose that the honest implementation is subject to misalignment with an angle of \(\theta^{\text{misalign}}_\mathrm{hon} = \asin(0.1) \approx 0.1002\) under the model described in \cite{WL22}, as mentioned previously. 
We chose the intensities as $\{\musig = 0.9, \mu_2 = 2\times 10^{-2},\mu_3 = 10^{-3}\}$, the photon number cut-off as $\Nph=10$, and the security parameter as $\esecure=10^{-8}$.  
For each data point we optimized the testing probability \(\gamma\) and the \Renyi\ parameter.

As in the case of the qubit BB84 protocol, the key rates resulting from our proof technique are mostly lower than the IID key rates from \cite{arx_KTL25} for the same protocol. Again, since we prove security against a wider class of attacks, one could expect such a behaviour.
We see that to reach positive key rates for losses up to \(\unit[25]{dB}\), we require about \(10^{12}\) signals sent.

We point out however that in the low-loss regime, we can approximately reach the asymptotic value with \(n=10^{12}\) signals, and for those data points it can be seen from the figure that we actually have a slight improvement over results from~\cite{arx_KTL25} for IID collective attacks (which was based on previous decoy methods). Thus, we confirm that our improved method can indeed improve the key rates as compared to the method used in~\cite{arx_KTL25}, although in the specific scenario we considered here, the improvement is small. We expect that if we were to compute key rates against IID collective attacks using this decoy method, we might be able to demonstrate a larger improvement (since the analysis in this work accounts for general attacks and hence has larger finite-size corrections as compared to~\cite{arx_KTL25}), but we leave a detailed comparison for future work, as our focus here is security against general attacks.  

Finally, in Fig.~\ref{fig:Decoy BB84 unique vs realistic acceptance}, we also present a comparison between unique acceptance and realistic acceptance sets for the decoy-state BB84 protocol. We note that in this case, the influence of a realistic acceptance set is much more pronounced than for the ideal qubit protocol. In particular, as the channel loss increases and the protocol approaches its maximum tolerable loss, the penalty from realistic acceptance sets becomes more significant. 
It may be of interest for future work to see if this can be avoided by applying adaptive key rate formulations such as in \cite{TTL24, inprep_HB25}.

Also in Fig.~\ref{fig:Decoy BB84 unique vs realistic acceptance}, we show a comparison of our results for the realistic acceptance set with the key rates obtained from a complementarity-based security proof, combining results from recent works~\cite{arx_TNS+24,arx_MZC24}. For this protocol, we find that the former appears to perform better in high-loss regimes, while the latter performs better in low-loss regimes. However, as complementarity-based proof techniques are quite different from the GEAT-based approach we used, it is difficult to draw broad conclusions from this about general protocols, and the question of which approach performs better for a given protocol would have to be addressed on a case-by-case basis.

\section{Conclusion}
\label{sec:conclusion}

In summary, in this work we have developed a flexible framework for security proofs of PM protocols against general attacks, with a particular focus on decoy-state protocols. To do so, we introduced  techniques for analyzing decoy-state protocols that are compatible with the GEAT, and have the further advantage that since they merge several steps that were handled separately in some previous works~\cite{WL22,NUL23,KL24,arx_KTL25}, they should yield better key rates than those works even in the asymptotic or IID scenarios (see Remark~\ref{remark:twostep}). Furthermore, regarding GEAT-specific contributions, we implemented a number of methods to improve the finite-size terms, including a method to optimize the choice of min-tradeoff function, and incorporating various improvements to the finite-size terms. By applying our framework to an example of a decoy-state protocol, we show that reasonably robust key rates can be achieved even in the finite-size regime.

We highlight that by using using the GEAT rather than the EAT in this work, we have also obtained an important advantage in that the resulting key rates have a reasonable level of loss tolerance, overcoming a difficulty noticed in the EAT-based analysis of~\cite{arx_GLvH+22}. More specifically, due to some technical issues regarding the EAT Markov conditions, the finite-size key rates computed from the EAT always had a \emph{subtractive} penalty on the order of the test-round probability~$\gamma$. 
This caused an issue that in any loss regime where the asymptotic key rate was of similar order of magnitude to~$\gamma$, it was not possible to obtain positive finite-size key rates, due to this subtractive penalty. In contrast, for a GEAT-based security proof, the effect of the test-round announcements is essentially just to rescale the first-order term by a \emph{multiplicative} penalty, which is much less significant. 

Finally, we note that as mentioned in the introduction, we can rely on recent works~\cite{arx_FKR+25,arx_AT25}\footnote{Note that the broad approach in those works was inspired by a work in preparation~\cite{inprep_HB25} presented at various conferences, but all proofs in those works are independent of~\cite{inprep_HB25}.} to state that the key rates computed in this work are also valid without restricting Eve to only interact with a single signal at a time. 
We now elaborate on this claim. Consider a PM protocol in which Alice sends out the signal states for the $n$ rounds, Eve interacts with them in some arbitrary fashion (without \emph{any} restriction on the number of signals she holds at a time), and then Bob measures the states as he receives them. Alice and Bob then compute the values $S_i C_i I_i$ for each round the same way as in our Protocol~\ref{Prot:PM Protocol} description. By the standard source-replacement analysis~\cite{BBM92,FL12}, the resulting state can be written in the form $\mathcal{M}^{\otimes n}(\rho_{A_1^n B_1^n \mathsf{E}})$, where $\rho_{A_1^n B_1^n \mathsf{E}}$ denotes a source-replaced quantum state before Alice and Bob's measurements (but after Eve has interacted arbitrarily with it to obtain some side-information $\mathsf{E}$), and $\mathcal{M}$ denotes a channel that (in each single round) performs Alice and Bob's measurements on the quantum registers $A_i B_i$ in that round, and processes the outcomes into the classical values $S_i C_i I_i$. 
Furthermore, the source-replacement analysis also implies that we have $\rho_{A_1^n} = \sigma_A^{\otimes n}$ for some fixed state $\sigma_A$. 
This state $\mathcal{M}^{\otimes n}(\rho_{A_1^n B_1^n \mathsf{E}})$ has a slightly different structure from the class of states that were considered in the GEAT (Theorem~\ref{thrm:EAT}), i.e.~those produced by a sequence of channels. However, from~\cite{arx_AT25} we have the following critical result (in fact, somewhat more general versions of it, where the channels can be different in each round):
\begin{theorem}[Special case of \cite{arx_AT25} Corollary~4.2, Eq.~(118)]
Consider any state of the form $\mathcal{M}^{\otimes n}(\rho_{A_1^n B_1^n \mathsf{E}})$ for some state $\rho_{A_1^n B_1^n \mathsf{E}}\in\mathcal{D}(A_1^n B_1^n \mathsf{E})$ satisfying $\rho_{A_1^n} = \sigma_A^{\otimes n}$ and some channel $\mathcal{M}: AB \to SCI$ with classical $C$. 
Let $f$ be an affine function satisfying
\begin{equation}
f(\mbf{p})\leq\inf_{\nu\in\Sigma(\mbf{p})}H(S|CIR)_\nu,
\end{equation}
where $\Sigma(\mbf{p})$ is the set of all states of the form $(\mathcal{M}\otimes\id_R)(\omega_{ABR})$ for some $\omega$ satisfying $\omega_A = \sigma_A$, such that the reduced state on the classical register $C$ has the same distribution as $\mbf{p}$. 
Let $\eps\in (0,1),\ \alpha\in(1,3/2)$, let 
$\Omega$ be an event on registers $C_1^n$, and let $h=\min_{c^n\in\Omega} f(\mathrm{freq}(c^n))$. Then: 
\begin{align}\label{eq:tensormodelbound}
    H^{\uparrow}_\alpha(S_1^n|C_1^n I_1^n \mathsf{E})&_{\mathcal{M}^{\otimes n}(\rho_{A_1^n B_1^n \mathsf{E}})_{|_{\Omega}}}
    \ge nh+nT_\alpha(f)-\frac{\alpha}{\alpha-1}\log \frac{1}{\pr
    {\Omega}}-n\left(\frac{\alpha-1}{2-\alpha}\right)^2 K(\alpha),
\end{align}
where $T_\alpha(f),\Var(\mbf{p},f),V(\mbf{p},f),K(\alpha)$ are defined the same way as in~\eqref{eq:2ndorderfuncs}.
\end{theorem}
In other words, this is precisely the statement that $\mathcal{M}^{\otimes n}(\rho_{A_1^n B_1^n \mathsf{E}})$ satisfies \emph{exactly the same} {\Renyi} entropy bound as in Theorem~\ref{thrm:EAT}, i.e.~the bound~\eqref{eq:EATbound}. Therefore, the key rates we computed in this work, based on that bound, also apply equally well for the state $\mathcal{M}^{\otimes n}(\rho_{A_1^n B_1^n \mathsf{E}})$, which corresponds to an \emph{arbitrary} coherent attack by Eve, without any restriction on the way she interacts with the signals.

In fact, the results in~\cite{arx_FKR+25,arx_AT25} also give sharper bounds (based on {\Renyi} entropies) than that shown in~\eqref{eq:tensormodelbound}, which should yield better key rates. We highlight that the techniques we developed here, especially for solving optimizations with a modified objective function as in Eq.~\eqref{eq:bestapproxrate decoy}, are also naturally compatible with the bounds derived in~\cite{arx_AT25}. 
Hence we aim to consolidate these approaches in future work. As a starting point, for readers already familiar with~\cite{arx_AT25}, we discuss in Appendix~\ref{app:Renyi} some further details relevant to analyzing optical QKD protocols with the proof techniques in that work; specifically, sharper conversions between {\Renyi} entropy and von Neumann entropy. 

\section*{Acknowledgements}
We thank Jie Lin for extensive assistance over the course of this work, particularly regarding the approach used in~\cite{arx_GLvH+22}.
We also thank Tony Metger and Martin Sandfuchs for helpful discussions.
L.K., A.A., N.L., and E.Y.Z.T.~conducted research at the Institute for Quantum Computing, at the University of Waterloo, which is supported by Innovation, Science, and Economic Development Canada. Support was also provided by NSERC under the Discovery Grants Program, Grant No. 341495. 
I.G. was supported by an Illinois Distinguished Fellowship and NSF Grant No. 2112890 during portions of this project.

\section*{Computational platform} 
The computations in this work were performed using the MATLAB package CVX~\cite{cvxpackage} with the solver MOSEK~\cite{mosek}. The code used to prepare the results in this paper will be available at \href{https://openqkdsecurity.wordpress.com/repositories-for-publications/}{https://openqkdsecurity.wordpress.com/repositories-for-publications/}.

\printbibliography

\appendix

\section{Deriving modified optimization}
\label{app:modprimal}

In this appendix, we give the proof of Theorem~\ref{th:modprimal}. We begin by proving the following lemma:
\begin{lemma}
For any constants $\varphi_0,\varphi_1 \geq 0$, the function 
$\newT$
defined in~\eqref{eq:approxrate}
is a convex function with Legendre-Fenchel conjugate $\newT^*$ given by~\eqref{eq:Tconj}, and satisfies 
\begin{align}\label{eq:Tenvelope}
\newT(\mbf{g}) = \newT^{**}(\mbf{g}) = \sup_{\lambdavec} \left( \mbf{g}\cdot\lambdavec - \newT^*(\lambdavec) \right).
\end{align}
\end{lemma}
\begin{proof}
We start the proof by showing that $\newT(\mbf{g})$ is a convex function. We first restate the definition of $\newT$ from~\eqref{eq:approxrate}:
\begin{equation*}
    \newT(\mbf{g}) \defvar \varphi_0(\left(\max(\mbf{g})-\min(\mbf{g})\right)^2 + \varphi_1\left(\max(\mbf{g})-\min(\mbf{g})\right).
\end{equation*}
First note that $\max(\mbf{g})$ is convex and $\min(\mbf{g})$ is concave. Thus, $\max(\mbf{g})-\min(\mbf{g})$ is convex; moreover, it is non-negative, therefore, its square is also a convex function. By noting that $\varphi_0,\varphi_1 \geq 0$, we conclude $\newT(\mbf{g})$ is a convex function. 

Now to prove its Legendre-Fenchel conjugate is the expression~\eqref{eq:Tconj}, we first write out the definition of the conjugate (see e.g.~\cite{BV04v8} Chapter~3.3):
\begin{equation}\label{eq:Tstar_sub}
    \newT^*(\lambdavec) = \sup_{\mbf{g}} \left( \mbf{g}\cdot\lambdavec -\varphi_0(\left(\max(\mbf{g})-\min(\mbf{g})\right)^2 - \varphi_1\left(\max(\mbf{g})-\min(\mbf{g})\right) \right).
\end{equation}
Let us first consider the case where the sum of the elements of $\lambdavec$ is zero, i.e., $\sum_c\lambda_c=0$. Note that for any vector $\mbf{g}$, we can define another vector $\Tilde{\mbf{g}}$, whose elements are given by:
\begin{equation}
    \label{eq:newg_element}
    \Tilde{g}_c=
    \begin{cases} 
\max(g)& \text{ if }\lambda_c\ge 0 \\
\min(g) & \text{ otherwise }
\end{cases}.
\end{equation}
It is not hard to see that for any vector $\mbf{g}$, the quantity within the supremum in~\eqref{eq:Tstar_sub} is upper bounded by the value with $\Tilde{\mbf{g}}$ in place of $\mbf{g}$. Furthermore, letting $\mathcal{C}_{\geq0}$ be the set of $c$ values with $\lambda_c\ge 0$, we have
\begin{align}
\Tilde{\mbf{g}}\cdot\lambdavec = \max(g)\sum_{c\in\mathcal{C}_{\geq0}}\lambda_c + \min(g)\sum_{c\notin\mathcal{C}_{\geq0}}\lambda_c 
= (\max(g) - \min(g)){\norm{\lambdavec}_1}/{2},
\end{align}
using the fact that $\sum_{c\in\mathcal{C}_{\geq0}}\lambda_c = -\sum_{c\notin\mathcal{C}_{\geq0}}\lambda_c = {\norm{\lambdavec}_1}/{2}$ since $\sum_c\lambda_c=0$. Therefore, the RHS of~\eqref{eq:Tstar_sub} simplifies to
\begin{align}\label{eq:Tstar_simp}
    \newT^*(\lambdavec) &= \sup_{\mbf{g}}\left( ({\norm{\lambdavec}_1}/{2}-\varphi_1)\left(\max(\mbf{g})-\min(\mbf{g})\right)-\varphi_0\left(\max(\mbf{g})-\min(\mbf{g})\right)^2\right)\cr
    &= \sup_{\beta\ge 0}\left( ({\norm{\lambdavec}_1}/{2}-\varphi_1)\beta-\varphi_0\beta^2\right),
\end{align}
where in the second line, we use the fact that the optimization only depends on the difference $\max(\mbf{g})-\min(\mbf{g})$. The optimization in~\eqref{eq:Tstar_simp} can be solved using simple calculus which leads to the following:
\begin{equation}
    \label{eq:Tstar_nontivial_ans}
    \newT^*(\lambdavec)=
    \begin{cases}
        \frac{({\norm{\lambdavec}_1}/{2}-\varphi_1)^2}{4\varphi_0}& \text{if} \norm{\lambdavec}_1\ge 2\varphi_1 \\
        0 & \text{if} \norm{\lambdavec}_1 < 2\varphi_1
    \end{cases}
\end{equation}
Let us now look at the situation where $\sum_c\lambda_c\neq 0$. Consider the particular choice $\mbf{g}=k\mbf{1}$ for some $k\in\mathbb{R}$. Taking the supremum over such choices of $\mbf{g}$, from~\eqref{eq:Tstar_sub} we have:
\begin{align}\label{eq:Tstar_trivial_ans}
    \newT^*(\lambdavec)&\geq\sup_{k}\left(k\sum_c\lambda_c\right)\cr
    &=+\infty,
\end{align}
where the second equality follows by setting $k\rightarrow +\infty$ (resp. $k\rightarrow -\infty$) if $\sum_c\lambda_c> 0$ (resp. $\sum_c\lambda_c< 0$). Combining~\eqref{eq:Tstar_nontivial_ans} with~\eqref{eq:Tstar_trivial_ans} results the expression in~\eqref{eq:Tconj}.

Finally, to show that~\eqref{eq:Tenvelope} holds, we simply use the fact that $\newT$ is a convex function with its domain being all of $\mathbb{R}^{|\mathcal{C}|-1}$. Therefore its epigraph is a convex closed set, and hence it is equal to the conjugate of its conjugate (see e.g.~\cite{BV04v8} Chapter 3.3.2), i.e.~we have $\newT(\mbf{g}) = \newT^{**}(\mbf{g}) = \sup_{\lambdavec} \left( \mbf{g}\cdot\lambdavec - \newT^*(\lambdavec) \right)$ as claimed.
\end{proof}

With this, we turn to proving the first equality~\eqref{eq:modprimal1} in Theorem~\ref{th:modprimal}. First, recall that as discussed in Sec.~\ref{subsec:ratefunctions}, any $\mbf{g}$ and $k_g$ satisfying the condition $\mbf{g}\cdot\mbf{q} + k_g \leq r_\mathrm{cross}(\mbf{q})$ (for all $\mbf{q}$) yields a valid crossover min-tradeoff function. This means that in fact any gradient vector $\mbf{g}$ corresponds to some valid choice(s) of crossover min-tradeoff function, simply by choosing $k_g$ to be some value satisfying that condition (such values always exist because $r_\mathrm{cross}(\mbf{q})$ is a non-negative convex function). Now observe that in the supremum over crossover min-tradeoff functions in~\eqref{eq:bestapproxrate}, for any fixed $\mbf{g}$, the best choice of $k_g$ would be the highest possible value satisfying that condition. Specifically, this implies we should choose it to be as follows, recalling that by definition an infimum is the highest possible lower bound on a set (see~\cite{TSB+22} Sec.~5 for another perspective based on Lagrange duals):
\begin{align}
k_g &= \inf_{\mbf{q}} \left(r_\mathrm{cross}(\mbf{q}) - \mbf{g}\cdot\mbf{q}\right) \nonumber\\
&= \inf_{J} \left(W(\rhogJ) - \mbf{g}\cdot\probst[\rhotJ]\right), \label{eq:optscalar}
\end{align}
where in the second line we have substituted in the definition of $r_\mathrm{cross}(\mbf{q})$ from~\eqref{eq:crossrateJ}, with $\rhogJ,\rhotJ$ being again understood as functions of $J$ via~\eqref{eq:rhotg_Choi}. Hence, we should always take the scalar term $k_g$ in the optimization~\eqref{eq:bestapproxrate} to be given by the above expression.

Substituting the formulas~\eqref{eq:Tenvelope} and~\eqref{eq:optscalar} into the argument of the supremum in~\eqref{eq:bestapproxrate}, it becomes
\begin{align}
& \mbf{g}\cdot\qhon + \inf_{J} \left( W(\rhogJ) - \mbf{g}\cdot\probst[\rhotJ] \right) - \sup_{\lambdavec} \left( \mbf{g}\cdot\lambdavec - \newT^*(\lambdavec) \right)  \nonumber\\
=\,& \inf_{(J,\lambdavec)\in\mathcal{D}} \left( W(\rhogJ) + \newT^*(\lambdavec)  + \mbf{g}\cdot\left(\qhon - \probst[\rhotJ] - \lambdavec\right) \right),
\end{align}
where $\mathcal{D}$ is a convex set defined as the set of tuples $(J,\lambdavec)$ such that $J$ is a Choi matrix and $\lambdavec \in \mathbb{R}^{|\mathcal{C}|-1}$ satisfies $\sum_c \lambda_c = 0$ (we can restrict the optimization over $\lambdavec$ to such values without loss of generality, because $\newT(\mbf{g})=+\infty$ whenever $\sum_c \lambda_c \neq 0$).
Our goal is to find the choice of $\mbf{g}$ that maximizes the above value, in other words to find the optimal solution $\mbf{g}$ to
\begin{align}\label{eq:maxminprob}
\sup_{\mbf{g}} \inf_{(J,\lambdavec)\in\mathcal{D}} \left( W(\rhogJ) + \newT^*(\lambdavec) + \mbf{g}\cdot\left(\qhon - \probst[\rhotJ] - \lambdavec\right) \right).
\end{align}
(Note that in the above, the objective function is finite everywhere over the domain $\mathcal{D}$, so we do not have issues with domain definitions.)
We now simply observe that the above problem is \emph{precisely} the Lagrange dual of the constrained convex optimization
\begin{align}\label{eq:modprimal3}
\begin{aligned}
\inf_{(J,\lambdavec)\in\mathcal{D}} &\quad W(\rhogJ) + \newT^*(\lambdavec) \\
\suchthat &\quad \qhon - \probst[\rhotJ] - \lambdavec = \mbf{0} ,
\end{aligned}
\end{align}
and furthermore this constrained optimization satisfies strict feasibility, for instance by setting $J=\id$ and $\lambdavec = \qhon - \probst[\rho^t_{\id}]$ 
(this point indeed lies in the relative interior of $\mathcal{D}$, because $J=\id$ is in the relative interior of the set of Choi matrices, and $\sum_c \lambda_c = 0$ by~\eqref{eq:lambdasum}). Hence by Slater's condition~(\cite{BV04v8} Chapter~5.2.3), we have strong duality, i.e.~the optimizations~\eqref{eq:maxminprob} and~\eqref{eq:modprimal3} have the same optimal value, and the optimal dual solution $\mbf{g}$ is attained. From~\eqref{eq:modprimal3}, we obtain the desired equality~\eqref{eq:modprimal1} by simply noting that the $\sum_c \lambda_c = 0$ restriction in the domain $\mathcal{D}$ is already enforced by the $\qhon - \probst[\rhotJ] - \lambdavec = \mbf{0}$ constraint (see~\eqref{eq:lambdasum}).

Turning to the next equality~\eqref{eq:modprimal2}, it is a straightforward transformation of the optimization problem that can be derived as follows. First observe that the optimization in~\eqref{eq:modprimal1} is clearly lower bounded by the optimization in~\eqref{eq:modprimal2} because every feasible point $(J,\lambdavec)$ in the former yields a feasible point $(J',\lambdavec',\boldsymbol{\svar}')$ of the latter with the same objective value, simply by taking $J'=J$, $\lambdavec'=\lambdavec$ and $\svar'_c = |\lambda_c|$. Conversely, the optimization in~\eqref{eq:modprimal2} is also lower bounded by the optimization in~\eqref{eq:modprimal1}, because every feasible point $(J,\lambdavec,\boldsymbol{\svar})$ in the former yields a feasible point $(J',\lambdavec')$ of the latter with an objective value that is no higher, by taking $J'=J$ and $\lambdavec'=\lambdavec$ (in which case we have $\sum_c \svar_c \geq 
\norm{\lambdavec'}_1$ and hence
$s\!\left(\sum_c \svar_c/{2}\right) \geq s\!\left(\norm{\lambdavec'}_1/2\right) = \newT^*(\lambdavec')$, since $s$ is a monotone increasing function and $\sum_c \lambda'_c = 0$ by~\eqref{eq:lambdasum}). Hence, the optimizations in~\eqref{eq:modprimal1} and~\eqref{eq:modprimal2} are equal.

Finally, we turn to proving the last bound~\eqref{eq:FWoptimality}. The intuition behind this property is that as mentioned above, the optimal dual solution to~\eqref{eq:modprimal3} is in fact the optimal choice of $\mbf{g}$ in~\eqref{eq:bestapproxrate}. However, since in arriving at the SDP~\eqref{eq:linmodprimal} we have implemented several transformations of the domain, objective and constraints in the ``fundamental'' optimization~\eqref{eq:modprimal3}, it is not immediately clear whether this property is inherited by the dual of the SDP~\eqref{eq:linmodprimal}. Hence, we shall instead directly prove the bound~\eqref{eq:FWoptimality} via an appropriate series of inequalities.

First note that the SDP~\eqref{eq:linmodprimal} is again strictly feasible (e.g.~following the above ideas, by choosing $J=\id$,  $\lambdavec = \qhon - \probst[\rho^t_{\id}]$ and any $\svar_c > |\lambda_c|$) and hence by Slater's condition, the dual is attained and has the same optimal value $r_\mathrm{SDP}$. By the definition of Lagrange duality, this means the optimal dual solution to the $\qhon - \probst[\rhotJ] - \lambdavec = \mbf{0}$ constraint in the SDP~\eqref{eq:linmodprimal} is a vector $\mbf{g}^\star \in \mathbb{R}^{\left|\Cnoperp\right|}$ such that\footnote{Here we are using the fact that SDP duals can be viewed as Lagrange duals. Strictly speaking, when writing an SDP dual, one would usually ``dualize'' all the constraints (including the implicit constraints enforcing that $J$ is a Choi matrix), whereas in the following expression we have only ``dualized'' the $\qhon - \probst[\rhotJ] - \lambdavec = \mbf{0}$ constraint and not the other contraints. However, it is not difficult to show that any dual feasible solution in the former version yields a dual feasible solution for the latter version that has at least the same value; in particular, an optimal dual solution in the former yields an optimal dual solution in the latter.}
\begin{align}\label{eq:SDPstrongdual}
\begin{aligned}
r_\mathrm{SDP} = 
\;
\inf_{(J,\lambdavec,\boldsymbol{\svar})\in\mathcal{D}'} &\quad L(\rhogJ,\boldsymbol{\svar}) + \mbf{g}^\star\cdot\left(\qhon - \probst[\rhotJ] - \lambdavec\right),
\end{aligned}
\end{align}
where for compactness we introduce the following notation: $\mathcal{D}'$ denotes the set of tuples $(J,\lambdavec,\boldsymbol{\svar})$ such that $J$ is a Choi matrix and $\lambdavec,\boldsymbol{\svar} \in \mathbb{R}^{|\mathcal{C}|-1}$ are vectors satisfying $-\boldsymbol{\svar} \leq \lambdavec \leq \boldsymbol{\svar}$ (i.e.~it incorporates the constraints other than $\qhon - \probst[\rhotJ] - \lambdavec = \mbf{0}$ in the SDP~\eqref{eq:linmodprimal}).

With this, note that given the above choice of gradient $\mbf{g}^\star$, if we choose the scalar term in $g^\star$ according to~\eqref{eq:optscalar} then the desired result~\eqref{eq:FWoptimality} holds:
\begin{align}
g^\star(\qhon) - \newT(\mbf{g}^\star) 
&= \mbf{g}^\star\cdot\qhon + \inf_{J} \left( W(\rhogJ) - \mbf{g}^\star\cdot\probst[\rhotJ] \right) - \sup_{\lambdavec} \left( \mbf{g}^\star\cdot\lambdavec - \newT^*(\lambdavec) \right)  \nonumber\\
&= \inf_{(J,\lambdavec)\in\mathcal{D}} \left( W(\rhogJ) + \newT^*(\lambdavec) + \mbf{g}^\star\cdot\left(\qhon - \probst[\rhotJ] - \lambdavec\right) \right) \nonumber\\
&= \inf_{(J,\lambdavec,\boldsymbol{\svar})\in\mathcal{D}'} \left( W(\rhogJ) + s\!\left(\sum_c \svar_c/{2}\right) + \mbf{g}^\star\cdot\left(\qhon - \probst[\rhotJ] - \lambdavec\right) \right) \nonumber\\
&\geq \inf_{(J,\lambdavec,\boldsymbol{\svar})\in\mathcal{D}'} \left( L(\rhogJ,\boldsymbol{\svar}) + \mbf{g}^\star\cdot\left(\qhon - \probst[\rhotJ] - \lambdavec\right) \right) \nonumber\\
&=r_\mathrm{SDP},
\end{align}
where each line is justified as follows. The first and second lines are obtained by simply substituting~\eqref{eq:Tenvelope} and~\eqref{eq:optscalar} and then simplifying the resulting expression, similar to before. The third line holds by the same arguments as in the above derivation of~\eqref{eq:modprimal2}. The fourth line holds because $L(\rhogJ,\boldsymbol{\svar})$ is a lower bound on $W(\rhogJ) + s\!\left(\sum_c \svar_c/{2}\right)$ by hypothesis. Finally, the fifth line is simply the equality~\eqref{eq:SDPstrongdual}.

\section{Min-tradeoff functions with squashing and source maps}
\label{app:squash}
For certain protocols, squashing \cite{BML08,GBN+14,ZCW+21} and source maps \cite{NUL23} allow for a much easier security proof. Therefore, in this section, we discuss how to incorporate these techniques into our framework. We only present squashing maps explicitly and comment on source maps later.

However, we phrase the following theorem rather abstractly in case there are other use cases besides squashing maps. Intuitively, in the following theorem the channel \(\Lambda\) should be thought of as a squashing map as in \cite{BML08,GBN+14} or \cite{ZCW+21}. Below the theorem we will give a detailed discussion on how it can be applied to squashing maps.

\begin{theorem}\label{thrm:Squashed mod primal}
    Let \( \Lambda: B \rightarrow B' \) be a CPTP map, and for any CPTP map $\mathcal{E}:A' \to B$, let \(J'\) denote the Choi state of the combined CPTP map \(\Lambda \circ \mathcal{E}: A' \rightarrow B' \). Let us define states \(\sigma_{J'}^g\) and \(\sigma_{J'}^t\) in terms of $J'$ as in Eq.~\eqref{eq:rhotg_Choi}, except with $J'$ in place of $J$. Moreover, suppose there exists a quantum-to-classical channel\footnote{Technically, we did not define \(\probst\) as a quantum-to-classical channel, but rather as a tuple of functions creating a list of probabilities. However, it can straightforwardly be viewed as a channel, by embedding the resulting classical probability distribution as a diagonal matrix.} \(\probst': AB'\rightarrow C\) such that \( \Phi_c'[\sigma_{J'}] = \Phi_c'\circ \Lambda[\rho_J] = \Phi_c [\rho_J]\) for all \(c \in \mathcal{C}\) and all attack channels \(\mathcal{E}\). Additionally, suppose there exist constants \(\theta_1,\theta_2\) and a subset \(\mathcal{F} \subset \mathcal{C}\) such that 
    \begin{equation}\label{eq:Squashing Constraint}
        \sum_{c\in\mathcal{F}} \Phi_{c}[\rho] \geq \theta_1\left( \theta_2 - \tr{\left(M^A \otimes \Pi^B \right) \rho} \right) \; \forall \rho \in \dop{=}(AB),
    \end{equation} 
    where \(\Pi^B\) is a projector onto a subspace in \(B\) which is invariant under \(\Lambda\), and \(M^A\) is one of Alice's POVM elements. Finally, let \(\hat{W}\) be a lower bound on \(W\) such that \(\hat{W}(\sigma^g_{J'}) \leq W(\rhogJ)\) for all attack channels \(\mathcal{E}\). Then, the key length in \eqref{eq:keylengthfinal} is lower bounded by
    \begin{align}
    &n\left(\inf_{J' \in \mathcal{D}'} \left(\hat{W}(\sigma_{J'}^g) + \mbf{g}\cdot\left(\qhon - \probst'[\sigma_{J'}^t]\right) -\frac{\alpha-1}{2-\alpha}\frac{\operatorname{ln}(2)}{2}\widetilde{V}\!\left(\probst'[\sigma_{J'}^t],\mbf{g}\right) \right) - \Delta_\mathrm{com} \right) \nonumber\\
    &\qquad -n\left(\frac{\alpha-1}{2-\alpha}\right)^2K(\alpha)-\lambdaEC-\ceil{\log\frac{1}{\eEV}}-\frac{\alpha}{\alpha-1}\log\frac{1}{\ePA}+2, \label{eq:keylengthfinalsquashed}
\end{align}
where the domain \(\mathcal{D}'\) is defined as
\begin{equation}
    \mathcal{D}' = \bigg\{ J' \in \Pos(AB') \bigg| \tr[B']{J'}=I_A, \; \sum_{c\in\mathcal{F}} \Phi_c'[\sigma_{J'}^t] \geq \theta_1\left( \theta_2 - \tr{\left(M^A \otimes \Pi^{B'} \right) \sigma_{J'}^t} \right) \bigg\},
\end{equation}
i.e.~it captures the constraint from equation \eqref{eq:Squashing Constraint} in addition to requiring that \(J'\) is a valid Choi state.
\end{theorem}

\begin{proof}
First, by assumption the entropy \(\hat{W}(\sigma_{J'}^g) = \hat{W}(\Lambda[\rhogJ]) \) is a lower bound on \(W(\rhogJ)\), for all attack channels \(\mathcal{E}\). Therefore, the conditional Renyi entropy \(H_\alpha^\uparrow(S^n|I^nE'_n)_{\rho_{|_{\Omega_{\mathrm{AT}}}}}\) can be lower bounded by
\begin{equation}\label{eq:Halpha squashed1}
\begin{aligned}
    &H_\alpha^\uparrow(S^n|I^nE'_n)_{\rho_{|_{\Omega_{\mathrm{AT}}}}} \\
    &\geq n\left(\inf_{J \in \tilde{\mathcal{D}}} \left(\hat{W}(\Lambda[\rhogJ]) + \mbf{g}\cdot\left(\qhon - \probst'\circ \Lambda[\rhotJ]\right) -\frac{\alpha-1}{2-\alpha}\frac{\operatorname{ln}(2)}{2}\widetilde{V}\!\left(\probst'\circ \Lambda[\rhotJ],\mbf{g}\right) \right) - \Delta_\mathrm{com} \right) \nonumber\\
    &\qquad -n\left(\frac{\alpha-1}{2-\alpha}\right)^2K(\alpha),
    \end{aligned}
\end{equation}
where the set \(\tilde{\mathcal{D}}\) is defined as
\begin{equation}
    \tilde{\mathcal{D}} = \bigg\{ J \in \Pos(AB) \bigg| \tr[B]{J}=I_A, \; \sum_{c\in\mathcal{F}} \Phi_c'\circ \Lambda[\rhotJ] \geq \theta_1\left( \theta_2 - \tr{\left(M^A \otimes \Pi^{B'} \right) \Lambda[\rhotJ]} \right) \bigg\}.
\end{equation}
This holds, since by assumption, every feasible point in Eq.~\eqref{eq:finalopt} is a feasible point of the optimization problem \eqref{eq:Halpha squashed1}. Hence, Eq.~\eqref{eq:Halpha squashed1} is indeed a lower bound on \(H_\alpha^\uparrow(S^n|I^nE'_n)_{\rho_{|_{\Omega_{\mathrm{AT}}}}}\).

Next, we combine the channels \(\Lambda\) and \(\mathcal{E}\), thus effectively giving the channel \(\Lambda\) into Eve's control. The new optimization variable is therefore \(J'\), the Choi state of the channel \(\Lambda \circ \mathcal{E}\), now creating states \(\sigma_{J'}^t\) and \(\sigma_{J'}^g\) on systems \(AB'\). The only term which could be affected by this change is the conditional entropy \(\hat{W}\), however including the map \(\Lambda\) in the minimization, can only reduce the optimum. Therefore, we find
\begin{equation}\label{eq:Halpha squashed2}
\begin{aligned}
    &H_\alpha^\uparrow(S^n|I^nE'_n)_{\rho_{|_{\Omega_{\mathrm{AT}}}}} \\
    &\geq n\left(\inf_{J' \in \mathcal{D}'} \left(W(\sigma_{J'}^g) + \mbf{g}\cdot\left(\qhon - \probst'[\sigma_{J'}^t]\right) -\frac{\alpha-1}{2-\alpha}\frac{\operatorname{ln}(2)}{2}\widetilde{V}\!\left(\probst'[\sigma_{J'}^t],\mbf{g}\right) \right) - \Delta_\mathrm{com} \right) \nonumber\\
    &\qquad -n\left(\frac{\alpha-1}{2-\alpha}\right)^2K(\alpha),
    \end{aligned}
\end{equation}
where the set \(\mathcal{D}'\) is now defined as
\begin{equation}
    \mathcal{D}' = \bigg\{ J' \in \Pos(AB') \bigg| \tr[B']{J'}=I_A, \; \sum_{c\in\mathcal{F}} \Phi_c'[\sigma_{J'}^t] \geq \theta_1\left( \theta_2 - \tr{\left(M^A \otimes \Pi^{B'} \right) \sigma_{J'}^t} \right) \bigg\}.
\end{equation}
Inserting this lower bound into Eq.~\eqref{eq:keylength} gives the theorem statement.
\end{proof}

Before we consider the application of Theorem \ref{thrm:Squashed mod primal} to squashing maps, let us first discuss its implications on the crossover min-tradeoff function. By applying the methods presented in Appendix~\ref{app:modprimal}, one can find a crossover min-tradeoff function which achieves
\begin{equation}
    g^\star(\qhon) - \newT(\mbf{g}^\star) \geq r_{\mathrm{SDP},\Lambda},
\end{equation}
where \(r_{\mathrm{SDP},\Lambda}\) is the version of Eq.~\eqref{eq:linmodprimal} after applying \(\Lambda\). Thus, the expected difference in performance between the original version and the min-tradeoff function resulting from the application of \(\Lambda\) are on the order of the difference in the asymptotic key rates before and after applying \(\Lambda\).

Now, let us describe the situation where \(\Lambda\) is a squashing map in more detail. In \cite{ZCW+21,L20}, it was shown that the conditional entropy \(H(S|IE)_{\rho}\) is lower bounded by the entropy \(H(S|IEE')_{\Lambda(\rho)}\) including the squashing map \(\Lambda\). Here, the register \(E'\) corresponds to the additional information Eve could gather by holding a purification of the squashed state. Hence, in the case of squashing maps the function \(\hat{W}\) in Theorem \ref{thrm:Squashed mod primal} is simply given by the conditional entropy of the squashed state.

Next, let us discuss two different kinds of squashing maps. First, if we consider the squashing maps from \cite{BML08,GBN+14}, then no additional constraints as in Eq.~\eqref{eq:Squashing Constraint} are required, which could equivalently be written as \(\theta_1=0\), enforcing only trivial constraints. However, much more importantly, this implies immediately, that one can use their squashing maps in our framework.

On the other hand, if the flag-state squasher from \cite{ZCW+21} is employed, constraints of the form as in Eq.~\eqref{eq:Squashing Constraint} are imperative. For example, consider the bounds on Bob's \(\leq N_B\)-photon subspace derived in \cite{KL24}. These state that the weight \(p(\leq N_B|x)\) inside Bob's \(\leq N_B\)-photon subspace is bounded as
\begin{equation}\label{eq:subspace bound}
    p(\leq N_B|x) \geq 1 - \frac{m_{\mathrm{mult}|x}}{c_{\geq N_B +1}},
\end{equation}
where the multi-click probability \(m_{\mathrm{mult}|x}\) given that Alice sent state \(x\) is given by
\begin{equation}
    m_{\mathrm{mult}|x} = \tr{ \left( M_x^A \otimes M_{\mathrm{mult}}^B \right) \rho} \; \forall \rho \in \dop{=}(AB),
\end{equation}
and \(c_{\geq N_B +1}\) is a constant determined by the optical detection setup. For a generic way of finding \(c_{\geq N_B +1}\) for arbitrary passive detection setups see \cite[Sec. VII]{KL24}.

The weight inside the Bob's \(\leq N_B\)-photon subspace can be equivalently written as 
\begin{equation}
    p(\leq N_B|x) = \tr{\left(M^A_x \otimes \Pi^B_{\leq N_B} \right) \rho} \; \forall \rho \in \dop{=}(AB).
\end{equation}
Hence, the bound in Eq.~\eqref{eq:subspace bound} can easily be recast in the form of Theorem \ref{thrm:Squashed mod primal} as
\begin{equation}
    \tr{ \left( M_x^A \otimes M_{\mathrm{mult}}^B \right) \rho} \geq c_{\geq N_B +1} \left( 1- \tr{\left(M^A_x \otimes \Pi^B_{\leq N_B} \right) \rho} \right) \forall \rho \in \dop{=}(AB).
\end{equation}
Therefore, the constants required are \(\theta_1 = c_{\geq N_B +1}\) and \(\theta_2 =1 \) and the set \(\mathcal{F}\) contains Alice's signal choice \(x\) and all of Bob's POVM elements included in the multi-click POVM element \(M_{\mathrm{mult}}^B\). Hence, one can also apply the flag-state squashing map in our framework. 

\begin{remark}
    Source maps as in \cite{NUL23} can be included in a similar manner as in Theorem~\ref{thrm:Squashed mod primal}, although one needs to reverse the order of \(\Lambda\) and \(\mathcal{E}\). Also, note that they are not required to produce any additional constraints as characterized by the set \(\mathcal{F}\). Thus, source maps can immediately be incorporated by applying \cite[Theorem 7]{NUL23}, which yields a lower bound on \(W\).
\end{remark}

\section{Improved conversions between {\Renyi} entropy and von Neumann entropy}
\label{app:Renyi}

The bounds derived in~\cite{arx_AT25} are similar in spirit to the GEAT in the sense that they involve the entropies of single rounds as discussed in Sec.~\ref{subsec:ratefunctions}. However, the relevant single-round quantity is instead the {\Renyi} entropy $H^\uparrow_\alpha(S|IE)$ (or in some cases, it may be convenient to relax it to $H_\alpha(S|IE)$ instead), rather than the von Neumann entropy $H(S|IE)$.\footnote{There are also a number of other improvements that significant sharpen the finite-size bounds and entirely avoid the separate optimization of min-tradeoff function choice, but we do not discuss those points further here.} In principle, to obtain the tightest finite-size key rates from that work, one should directly analyze those {\Renyi} entropies rather than $H(S|IE)$. However, since an extensive body of work has already been devoted to the latter~\cite{CML16,WLC18,WL22,NUL23,KL24} (including the approaches we used here for formalizing the $W$ function in Sec.~\ref{sec:Qubit BB84 with loss}--\ref{sec:Decoy-state with improved analysis}), it may still be useful to bound the former in terms of the latter, so that the approaches from those works can be applied. We now present some methods for doing so that may be particularly useful for optical QKD protocols --- in particular, some of the resulting finite-size corrections roughly depend on the  \emph{detected} rather than \emph{total} number of rounds, which can be a significant difference.

\begin{remark}
In some applications of entropy accumulation,
one may have to consider the entropies $H_\alpha(SC|IE)$ rather than just $H_\alpha(S|IE)$ (recall that $C$ is the register containing data for the acceptance test) --- the former is potentially relevant for protocols where the $C$ registers might not satisfy the Markov or non-signaling conditions of the EAT or GEAT respectively; refer to e.g.~\cite{ARV19} for further discussion. The bounds presented below generalize straightforwardly to that scenario, except they might depend on $\dim(SC)$ rather than $\dim(S)$. We stress however that a useful proof tactic is that if $C$ can be ``projectively reconstructed'' from $SIE$ in the sense of~\cite[Lemma~B.7]{DFR20}, then that lemma gives us 
\begin{align}
H_\alpha(SC|IE) = H_\alpha(S|IE) \quad\text{and}\quad H(SC|IE) = H(S|IE),
\end{align}
so in such scenarios we can freely add or remove $C$ from the left-hand-side conditioning registers using these relations, allowing the following bounds to be applied directly.
\end{remark}

To begin, a known simple bound in terms of $H(S|IE)$ is~\cite[Lemma~B.9]{DFR20}:
\begin{align}\label{eq:tovN}
H^\uparrow_\alpha(S|IE) \geq H_\alpha(S|IE) \geq H(S|IE) - (\alpha-1)\log^2\left(1+2 \dim(S) \right),
\end{align}
for any $\alpha \in \left(1,1+\frac{1}{\log\left(1+2\dim(S)\right)}\right)$. A more elaborate bound was derived in~\cite[Corollary~IV.2 with Corollary III.5]{DF19}, of the form
\begin{align}\label{eq:tovN2}
H^\uparrow_\alpha(S|IE) \geq H_\alpha(S|IE) \geq H(S|IE) - (\alpha-1)\frac{\ln 2}{2}\log^2\left(1+2\dim(S)^\CvsQ \right) -  \widetilde{K}(\alpha)(\alpha-1)^2,
\end{align}
for any $\alpha\in(1,2)$, where $\CvsQ=1$ if the $S$ system is classical and $\CvsQ=2$ if it is quantum (for QKD we would of course focus on the former case). The quantity $\widetilde{K}(\alpha)$ is a somewhat complicated expression similar to the $K(\alpha)$ term we presented in Eq.~\eqref{eq:2ndorderfuncs}, but it can be explicitly upper bounded.\footnote{While it involves various other {\Renyi} entropies, one can simply apply crude bounds on those entropies in terms of the system dimensions without too much loss of tightness, because in any case it only affects a term of order $O((\alpha-1)^2)$.}
By comparing the prefactors on the $(\alpha-1)$ terms, we see that for classical $S$, the bound~\eqref{eq:tovN2} is substantially better than~\eqref{eq:tovN} whenever $\alpha$ is close enough to $1$ for the $\widetilde{K}(\alpha)(\alpha-1)^2$ term to be negligible. (If $S$ is quantum, it can still be better if $\dim(S)$ is small.) 

Bounds of essentially the above form were often implicitly used in many theorems involving {\Renyi} entropies, including the proofs of the EAT and GEAT bounds in~\cite{DFR20,DF19,Metger2022Mar} (albeit in a fairly complicated fashion). 
However, we now show that the bounds can be substantially refined for applications in optical QKD. We begin by stating the following lemma:
\begin{lemma}\label{lemma:approxlin}
Consider any $w\in[0,1]$ and any states $\nu,\nu',\nu''\in\dop{=}(CQ)$ with classical $C$, such that $\nu_{CQ} = w \nu'_{CQ} + (1-w) \nu''_{CQ}$.\footnote{We do not require $\nu',\nu''$ to have disjoint supports, so for instance we do not need an explicit classical register in $\nu$ indicating whether $\nu'$ or $\nu''$ ``occurred''.} 
Let $\Delta_\alpha$ denote
\begin{align}
\Delta_\alpha \defvar \frac{\alpha-1}{\alpha} \frac{\ln2}{2} \log\dim(C).
\end{align}
Then for all $\alpha \in (1,\infty)$ we have:
\begin{align}\label{eq:approxlin}
\begin{gathered}
H^\uparrow_\alpha(C|Q)_{\nu} \geq (1-\Delta_\alpha) wH^\uparrow_\alpha(C|Q)_{\nu'}, \\ H_\alpha(C|Q)_{\nu} \geq (1-\alpha\Delta_\alpha) wH_\alpha(C|Q)_{\nu'},
\end{gathered}
\end{align}
and for all $\alpha \in \left(1,1+\frac{2}{
\ln\dim(C) 
}\right]$ we have:
\begin{align} \label{eq:approxlin2}
\begin{gathered}
H^\uparrow_\alpha(C|Q)_{\nu} 
\geq \left(1 - (1-w)\Delta_\alpha - 2{w^2}\Delta_\alpha^2 \right) wH^\uparrow_\alpha(C|Q)_{\nu'} , \\
H_\alpha(C|Q)_{\nu} 
\geq \left(1 - (1-w)\alpha\Delta_\alpha - 2{w^2}\alpha^2\Delta_\alpha^2 \right) wH_\alpha(C|Q)_{\nu'} .
\end{gathered}
\end{align}
\end{lemma}

We defer the proof to the end of this appendix, first discussing its implications and applications here. We focus mainly on the first bound~\eqref{eq:approxlin} as it is somewhat simpler. The point of it is that in QKD protocols, we often have the property that any state $\nu_{SIE}$ that can be produced in a single round (after the measurements, announcements and sifting) indeed has the form $\nu = w \nu' + (1-w) \nu''$, where furthermore $\nu''$ is such that we cannot get lower bounds on $H_\alpha(S|IE)_{\nu''}$ other than the trivial lower bound of zero. As a concrete example, in the decoy-state analysis in Sec.~\ref{sec:Decoy-state with improved analysis}, $\nu''$ could be the state conditioned on the event\footnote{This is a fairly elaborate event, so in potential future applications, one might want to instead just focus on a subset of the conditions in this event for simplicity, though the resulting bounds would be less tight.}
\begin{align}\label{eq:exampleevent}
\text{[more than $1$ photon was sent] or [the round was sifted out]},
\end{align}
since when this event happens, either Eve can learn $S$ perfectly or the value of $S$ is set to a deterministic value. Note that the value of $w$ can be ``variable'' (e.g.~it depends on Eve's attack) rather than a fixed value; for instance, this is indeed the case in the above example, because the sifting procedure involves checking whether Bob received a detection, and the probability of this can vary depending on Eve's attack.

For such a state (with classical $S$), its von Neumann entropy would satisfy 
\begin{align}
H(S|IE)_{\nu} \geq w H(S|IE)_{\nu'},
\end{align}
by concavity of von Neumann entropy, with equality whenever $\nu',\nu''$ have disjoint supports and $H_\alpha(S|IE)_{\nu''} = 0$. Most numerical approaches for studying such states, including those in this work, would implicitly use this (or similar ideas) to replace $H(S|IE)_{\nu}$ with the lower bound $w H(S|IE)_{\nu'}$ and study the latter instead. Now observe that the bound~\eqref{eq:approxlin} states that we have an approximate analogous result $H^\uparrow_\alpha(S|IE)_{\nu} \gtrsim wH^\uparrow_\alpha(S|IE)_{\nu'}$ for {\Renyi} entropy, up to a prefactor that approaches $1$ as $\alpha\to1$. 
Consequently, if we \emph{now} bound $H^\uparrow_\alpha(S|IE)_{\nu'}$ using either~\eqref{eq:tovN} or~\eqref{eq:tovN2},
we conclude that the following bounds hold (for the same ranges of $\alpha$ as in~\eqref{eq:tovN}--\eqref{eq:tovN2}):
\begin{gather}
H^\uparrow_\alpha(S|IE)_\nu \geq \left(1 - \Delta_\alpha \right) \left(wH(S|IE)_{\nu'} - w(\alpha-1)\log^2\left(1+2 \dim(S) \right) \right), \label{eq:newtovN} \\
H^\uparrow_\alpha(S|IE)_\nu \geq \left(1 - \Delta_\alpha \right) \left(wH(S|IE)_{\nu'} - w(\alpha-1)\frac{\ln 2}{2}\log^2\left(1+2\dim(S) \right) -  w\widetilde{K}(\alpha)(\alpha-1)^2 \right) \label{eq:newtovN2},
\end{gather}
where in the second case we have set $\CvsQ=1$ since $S$ is classical. Note that the prefactors of $1-\Delta_\alpha$ are just constants, and thus these lower bounds should still be numerically tractable under the methods used in this work --- recall that these numerical approaches were likely to be analyzing $w H(S|IE)_{\nu'}$ in place of $H(S|IE)_{\nu}$ anyway, and the other terms in these bounds are linear in $w$ (and therefore also usually linear in the optimization variables in the relevant parametrizations, thereby retaining convexity). In fact, for some simpler use cases, the value $w$ is in fact a constant in the optimization, making this even easier to analyze --- for instance, this holds if e.g.~$w$ only represents the probability of sending at most 1 photon in a single round.

The critical benefit of using these bounds instead of~\eqref{eq:tovN}--\eqref{eq:tovN2} directly is that the correction terms have picked up a prefactor of $w$, which substantially reduces their magnitude if $w$ is small. In fact, a heuristic scaling analysis of $\alpha$ informally suggests that it might cause the finite-size correction terms on the final {\Renyi} entropy bounds to scale as $O(\sqrt{wn})$ instead of $O(\sqrt{n})$ (putting aside the complication that $w$ might potentially be an optimization variable rather than a constant, as well as some subtleties regarding the scaling of a ``penalty function'' term in the~\cite{arx_AT25} results). This change is especially dramatic if for instance $w$ is the probability of the round being detected (or the even more restrictive event we defined above in~\eqref{eq:exampleevent}), because then $wn$ roughly corresponds to the number of detected rounds, which can be orders of magnitude smaller than $n$ at high loss. 

This comes at a small price of having an overall prefactor of $1-\Delta_\alpha$ that reduces the values slightly. However, we can roughly quantify the maximum possible effect of this prefactor with a series of crude bounds, focusing for instance on the bound~\eqref{eq:newtovN2}:
\begin{align}
& \left(1 - \frac{\alpha-1}{\alpha} \frac{\ln2}{2} \log\dim(S) \right) \bigg(wH(S|IE)_{\nu'} - w(\alpha-1)\frac{\ln 2}{2}\log^2\left(1+2\dim(S)\right)  -  O((\alpha-1)^2) \bigg) \nonumber\\
=& wH(S|IE)_{\nu'} - wH(S|IE)_{\nu'} \frac{\alpha-1}{\alpha} \frac{\ln2}{2} \log\dim(S) - w(\alpha-1)\frac{\ln 2}{2}\log^2\left(1+2\dim(S)\right)  -  O((\alpha-1)^2) \nonumber\\
\geq& wH(S|IE)_{\nu'} - w(\alpha-1) \frac{\ln2}{2} \log^2\dim(S)  - w(\alpha-1)\frac{\ln 2}{2}\log^2\left(1+2\dim(S)\right)  -  O((\alpha-1)^2) \nonumber\\
\geq& wH(S|IE)_{\nu'} - w(\alpha-1)(\ln 2)\log^2\left(1+2\dim(S)\right)  -  O((\alpha-1)^2).
\end{align}
From this we immediately see that the new bound~\eqref{eq:newtovN2} should be better than its previous counterpart~\eqref{eq:tovN2} whenever $w\leq1/2$ and $(\alpha-1)^2$ is small; in fact, since the above bounds were somewhat crude, it is likely to be better over a larger range than that. (On the other hand, it does highlight that if for instance we simply consider $w$ to be the probability of a generation round, i.e.~$w=1-\gamma$, then this approach may not give an improvement, since $\gamma$ is often small.)

Finally, we briefly comment on the other bound~\eqref{eq:approxlin2} in Lemma~\ref{lemma:approxlin}. Observe that it is slightly sharper than~\eqref{eq:approxlin}, in that it has an additional prefactor of $1-w$ on the $O(\alpha-1)$ term, hence making that term smaller. Unfortunately, it poses potential challenges for numerical work because it has nonlinear dependencies on $w$. This may not be a problem if we are considering examples in which $w$ is a fixed value (for instance, the probability of sending at most 1 photon), so that the prefactor in~\eqref{eq:approxlin2} is simply a constant, but in situations where $w$ is an optimization variable, it could be difficult to handle in this bound. Furthermore, recalling that we are most likely interested in scenarios where $w$ is small, the improvement is unlikely to be substantial. Hence this version may have somewhat limited use.

We wrap up by presenting the proof of Lemma~\ref{lemma:approxlin}. The informal intuition behind the proof is that since the {\Renyi} entropies converge to von Neumann entropy as $\alpha\to1$, we might expect to get results analogous to the von Neumann entropy case via Taylor expansions about $\alpha=1$.

\begin{proof}
Extend $\nu_{CQ}$ onto a classical register $Z$ via $\nu_{CQZ} = w \nu'_{CQ} \otimes \pure{0}_Z + (1-w) \nu''_{CQ} \otimes \pure{1}_Z$. 
We first prove the bound on $H_\alpha(C|Q)_{\nu}$ in~\eqref{eq:approxlin} (we consider this case first because the intermediate formulas are slightly simpler than the $H^\uparrow_\alpha(C|Q)_{\nu}$ case).
Specifically, noting that $C$ is classical so its entropies are non-negative, we find the claimed result for all $\alpha\in(1,\infty)$:
\begin{align}
H_\alpha(C|Q)_{\nu} 
&\geq H_\alpha(C|QZ)_{\nu} \nonumber\\
&= \frac{1}{1-\alpha}\log\left(w 2^{(1-\alpha)H_\alpha(C|Q)_{\nu'}} + (1-w) 2^{(1-\alpha)H_\alpha(C|Q)_{\nu''}}\right)
\nonumber\\
&\geq \frac{1}{1-\alpha}\log\left(w 2^{(1-\alpha)H_\alpha(C|Q)_{\nu'}} + (1-w)\right) \text{ since } H_\alpha(C|Q)_{\nu''}\geq 0 \nonumber\\
&= -\frac{1}{(\ln2)(\alpha-1) }\ln\left(w e^{-(\ln2)(\alpha-1)  H_\alpha(C|Q)_{\nu'}} + (1-w)\right) \nonumber\\
&\geq wH_\alpha(C|Q)_{\nu'} \left(1 - \frac{(\ln2)(\alpha-1)  H_\alpha(C|Q)_{\nu'}}{2} \right) \;\text{since } (\ln2)(\alpha-1)  H_\alpha(C|Q)_{\nu'}\geq 0 \nonumber\\
&\geq wH_\alpha(C|Q)_{\nu'} \left(1 - (\alpha-1) \frac{\ln2}{2} \log\dim(S) \right) = wH_\alpha(C|Q)_{\nu'} \left(1 - \alpha\Delta_\alpha \right), 
\end{align}
where in the second line we use~\cite[Proposition~5.1]{Tom16}, in the third line we use the fact that the expression is nondecreasing in $H_\alpha(C|Q)_{\nu''}$, and in the fifth line we apply the following inequality that holds for all $w\in[0,1]$ and $x\in[0,\infty)$:
\begin{align}\label{eq:taylor}
-\ln\left(w e^{-x} + (1-w)\right) 
\geq -\ln\left(w \left(-x + \frac{x^2}{2}\right) + 1\right) 
\geq w \left(x - \frac{x^2}{2}\right)
= wx \left(1 - \frac{x}{2}\right),
\end{align}
where the first inequality uses the bound $e^{-x} \leq 1-x+x^2/2$ for $x\in[0,\infty)$ (and implicitly the fact that the bounding quantity $1-x+x^2/2$ is still non-negative so the logarithm value is still well-defined), and the second inequality is the generic bound $\ln(1-x) \leq -x$ for all $x$ in the domain, i.e.~$x\in(-\infty,1)$.  (These bounds on the exponential and logarithm functions can be obtained using Taylor's theorem with mean-value form of the remainder term, noting that the sign of the remainder term is fixed in the regimes we consider; alternatively, the bound on the logarithm function can be obtained directly from the fact that it is concave.)

The $H^\uparrow_\alpha$ bound in~\eqref{eq:approxlin} follows from similar calculations, except we start with a slightly different relation for $H^\uparrow_\alpha$ in~\cite[Proposition~5.1]{Tom16}, 
\begin{align}
H^\uparrow_\alpha(C|QZ)_{\nu} = \frac{\alpha}{1-\alpha}\log\left(w 2^{\frac{1-\alpha}{\alpha}H^\uparrow_\alpha(C|Q)_{\nu'}} + (1-w) 2^{\frac{1-\alpha}{\alpha}H^\uparrow_\alpha(C|Q)_{\nu''}}\right), 
\end{align}
then carry out the remaining calculations the same way.

To prove the remaining bounds, we use  higher-order Taylor expansion bounds to sharpen the inequality~\eqref{eq:taylor}. Specifically, using the same bound on the exponential function but instead combining it with $\ln(1-x) \leq -x-x^2/2$ for $x \in [0,1)$, we see that the following holds for all $w\in[0,1]$ and $x\in[0,2]$: 
\begin{align}
-\ln\left(w e^{-x} + (1-w)\right) 
&\geq -\ln\left(w \left(-x + \frac{x^2}{2} \right) + 1\right)  \nonumber\\
&\geq w \left(x - \frac{x^2}{2}\right) + \frac{w^2}{2}\left(x - \frac{x^2}{2}\right)^2 \text{ since } x\leq2 \implies w \left(x - \frac{x^2}{2}\right)\geq 0 \nonumber\\
&= wx - \left(\frac{w}{2} - \frac{w^2}{2} \right) x^2 - \frac{w^2}{2}x^3 + \frac{w^2}{8}x^4 \nonumber\\
&= wx \left(1 - \frac{1-w}{2}x - \frac{w}{2}x^2 + \frac{w}{8}x^3 \right) .
\end{align}
With this we conclude that via analogous calculations to above, 
the bounds in~\eqref{eq:approxlin2} hold. Note that in doing so we dropped the last term in the above bound, because it is a positive contribution for which there may not be any straightforward lower bounds other than the trivial bound of zero (unless we allow the final bound to have nonlinear dependence on the entropy term as well). 
We also highlight that the (somewhat loosely constructed) restriction on allowed $\alpha$ values is to ensure that $(\ln2)(\alpha-1)  H_\alpha(C|Q)_{\nu'}$ or $(\ln2)\frac{1-\alpha}{\alpha}H^\uparrow_\alpha(C|Q)_{\nu'}$ is upper bounded by
\begin{align}
(\ln2)(\alpha-1)  \log\dim(S) 
 = (\alpha-1) \ln\dim(S) 
\leq 2,
\end{align}
so we can validly apply the above bound.

One could also consider higher-order bounds on the exponential function instead, but since the next valid bound that holds for all $x\in[0,\infty)$ requires going to the next even-order expansion (i.e.~$e^{-x} \leq 1-x+x^2/2!-x^3/3!+x^4/4!$), the resulting bound is very lengthy and seems unlikely to be of much use (furthermore, at best it only sharpens the $O((\alpha-1)^2)$ corrections in the final results, which may already have been fairly negligible). Another option would be to directly analyze the Taylor expansion of the function $-\ln\left(w e^{-x} + (1-w)\right) $, taking care to ensure the remainder term has the correct sign in the required regimes, but again the improvements would only be on the $O((\alpha-1)^2)$ corrections.
\end{proof}

\end{document}
\typeout{get arXiv to do 4 passes: Label(s) may have changed. Rerun}